\documentclass[11pt,letterpaper,twoside]{article}
\usepackage{./fullpage}
\usepackage{graphicx}
\usepackage{epsf}
\usepackage{amsthm}
\usepackage{amsmath}
\usepackage{amssymb}
\usepackage{txfonts}
\usepackage{enumerate}
\usepackage{color}
\usepackage{./ucb-titlepage}
\usepackage{./fancyheadings}

\newcounter{ecount}
  {\end{list}}

  {\end{list}}

\makeatletter
\long\def\unmarkedfootnote#1{{\long\def\@makefntext##1{##1}\footnotetext{#1}}}
\makeatother

\title{Approximation Bounds for Interpolation and Normals on  \\
       Triangulated Surfaces and Manifolds}

\author{Marc Khoury and Jonathan Richard Shewchuk}

\trnumber{}

\support{Supported in part by
the National Science Foundation under Awards CCF-1423560 and CCF-1909204
and in part by an Alfred P.\ Sloan Research Fellowship.
The claims in this document are those of the authors.
They are not endorsed by the sponsors or the U.S.\ Government.
}


\abstract{
How good is a triangulation as an approximation of
a smooth curved surface or manifold?
We provide bounds on the {\em interpolation error},
the error in the position of the surface, and
the {\em normal error}, the error in the normal vectors of the surface,
as approximated by a piecewise linearly triangulated surface
whose vertices lie on the original, smooth surface.
The interpolation error is the distance from
an arbitrary point on the triangulation to
the nearest point on the original, smooth manifold, or vice versa.
The normal error is the angle separating
the vector (or space) normal to a triangle from
the vector (or space) normal to the smooth manifold
(measured at a suitable point near the triangle).
We also study the {\em normal variation}, the angle separating
the normal vectors (or normal spaces)
at two different points on a smooth manifold.
Our bounds apply to manifolds of any dimension embedded
in Euclidean spaces of any dimension, and
our interpolation error bounds apply to simplices of any dimension,
although our normal error bounds apply only to triangles.
These bounds are expressed in terms of the sizes of suitable medial balls
(the {\em empty ball size} or {\em local feature size} measured
at certain points on the manifold), and have applications in
Delaunay triangulation-based algorithms for
provably good surface reconstruction and provably good mesh generation.
Our bounds have better constants than the prior bounds we know of---and
for several results in higher dimensions,
our bounds are the first to give explicit constants.
}

\begin{document}

\maketitle
\cleardoublepage

\renewcommand{\thepage}{\roman{page}}

\tableofcontents

\setlength{\parskip}{\medskipamount}

\renewcommand{\thepage}{\arabic{page}}
\setcounter{page}{1}

\pagestyle{fancy}
\headheight 14.4pt
\renewcommand{\sectionmark}[1]{\markboth{}{#1}}
\renewcommand{\subsectionmark}[1]{}
\lhead[\rm\thepage]{}
\chead[{\sl Marc Khoury and Jonathan Richard Shewchuk}]{{\sl \rightmark}}
\rhead[]{\rm\thepage}
\cfoot[]{}
\thispagestyle{empty}

\newtheorem{theorem}{Theorem}
\newtheorem{lemma}[theorem]{Lemma}
\newtheorem{corollary}[theorem]{Corollary}

\newcommand{\R}{\mathbb{R}}
\newcommand{\del}{\mathrm{Del}\,}
\newcommand{\vor}{\mathrm{Vor}\,}
\newcommand{\conv}{\mathrm{conv}\,}
\newcommand{\aff}{\mathrm{aff}\,}
\newcommand{\proj}{\mathrm{proj}}
\newcommand{\Int}{\mathrm{int}\,}
\newcommand{\lfs}{\mathrm{lfs}}
\newcommand{\ebs}{\mathrm{ebs}}
\newcommand{\resvor}[2]{\mathrm{Vor}|_{#2}\,#1}
\newcommand{\resdel}[2]{\mathrm{Del}|_{#2}\,#1}
\newcommand{\fixdel}[2]{\mathrm{Del}^*|_{#2}\,#1}

\section{Introduction}

Triangulations of surfaces are used heavily in computer graphics,
visualization, and geometric modeling;
they also find applications in scientific computing.
Also useful are triangulations of manifolds in
spaces of dimension higher than three---for example, as
a tool for studying the topology of algebraic varieties.
A~surface triangulation (sometimes called a {\em surface mesh}) replaces
a curved surface with flat triangles---or in higher dimensions,
simplices---which are easy to process
and suitable for graphics rendering engines; but they introduce error.
How good is a triangulation as an approximation of a curved surface?

The two criteria most important in practice are the {\em interpolation error},
the error in the position of the surface, and
the {\em normal error}, the error in the normal vectors of the surface.
Let $\Sigma$ be a surface or manifold embedded in a Euclidean space $\R^d$, and
let $\Lambda$ be a piecewise linear surface or manifold formed by
a triangulation that approximates $\Sigma$.
The interpolation error can be quantified as the distance from
an arbitrary point on $\Lambda$ to the nearest point on $\Sigma$, or
vice versa.
The normal error can be quantified by choosing two nearby points
$x \in \Lambda$ and $y \in \Sigma$---a natural choice of $y$ is
the point on $\Sigma$ nearest $x$---and
measuring the angle separating the vector normal to $\Lambda$ at $x$ from
the vector normal to $\Sigma$ at $y$.
(The vector normal to $\Lambda$ is usually undefined if
$x$ lies on a boundary where simplices meet, but our results will treat
simplices individually rather than treat $\Lambda$ as a whole.)

Some notation:  we employ a correspondence between the two surfaces called
the {\em nearest-point map}\footnote{
We follow the convention of Cheng et al.~\cite{cheng12} and use
the Greek letter {\em nu}, which unfortunately is hard to distinguish from
the italic Roman letter $v$.
}
$\nu$, which maps a point $x \in \R^d$ to
the point $\nu(x)$ nearest $x$ on $\Sigma$ (if that point is unique).
We will frequently use the abbreviation $\tilde{x}$ to denote $\nu(x)$.
Given two points $p, q \in \R^d$,
$pq$ denotes a line segment with endpoints $p$ and $q$, and
$|pq|$ denotes its Euclidean length $\|p - q\|_2$.
For a point $p$ on a surface $\Sigma \subset \R^3$,
$n_p$ denotes a vector normal to $\Sigma$ at $p$
(whose magnitude is irrelevant).
For a triangle $\tau \subset \R^3$, $n_\tau$ denotes a vector normal to $\tau$.
Let $\angle (n_\tau, n_p)$ denote the angle separating $n_\tau$ from $n_p$.
In higher-dimensional Euclidean spaces, the normal vectors may be replaced by
normal subspaces; see Section~\ref{tour}.

The goal of this paper is to provide strong bounds on
the interpolation errors for simplices (of any dimension) and
the normal errors for triangles,
based on assumptions about the sizes of medial balls
(defined in Section~\ref{tour}).
Specifically, given a simplex $\tau$ whose vertices lie on $\Sigma$ and
a point $x \in \tau$, we bound the distance $|x\tilde{x}|$ and,
if $\tau$ is a triangle, we bound the angle $\angle (n_\tau, n_{\tilde{x}})$.
Besides the interpolation and normal errors,
we also study the {\em normal variation}, the angle separating
the normal vectors (or normal spaces) at two different points on $\Sigma$.
(We need to understand the normal variation to study the normal error;
it is also used to prove that certain triangulations are homeomorphic to
a surface~\cite{cheng12,dey07}.)
Bounds on all three of these quantities---the interpolation error,
the normal error, and the normal variation---have been derived in prior works
\cite{amenta99b,amenta02,amenta07,cheng05,cheng12,dey07}
and form a foundation for
the correctness and accuracy of many algorithms in surface reconstruction
\cite{amenta99b,amenta02,amenta01,amenta07,boissonnat10,boissonnat14,
cheng05,dey08b,dey07,khoury16}
and mesh generation
\cite{boissonnat06,boissonnat05,cheng10,cheng12,dey08,oudot05,rineau07}
based on Delaunay triangulations.
Our notably improved bounds directly imply improved {\em sampling bounds} for
all of those algorithms.
By ``sampling bounds,'' we mean estimates of
how densely points must be sampled on a surface to guarantee that
the reconstructed surface or the surface mesh has
a good approximation accuracy and the correct topology.

A second goal of this paper is to generalize our bounds to
manifolds in higher dimensions.
Our bound on the interpolation error applies to a simplex of any dimension
with its vertices on a manifold of any dimension in a space of any dimension.
Our bounds on the normal error apply only to triangles, albeit
on a manifold of any dimension (greater than $1$) in a space of any dimension.
(We would like to study normal errors for simplices of higher dimension, but
the interaction between the shape of, say, a tetrahedron in $\R^4$ and
the stability of its normal space is complicated.
It deserves more study, but not in this paper.)

Our bounds on the normal variation also apply in higher dimensions, but
with a twist.
The {\em codimension} of a $k$-manifold $\Sigma \subset \R^d$ is $d - k$.
We have two {\em normal variation lemmas} (Section~\ref{variation}):
one for codimension~$1$, which bounds an angle
$\angle (n_p, n_q) \in [0^\circ, 180^\circ]$ between two normal vectors, and
one for higher codimensions, which bounds an angle
$\angle (N_p\Sigma, N_q\Sigma) \in [0^\circ, 90^\circ]$
between two normal spaces (see Section~\ref{tour} for
definitions of normal spaces and the angles between them).
The reason for two separate lemmas is that
the codimension~$1$ bound is stronger;
codimension~$2$ introduces configurations that weaken the bound and
cannot occur in codimension~$1$.
As a consequence,
some of our bounds on the normal errors also depend on the codimension.

Our bound on the interpolation error improves a prior bound
by a factor of about $30$ (see Section~\ref{tour}), and
one of our bounds on the normal error improves a prior bound
by a factor of about $1.9$ (see Section~\ref{tnls}).
Even small constant-factor improvements in the bounds are valuable;
for example, the number of triangles necessary for a surface mesh to guarantee
a specified accuracy in the normals is reduced by a factor of $1.9^2 = 3.61$,
helping to substantially speed up the application using the mesh.
In dimensions higher than three,
we are not aware of prior bounds with explicitly stated constants, but
there are asymptotic results~\cite{cheng05};
part of our contributions is to give strong explicit bounds.
Our bound on the interpolation error is {\em sharp},
meaning that it cannot be improved (without making additional assumptions).
(We use {\em sharp} to mean that not even the constants can be improved,
as opposed to {\em tight}, which is sometimes used in an asymptotic sense.)
We conjecture that our bound on the normal variation in codimension~$1$ is
sharp too.

The bounds help to clarify the relationship between approximation accuracy,
the sizes and shapes of the simplices in a surface mesh, and
the geometry of the surface itself.
Reducing the sizes of the simplices tends to reduce both the interpolation
and normal errors;
unsurprisingly, finer meshes offer better approximations than coarser ones.
The interpolation errors on a simplex scale quadratically with
the size of the simplex.
This is good news:
shrinking the simplices reduces the interpolation error quickly.
The normal errors scale linearly (not quadratically) with
the size of the simplex.
Roughly speaking, both types of error scale linearly with
the curvature of the manifold, measured at a selected point; more precisely,
they scale inversely with the radii of selected medial balls
(defined in Section~\ref{tour}), which we use to impose appropriate bounds
on both the curvature and the proximity of different parts of a manifold.
Therefore, portions of a manifold with greater curvature require
smaller simplices.

Interpolation errors are largely insensitive to the shape of a simplex.
Our bound on the interpolation error $|x\tilde{x}|$ is
proportional to the square of the {\em min-containment radius} of
the simplex containing $x$---the radius of its smallest enclosing ball
(see Section~\ref{prox}).
As this bound is sharp,
the min-containment radius is exactly the right measure to quantify
the effects of a simplex's size and shape on the interpolation error.

By contrast, normal errors are very sensitive to the shape of a simplex.
Skinny simplices underperform simplices that are close to equilateral, and
really skinny simplices can yield catastrophically wrong normals.
As a rough approximation, the worst-case normal error on a triangle is
linearly proportional to the triangle's circumradius,
defined in Section~\ref{tour}.
(See Sections~\ref{tnls} and~\ref{etnls} and
Amenta, Choi, Dey, and Leekha~\cite{amenta02}).
For triangles with a fixed longest edge length, the worst normal errors are
suffered by triangles with angles close to $180^\circ$, because
the circumradius approaches infinity as the largest angle approaches
$180^\circ$.
We give several bounds on the normal error for a triangle:
the simplest one depends on the triangle's circumradius, whereas
a stronger bound depends on one of triangle's angles as well,
giving us a more nuanced understanding of the relationship between
triangle shape and normal errors.

\section{A Tour of the Bounds}
\label{tour}

To create a surface mesh that meets specified constraints on accuracy,
one must consider the geometry of $\Sigma$ and
the size and (sometimes) the shape of each simplex.
Our bounds use three parameters to measure a simplex $\tau$:
the min-containment radius of $\tau$ and, for triangles only,
the circumradius of $\tau$ and (optionally) one of $\tau$'s plane angles.

For a simplex $\tau \subset \R^d$, the {\em smallest enclosing ball} of $\tau$
(also known as the {\em min-containment ball}) is
the smallest closed $d$-dimensional ball $B_\tau \supseteq \tau$,
illustrated in Figure~\ref{radii}.
The {\em min-containment radius} of $\tau$ is
the radius of $\tau$'s smallest enclosing ball; we write it as $r$
(though sometimes in this paper,
$r$ will be the radius of any arbitrary enclosing ball).
The {\em diametric ball} of $\tau$ is the smallest closed $d$-ball $B$
such that all $\tau$'s vertices lie on $B$'s boundary,
also illustrated in Figure~\ref{radii}.
The {\em circumcenter} and {\em circumradius} of $\tau$ are
the center and radius of $\tau$'s diametric ball, respectively;
we write the circumradius as $R$.
For every simplex, $r \leq R$; but if $\tau$ is ``badly'' shaped,
$R$ can be arbitrary large compared to $r$.
(Recall that for a triangle, $R \rightarrow \infty$ as
the largest angle approaches $180^\circ$ and the longest edge remains fixed.)
A simplex $\tau$ always contains the center of its smallest enclosing ball, but
frequently not its circumcenter.
The center of $\tau$'s smallest enclosing ball is
the point on $\tau$ closest to $\tau$'s circumcenter.
(See Rajan~\cite[Lemma~3]{rajan91} for
an algebraic proof based on quadratic program duality, or
Shewchuk~\cite[Lemma~24]{shewchuk08} for a geometric proof.)
Hence, $r = R$ if and only if $\tau$ contains its circumcenter.

\begin{figure}
\centerline{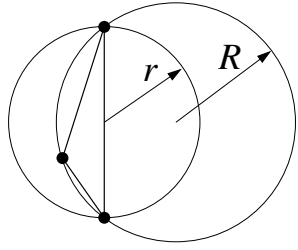}

\caption{\label{radii}  \protect\small \sf
The smallest enclosing ball of a triangle, with radius $r$, and
the triangle's diametric ball, with radius $R$.
}
\end{figure}

The {\em circumcircle} (circumscribing circle) of
a triangle $\tau \subset \R^d$ is
the unique circle that passes through all three vertices of $\tau$.
The circumcircle has the same center and radius $R$ as $\tau$'s diametric ball
(i.e., $\tau$'s circumcenter and circumradius).
A {\em plane angle} of a triangle $\tau$ is one of the usual three angles
we associate with a triangle, though $\tau$ might be embedded in
a high-dimensional space.
A triangle contains its circumcenter (and has $r = R$) if and only if
it has no angle greater than $90^\circ$.

There are two salient aspects to the geometry of $\Sigma$.
One is curvature:  a surface with greater curvature needs smaller triangles.
(Nonsmooth phenomena like sharp edges can make the triangle normals
inaccurate no matter how small the triangles are, and are best addressed by
matching the triangle edges to the surface discontinuities.
We don't address that problem here.)
A more subtle aspect is that a surface can ``double back'' and
come close to itself in Euclidean space:
for example, if a mesh of a hand has a triangle connecting
the pad of the thumb to a knuckle of the index finger,
the triangle misrepresents the surface badly.

The early literature on provably good surface reconstruction identified
the {\em medial axis}---more specifically, the sizes of medial balls---as
an effective way to gauge the triangle sizes required
as a consequence of both curvature and the proximity of parts like fingers.
Let $\Sigma$ be a bounded, smooth $k$-manifold embedded in~$\R^d$.
Let $B \subset \R^d$ be an open ball.
We call $B$ {\em surface-free} if $B \cap \Sigma = \emptyset$.
We say $B$ {\em touches} $\Sigma$ if $B \cap \Sigma = \emptyset$ but
$B$'s boundary intersects $\Sigma$; that is,
$B$ is surface-free but its closure is not.
In that case, $B$ is tangent to $\Sigma$ at the point(s) where they intersect.
There are two types of {\em medial ball};
both types are surface-free balls that touch $\Sigma$,
as illustrated in Figure~\ref{medial}.
Every surface-free ball whose boundary touches $\Sigma$ at more than one point
is a medial ball; most medial balls are of this first type.
Let $W \subset \R^d$ be the set containing the center of every
medial ball of this first type; these are the points $w \in W$ where
the nearest-point map $\nu(w)$ is not uniquely defined.
(Recall that $\nu$ maps a point $x \in \R^3$ to
the point $\tilde{x} = \nu(x)$ nearest $x$ on $\Sigma$.)
The {\em medial axis} $M \in \R^d$ is the closure of $W$,
as illustrated.
Each point added to $M$ by taking the closure is the center of
a medial ball of the second type, which touches $\Sigma$ at just one point.

\begin{figure}
\centerline{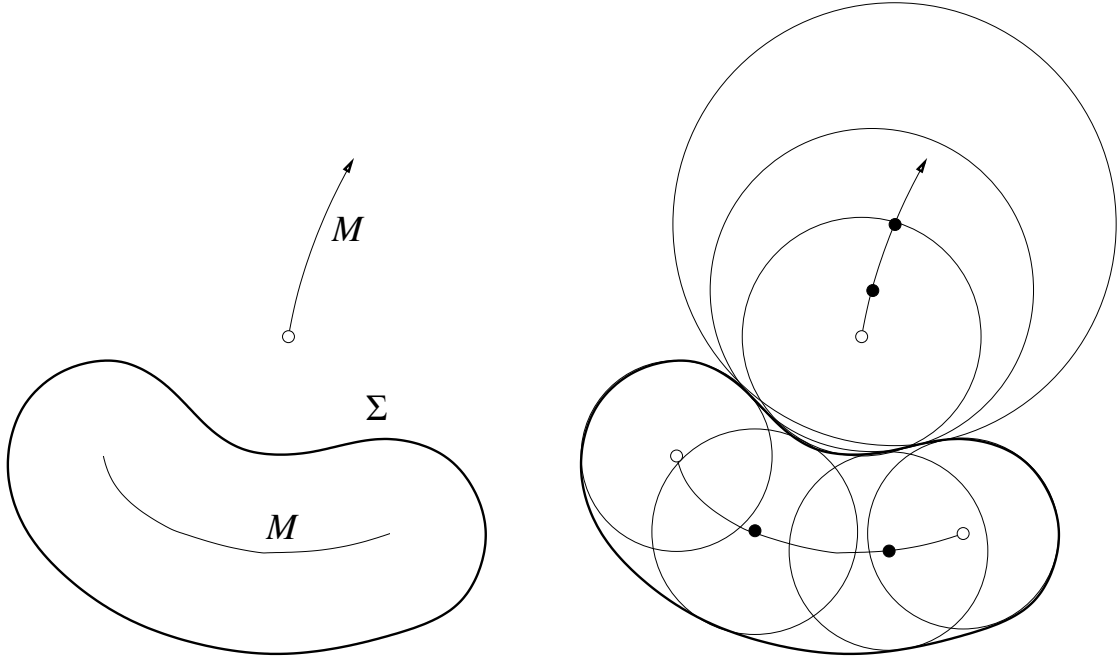}

\caption{\label{medial}  \protect\small \sf
Left:  A $1$-manifold $\Sigma$ and its medial axis $M$.
Right:  Some of the medial balls that define $M$.
Those with black centers are medial balls of the first type;
those with white centers are of the second type.
}
\end{figure}

We will often refer to the medial balls tangent to $\Sigma$ at a point
$p \in \Sigma$.
In codimension~$1$, there are typically two such balls
(but sometimes just one),
one enclosed by $\Sigma$ and (optionally) one outside $\Sigma$.
In higher codimensions, there are infinitely many.
All their centers lie in the normal space $N_p\Sigma$.
A useful construction we will use later is
to choose a point $q \in N_p\Sigma \setminus \{ p \}$ and
imagine an open ball tangent to $\Sigma$ at $p$ whose radius is initially zero;
then the ball grows so that its center moves along the ray $\vec{pq}$ while
its boundary remains touching $p$.
Typically, at some point the ball will not be able to grow further without
intersecting $\Sigma$.
At the last instant when the ball is still surface-free,
it is a medial ball, and its center is a point in the medial axis $M$.
Typically the ball cannot grow further because it touches
a second point on $\Sigma$ (producing a medial ball of the first type), but
sometimes it is constrained solely by the curvature of $\Sigma$ at $p$ itself
(producing a medial ball of the second type).
In some cases when $p$ lies on the boundary of the convex hull of $\Sigma$,
the ball can grow to infinite radius and
degenerate into an open halfspace while remaining surface-free.
It is occasionally useful to refer to such a degenerate medial ball,
although it does not contribute a point to $M$.

\begin{figure}
\centerline{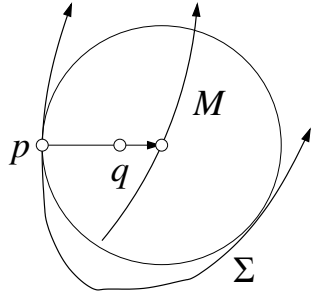}

\caption{\label{tanball}  \protect\small \sf
The medial ball tangent to $\Sigma$ at $p$ whose center lies on
the ray $\vec{pq}$.
}
\end{figure}

For any $p \in \Sigma$, the {\em empty ball size} $\ebs(p)$ is the radius of
the smallest medial ball tangent to $\Sigma$ at $p$.
The {\em local feature size} $\lfs(p)$ is the distance from $p$ to
the medial axis (i.e., from $p$ to the nearest point on $M$).
Formally,
\[
\lfs(p) = \min_{m \in M} |pm|;
\hspace{.5in}
\ebs(p) = \min_{m \in M \cap N_p\Sigma} |pm|.
\]
This definition makes clear that $\lfs(p) \leq \ebs(p)$.
Both measures simultaneously constrain the curvature of $\Sigma$ at $p$
(the principle curvatures cannot exceed $1 / \ebs(p)$) and
the proximity of other ``parts'' of the manifold
(recall the example of fingers of a hand).
The empty ball size has the advantage that
it is more local in nature than the local feature size, so
bounds expressed in terms of $\ebs(p)$ are more generally applicable
(which is why we are introducing $\ebs$ here).
The local feature size $\lfs(p)$ constrains the curvature not only at $p$, but
also at nearby points, permitting the proof of stronger conclusions.
The local feature size is $1$-Lipschitz,
meaning that for all $p, q \in \Sigma$, $\lfs(p) \leq \lfs(q) + |pq|$;
whereas the empty ball size can vary rapidly over $\Sigma$.

One of the main contribution of the early literature on
provably good surface reconstruction was to recognize that
the local feature size (scaled down by a constant factor) is
a good guide to how closely points need to be spaced on $\Sigma$ to ensure
that surface reconstruction algorithms will produce a correct output that
approximates $\Sigma$ well~\cite{amenta99b,amenta98b}.
Subsequently, provably good surface mesh generation algorithms also adopted
these observations~\cite{boissonnat03,boissonnat05,cheng01}.

The interpolation and normal errors are (approximately)
inversely proportional to $\ebs(p)$ or $\lfs(p)$ for some relevant point $p$.
That is, the errors increase with a decreasing radius of curvature
(i.e., an increasing curvature).
If $\Sigma$ is not smooth, each point $p$ where $\Sigma$ is not smooth has
$\ebs(p) = \lfs(p) = 0$, and $p$ lies on the medial axis $M$.
Our bounds do not apply at such points (the bounds are infinite).
However, the bounds still apply at other points where $\ebs$ is positive.

Our first result is a Surface Interpolation Lemma (Section~\ref{prox}),
which holds for a $j$-simplex $\tau$ whose vertices lie on
a $k$-manifold $\Sigma \subset \R^d$
for any $j$, $k$, and $d$ (even if $j > k$, oddly).
Let $r$ be the min-containment radius of $\tau$.
Given any point $x \in \tau$ and the nearest point $\tilde{x} \in \Sigma$,
\begin{equation}
|x\tilde{x}| \leq \ebs(\tilde{x}) - \sqrt{\ebs(\tilde{x})^2 - r^2}.
\label{interpineq}
\end{equation}

This bound is somewhat opaque.
It grows as $r$ grows and shrinks as $\ebs(\tilde{x})$ grows,
contrary to what you might expect at a first glance.
For the sake of understanding its asymptotics,
we plot the bound in Figure~\ref{interpbound}
(as well as a more specific bound given in Lemma~\ref{interplemma}) and
look at its Taylor series around $r = 0$,
\begin{equation}
|x\tilde{x}| \leq \frac{r^2}{2 \, \ebs(\tilde{x})} +
                  \frac{r^4}{8 \, \ebs(\tilde{x})^3} +
                  \frac{r^6}{16 \, \ebs(\tilde{x})^5} +
                  {\cal O}\left( \frac{r^8}{\ebs(\tilde{x})^7} \right).
\label{interptaylor}
\end{equation}
The bound~(\ref{interpineq}) is
in the interval $[ r^2 / (2 \, \ebs(\tilde{x})), r^2 / \ebs(\tilde{x}) ]$
over its legal range $r \in [0, \ebs(\tilde{x}) ]$.
Hence, if we scale $\tau$ by a factor of $\alpha$,
we scale the interpolation error by approximately $\alpha^2$
(which is good news for achieving small errors).
The interpolation error shrinks inversely as
the empty ball size or local feature size grows.
Note that $\ebs(\tilde{x})$ can be replaced by $\lfs(\tilde{x})$, as
$\lfs(\tilde{x}) \leq \ebs(\tilde{x})$.

\begin{figure}
\centerline{
  \setlength{\unitlength}{2.8in}
  \begin{picture}(1,1)
  \put(0,0){\includegraphics[width=2.8in]{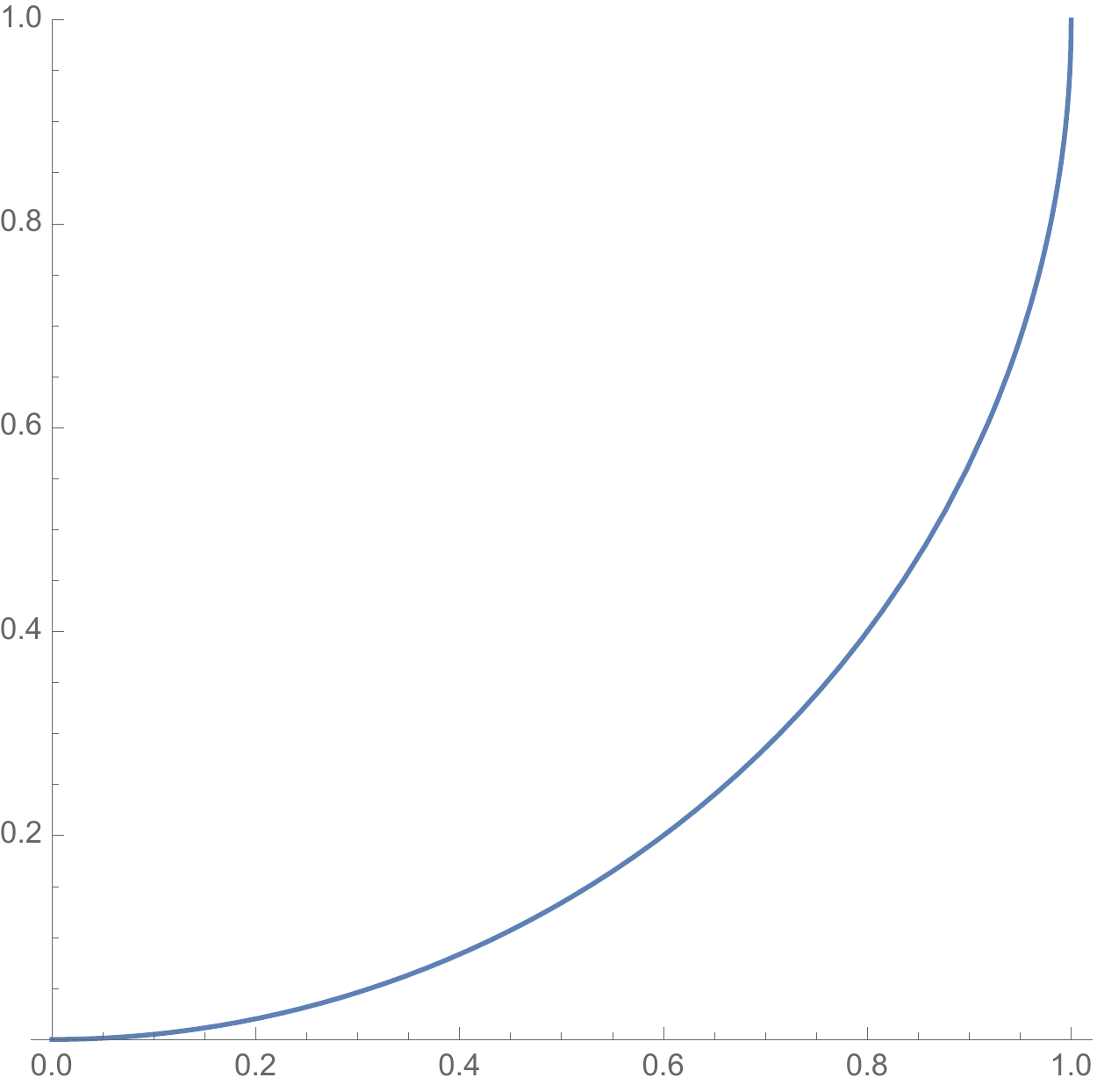}}
  \put(0.88,0){$r$}
  \put(0.07,0.93){bound on $|x\tilde{x}|$}
  \put(0.45,0.33){$1 - \sqrt{1 - r^2}$}
  \end{picture}
  \setlength{\unitlength}{2.8in}
  \begin{picture}(1,1)
  \put(0,0){\includegraphics[width=2.8in]{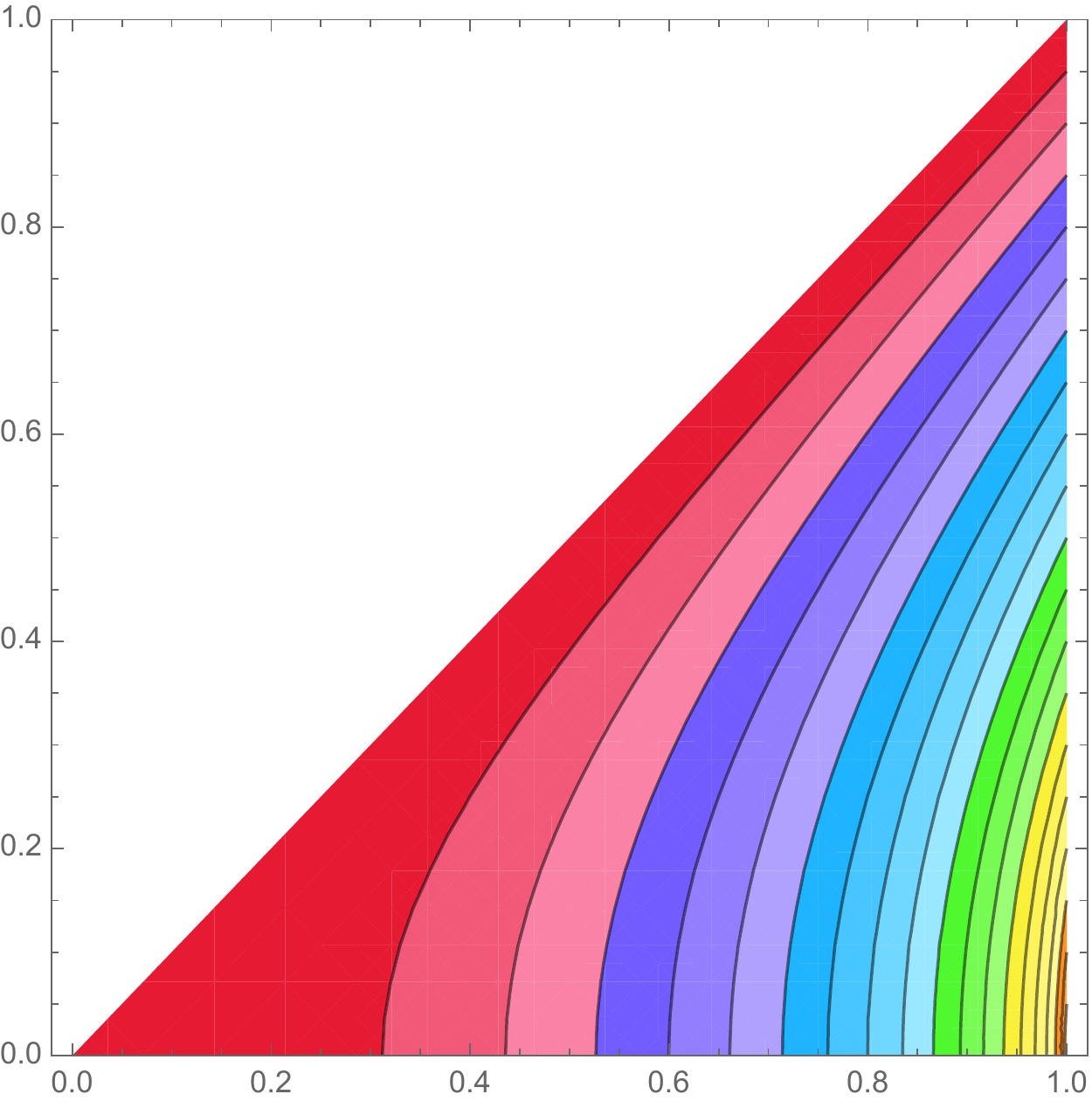}}
  \put(0.88,0){$r$}
  \put(0.08,0.89){$|xc|$}
  \end{picture}
  \includegraphics[width=0.55in]{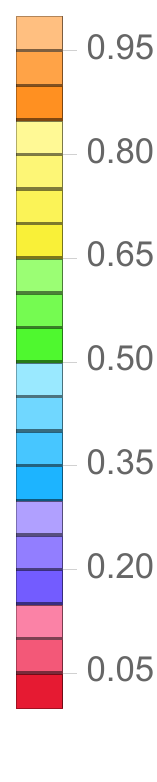}
}


\caption{\label{interpbound}  \protect\small \sf
Upper bounds for $|x\tilde{x}|$ where
$x$ is a point on a simplex $\tau$ whose vertices lie on a manifold $\Sigma$,
and $\tilde{x}$ is the point nearest $x$ on $\Sigma$.
We assume $\ebs(\tilde{x}) = 1$.
Left:  a bound on $|x\tilde{x}|$ for all $x \in \tau$,
as a function of the radius $r$ of $\tau$'s smallest enclosing ball.
Right:  isocontour plot of the bound on $|x\tilde{x}|$ as a function of $r$
(on the horizontal axis) and
the distance $|xc|$ from $x$ to the circumcenter $c$ of $\tau$
(on the vertical axis), showing how the interpolation error changes
from the circumcenter (bottom edge of plot) to
the vertices (diagonal of plot) of $\tau$.
The bottom edge of this isocontour plot corresponds to the case $x = c$ and
the left graph.
}
\end{figure}
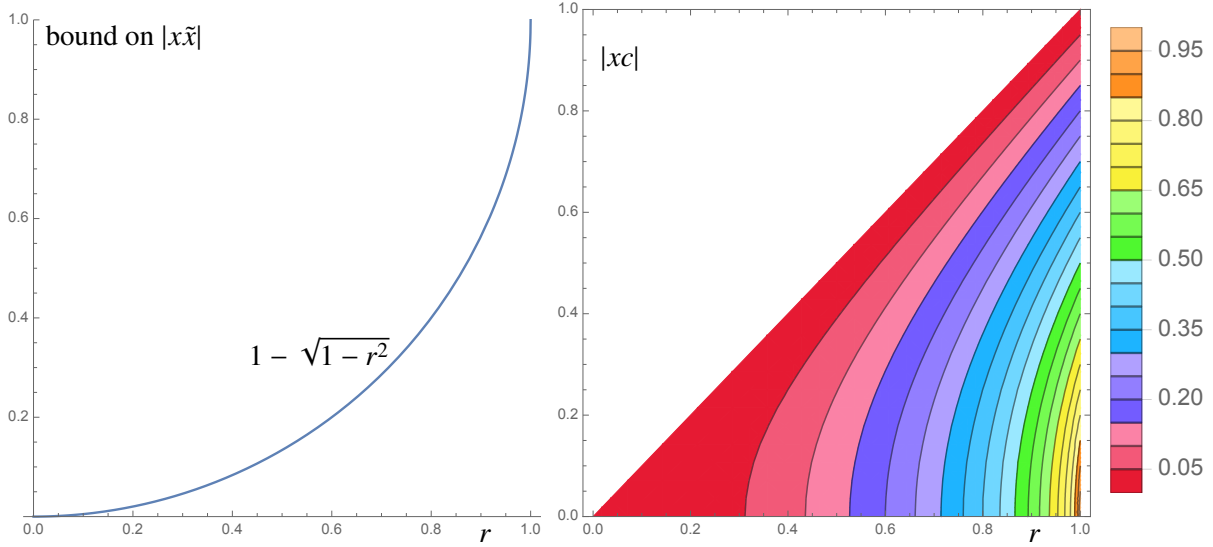

The bound~(\ref{interpineq}) is sharp, meaning that
under reasonably general conditions, there is a matching lower bound.
({\em Exactly} matching, not {\em asymptotically} matching.)
This implies that the min-containment radius is exactly
the right way to characterize the influence of $\tau$'s size on
the worst-case interpolation error.
The chief difficulty of the proof is showing that the bound holds
for the min-containment radius, and not only for the circumradius.

Compare the bound~(\ref{interptaylor}) with
the bound of $15 R^2 / \ebs(\tilde{x})$
implied by Cheng et al.~\cite[Proposition 13.19]{cheng12}.
We improve on that by a factor of up to $30$ for small $r$, or
by an arbitrarily large amount for triangles with $R \gg r$.

What if we reverse the question and ask to bound the distance from
a point $y \in \Sigma$ to the nearest point $\bar{y}$ on
a surface mesh $\Lambda$ whose vertices lie on $\Sigma$?
We assume that $\nu(\Lambda) = \Sigma$; that is,
for every point $y \in \Sigma$,
there is some point $x \in \Lambda$ such that $\tilde{x} = y$.
(This seems like a reasonable necessary criterion for $\Lambda$ to be
a ``good'' triangulation of $\Sigma$.)
The nearest-point relationship between $\Sigma$ and $\Lambda$ is not symmetric:
it is usually not true that $\bar{y} = x$.
Nevertheless, it is clearly true that $|y\bar{y}| \leq |xy|$.
Therefore, our upper bound~(\ref{interpineq}) on $|x\tilde{x}|$ is also
an upper bound on $|\tilde{x}\bar{\tilde{x}}|$.

Before we discuss normal errors,
we must discuss our Normal Variation Lemmas (Section~\ref{variation}).
The smoothness of a manifold $\Sigma$ implies that if two points are
close to each other, their normal spaces differ by only a small angle, and
likewise for their tangent spaces.
Given two points $p, q \in \Sigma$,
a {\em normal variation lemma} gives an upper bound on
the angle between their normal vectors (in codimension~$1$) or
their normal spaces (in codimension~$2$ or higher).

What are tangent spaces and normal spaces?
A {\em $k$-flat}, also known as an {\em $k$-dimensional affine subspace}, is
a $k$-dimensional space that is a subset of $\R^d$.
It is essentially the same as a $k$-dimensional subspace (from linear algebra),
but whereas a subspace must contain the origin, a flat has no such requirement.
Given a smooth $k$-manifold $\Sigma \subset \R^d$ and a point $p \in \Sigma$,
the {\em tangent space} $T_p\Sigma$ is the $k$-flat tangent to $p$ at $\Sigma$,
and the {\em normal space} $N_p\Sigma$ is the $(d - k)$-flat through $p$
that is entirely orthogonal (complementary) to $T_p\Sigma$; that is,
every line in $N_p\Sigma$ is perpendicular to every line in $T_p\Sigma$.

Recall that the codimension of $\Sigma$ is $d - k$.
In the special (but common) case of codimension~$1$,
a $(d - 1)$-manifold without boundary divides $\R^d$ into
an unbounded region we call ``outside'' and
one or more bounded regions we call ``inside.''
Hence for codimension~$1$ we use the convention that any normal vector $n_p$
is directed outward.
The normal space $N_p\Sigma$ is a line parallel to $n_p$, but
$n_p$ is directed and $N_p\Sigma$ is not.
In codimension~$2$ or higher, the normal space has dimension~$2$ or higher
(matching the codimension of $\Sigma$) and
$\Sigma$ might not even be orientable, so
we don't assign $N_p\Sigma$ a direction.

Let $F, G \subseteq \R^d$ be two flats, and suppose that
the dimension of $F$ is less than or equal to the dimension of $G$.
We define the angle separating $F$ from $G$ to be
\[
\angle (F, G) = \angle (G, F) =
\max_{\ell_F \subset F} \min_{\ell_G \subset G} \angle (\ell_F, \ell_G)
\]
where $\ell_F$ and $\ell_G$ are lines.
Note that if $F$ and $G$ are of different dimensions,
the ``$\max$'' must apply over the lower-dimensional flat and
the ``$\min$'' over the higher-dimensional flat.
This angle is always in the range $[0^\circ, 90^\circ]$;
we use angles greater than $90^\circ$ only for directed vectors.
If $F_\perp$ denotes a flat complementary to $F$,
it is well known that $\angle (F, G) = \angle (G_\perp, F_\perp)$;
hence, for two points $p, q \in \Sigma$,
$\angle (N_p\Sigma, N_q\Sigma) = \angle (T_p\Sigma, T_q\Sigma)$.
Note that there is more than one way to define ``angles between subspaces.'' 
The best-known way originates with an 1875 paper of Jordan~\cite{jordan75};
by this reckoning, one needs multiple angles to fully characterize
the angular relationships between two high-dimensional flats.
Our definition corresponds to the greatest of these angles
(including the $90^\circ$ angles, which are not included in
Jordan's {\em canonical angles}), so our upper bound holds for all the angles.

It is convenient to specify our bounds on
$\angle (N_p\Sigma, N_q\Sigma) = \angle (T_p\Sigma, T_q\Sigma)$
in terms of a parameter $\delta = |pq| / \lfs(p)$.
The worst-case value of $\angle (N_p\Sigma, N_q\Sigma)$ is
$\delta + \mathcal{O}(\delta^3)$ radians for small $\delta$.
Hence, the worst-case normal variation is approximately linear in $|pq|$ and
approximately inversely proportional to $\lfs(p)$.

We give two Normal Variation Lemmas that, collectively, apply to
smooth $k$-manifolds embedded in $\R^d$ for every $d$ and $k < d$.
They are stronger than the best prior bounds, especially for $d > 3$.
There are two separate lemmas because we obtain
a better bound for codimension one than for codimension two and higher.
Our main result in codimension $1$ is that for $\delta \leq 0.9717$,
$\angle (n_p, n_q) \leq \eta_1(\delta) \in [0^\circ, 180^\circ]$ where
\[
\eta_1(\delta) = \arccos \left( 1 - \frac{\delta^2}{2 \sqrt{1 - \delta^2}}
                                \right)
\approx
\delta + \frac{7}{24} \delta^3 + \frac{123}{640} \delta^5 +
\frac{1\mbox{,}083}{7\mbox{,}168} \delta^7 + O(\delta^9).
\]
Our main result for general codimensions is that for $\delta \leq 0.7861$,
$\angle (N_p\Sigma, N_q\Sigma) = \angle (T_p\Sigma, T_q\Sigma) \leq
\eta_2(\delta) \in [0^\circ, 90^\circ]$ where
\[
\eta_2(\delta) = \arccos
  \sqrt{1 - \frac{\delta^2}{\sqrt{1 - \delta^2}}}
\approx
\delta + \frac{5}{12} \delta^3 + \frac{57}{160} \delta^5 +
\frac{327}{896} \delta^7 + O(\delta^9).
\]
We conjecture that our bound for codimension~$1$ is sharp, meaning that
it cannot be improved without imposing additional restrictions.
Our bound for codimension~$2$ is not sharp and leaves room for improvement.
See Section~\ref{variation} for additional bounds (and plots thereof) that are
stronger when the distance from $q$ to $p$'s tangent plane is known.

Figure~\ref{nvlplots} compares our two bounds and two prior bounds for
surfaces in $\R^3$, both by Amenta and Dey~\cite{amenta07}.
The stronger prior bound is $\angle (n_p, n_q) \leq - \ln(1 - \delta)$ radians
for $\delta \leq 0.9567$.
(A derivation of both bounds can also be found in Cheng et al.~\cite{cheng12}.
Amenta and Bern~\cite{amenta99b} gave an early normal variation lemma with
a weaker bound, but the proof was erroneous.)
This bound fades to $90^\circ$ at $\delta \approx 0.7921$ and
to $180^\circ$ at $\delta \approx 0.9567$, whereas
our bound for codimension~$1$ fades to $90^\circ$ at $\delta \approx 0.9101$
and to $180^\circ$ at $\delta \approx 0.9717$.
Our bound for higher codimensions fades to $90^\circ$ at
$\delta \approx 0.7861$ and stops there (because
we do not assign directions to normal spaces of dimension~$2$ or higher).
Amenta and Dey~\cite{amenta07} also proved a bound of
$\delta / (1 - \delta)$ radians, which has become better known.
We include it in Figure~\ref{nvlplots} (in purple) to show how much is lost
by using the well-known bound instead of the stronger bounds.
The Amenta--Dey bounds are of the form
$\angle (N_p\Sigma, N_q\Sigma) \leq \delta + \mathcal{O}(\delta^2)$ radians,
whereas our bounds show that
$\angle (N_p\Sigma, N_q\Sigma) \leq \delta + \mathcal{O}(\delta^3)$ radians.

\begin{figure}
\centerline{
  \setlength{\unitlength}{4in}
  \begin{picture}(1,0.642)
  \put(0,0){\includegraphics[width=4in]{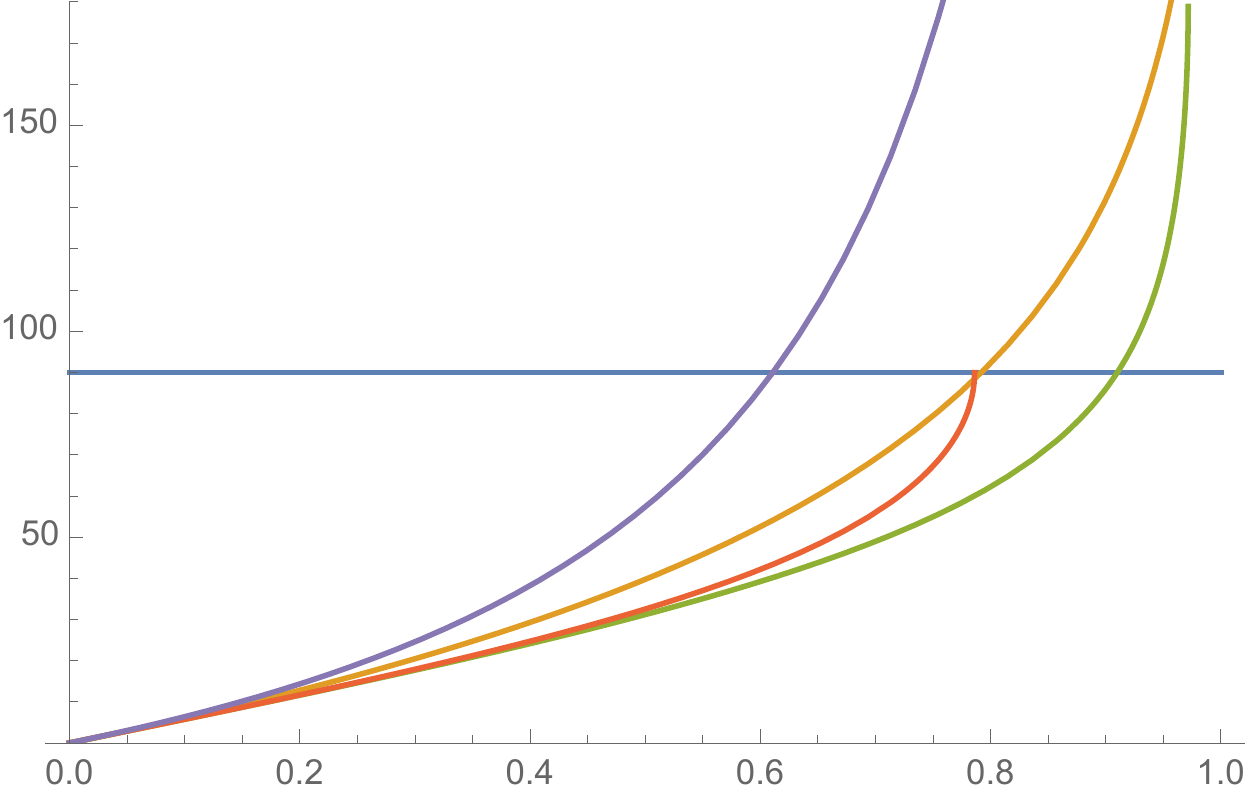}}
  \put(0.88,0){$\delta$}
  \put(0.07,0.61){bound on $\angle(N_p\Sigma, N_q\Sigma)$, degrees}
  \put(0.865,0.28){$\eta_1(\delta)$}
  \put(0.78,0.31){$\eta_2(\delta)$}
  \put(0.73,0.51){$- \ln (1 - \delta)$}
  \put(0.755,0.46){radians}
  \put(0.58,0.55){$\delta / (1 - \delta)$}
  \put(0.59,0.5){radians}
  \end{picture}
}


\caption{\label{nvlplots}  \protect\small \sf
Upper bounds in degrees for $\angle (N_p\Sigma, N_q\Sigma)$ as a function of
$\delta = |pq| / \lfs(p)$, provided by several normal variation lemmas.
The brown curve is the bound $- \ln(1 - \delta)$ radians proved by
Amenta and Dey~\cite{amenta07} for surfaces without boundary in $\R^3$.
The purple curve is the weaker but better-known bound
$\delta / (1 - \delta)$ radians, also by Amenta and Dey~\cite{amenta07}.
The green curve is our bound for codimension $1$---that is,
for $(d - 1)$-manifolds without boundary in $\R^d$.
The red curve is our bound for codimension~$2$ or greater---that is,
for $k$-manifolds without boundary in $\R^d$ with $d - k \geq 2$.
Bounds between $90^\circ$ and $180^\circ$ are meaningful for
manifolds without boundary in codimension~$1$.
The red curve stops at $90^\circ$ because
we do not assign directions to normal spaces of dimension~$2$ or higher.
}
\end{figure}
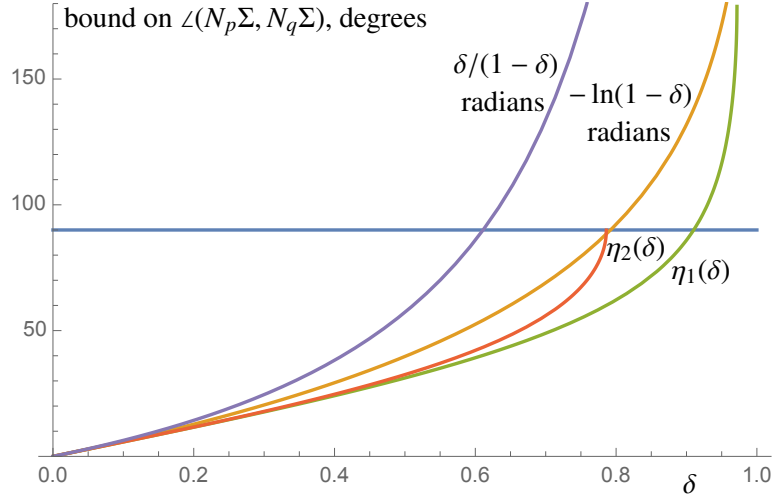

Cheng, Dey, and Ramos~\cite{cheng05} prove
a general-dimensional normal variation lemma for $k$-manifolds in~$\R^d$,
showing that in the worse case, $\angle (N_p\Sigma, N_q\Sigma)$ grows linearly
with $\delta$ for small $\delta$; but they express their bound in
an asymptotic form with an unspecified constant coefficient,
which makes a comparison with our bounds difficult.
We think it is a useful and practical contribution to provide
explicit numerical bounds $\eta_1(\delta)$ and $\eta_2(\delta)$ for $d > 3$.
Although our bound $\eta_2(\delta)$ is not sharp,
for $\delta \leq 0.7$ it is not much bigger than $\eta_1(\delta)$,
which we conjecture is a lower bound for all codimensions.

Finally, our results include several Triangle Normal Lemmas
(Sections~\ref{tnls} and~\ref{etnls}).
For a triangle $\tau$ whose vertices lie on a $k$-manifold $\Sigma$,
let $\nu(\tau)$ be the image of $\tau$ under the nearest-point map.
We derive bounds on how well $\tau$'s normal vector
locally approximates the vectors normal to $\Sigma$ on $\nu(\tau)$.
For a $j$-simplex $\tau \subset \R^d$,
its tangent space is its affine hull, a $j$-flat denoted $\aff \tau$.
For convenience, we define a particular normal space for simplices:
let $N_\tau$ denote the set of points in $\R^d$ that are
equidistant to all the vertices of $\tau$.
$N_\tau$ is a $(d - j)$-flat complementary to $\aff \tau$.
The intersection of $N_\tau$ and $\aff \tau$ is $\tau$'s circumcenter.

Our basic Triangle Normal Lemma applies only at the vertices of $\tau$.
Let $R$ be $\tau$'s circumradius.
Let $v$ be a vertex of $\tau$ and let $\phi$ be $\tau$'s plane angle at $v$.
Then
\begin{equation}
\angle(N_\tau, N_v\Sigma) = \angle(\aff{\tau}, T_v\Sigma) \leq
\arcsin{\left( \frac{R}{\ebs(v)} \max \left\{ \cot \frac{\phi}{2}, 1 \right\}
               \right)}.
\label{tnlbound2}
\end{equation}
Note that the argument $\cot \frac{\phi}{2}$ dominates if $\phi$ is acute and
the argument $1$ dominates if $\phi$ is obtuse.
If $v$ is the vertex at $\tau$'s largest plane angle (so $\phi \geq 60^\circ$),
then
\begin{equation}
\angle(N_\tau, N_v\Sigma) = \angle(\aff{\tau}, T_v\Sigma) \leq
\arcsin \frac{\sqrt{3}R}{\ebs(v)}.
\label{tnlbound1}
\end{equation}
Figure~\ref{tnlplots} plots both bounds,
(\ref{tnlbound1}) at left and~(\ref{tnlbound2}) at right.
Note that $\ebs(v)$ can be replaced by $\lfs(v)$.
It is interesting that the worst case preventing the bound~(\ref{tnlbound1})
from being better is incurred by an equilateral triangle
(rather than a triangle with a very large or small angle, as one might expect).

\begin{figure}
\centerline{
  \setlength{\unitlength}{2.8in}
  \begin{picture}(1,0.967)
  \put(0,0){\includegraphics[width=2.8in]{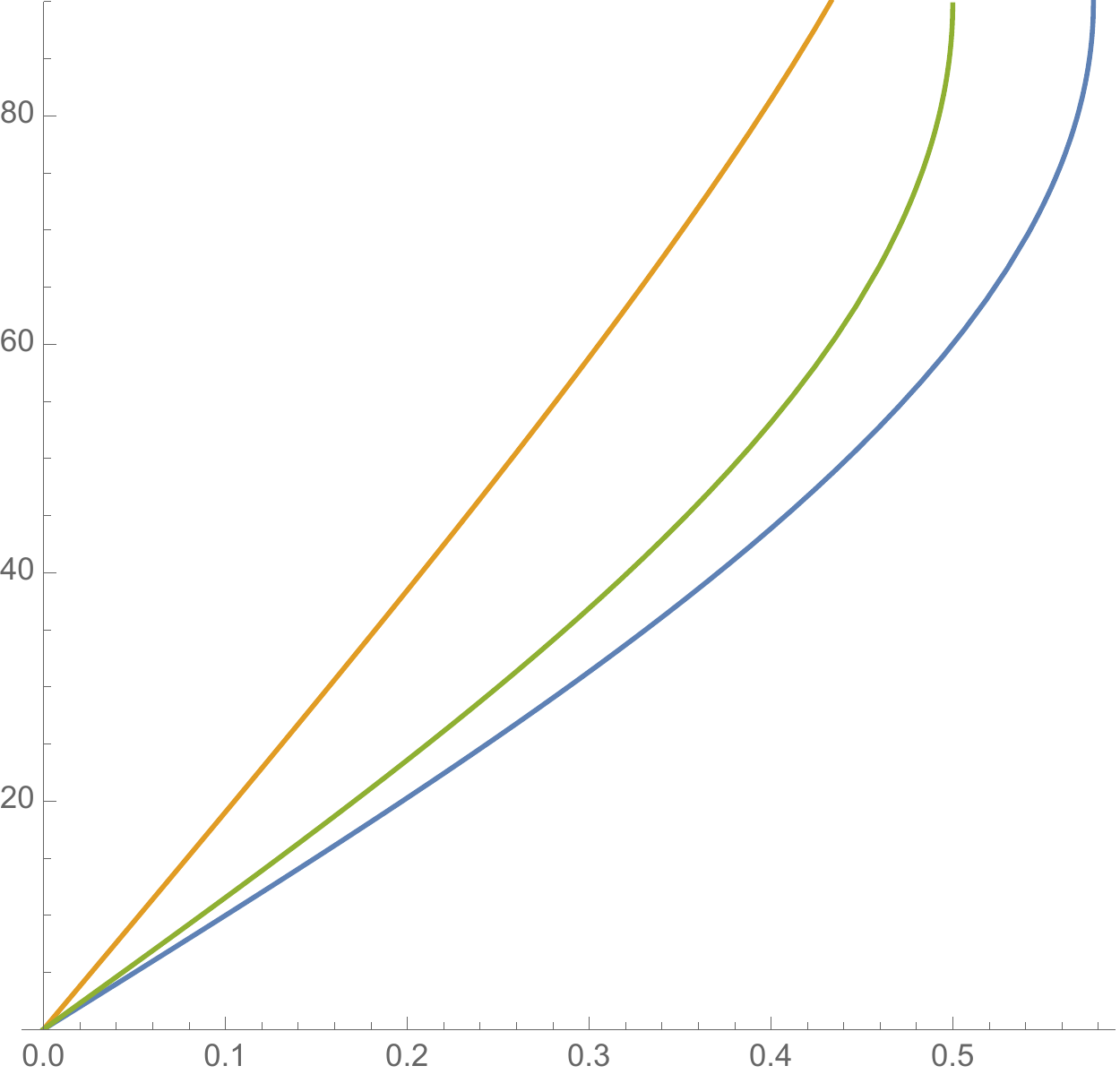}}
  \put(0.95,0){$R$}
  \put(0.07,0.91){bound on $\angle(N_\tau, N_v\Sigma)$}
  \put(0.5,0.3){$\arcsin (\sqrt{3} R)$}
  \put(0.18,0.6){Amenta et al.}
  \put(0.54,0.6){Cheng}
  \put(0.52,0.53){et al.}
  \end{picture}
  \begin{picture}(1,0.989)
  \put(0,0){\includegraphics[width=2.8in]{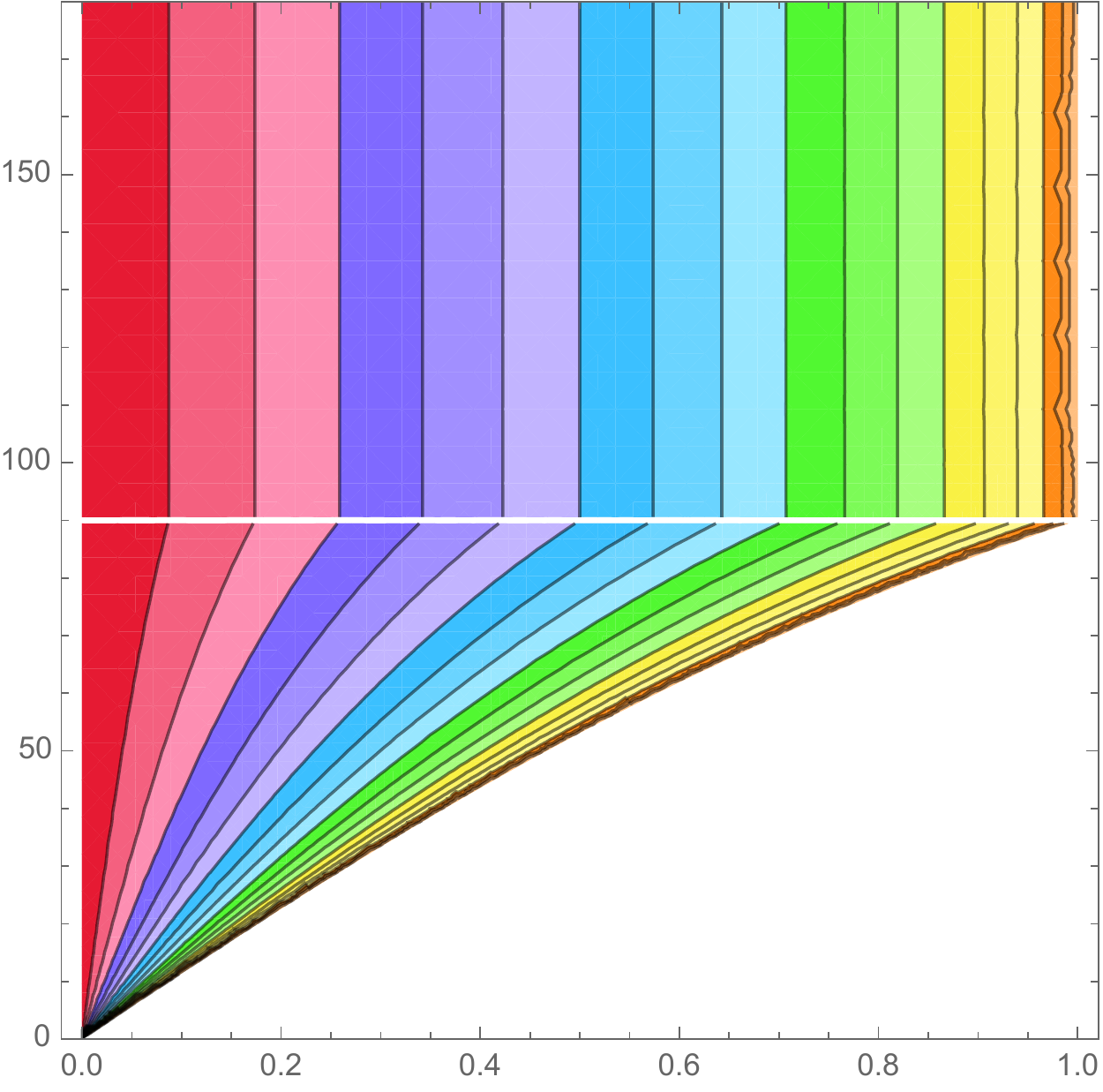}}
  \put(0.88,0){$R$}
  \put(0.01,0.92){$\phi$}
  \end{picture}
  \includegraphics[width=0.45in]{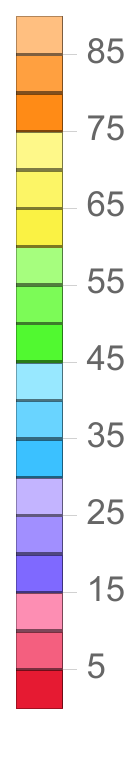}
}


\caption{\label{tnlplots}  \protect\small \sf
Upper bounds in degrees for
$\angle(N_\tau, N_v\Sigma) = \angle(\aff{\tau}, T_v\Sigma)$,
where $\tau$ is a triangle whose vertices lie on a manifold $\Sigma$ and
$v$ is a vertex of $\tau$.
We assume $\ebs(v) = 1$.
Left:  three bounds on $\angle(N_\tau, N_v\Sigma)$ for the case where
$v$ is the vertex at $\tau$'s largest plane angle
(or any angle $60^\circ$ or greater),
as a function of the circumradius $R$ of $\tau$.
The blue curve is our new bound~(\ref{tnlbound1}).
The green curve is the best (albeit little-known) prior bound we are aware of,
$\arcsin(2R)$, due to Cheng, Dey, Edelsbrunner, and Sullivan~\cite{cheng01}.
The brown curve is a much better-known prior bound, due to
Amenta, Choi, Dey, and Leekha~\cite{amenta02} (see Lemma~\ref{lem:otnl}).
Right:  isocontour plot of our bound~(\ref{tnlbound2}) as a function of
the circumradius $R$ (on the horizontal axis) and
the angle $\phi$ at the vertex $v$ (on the vertical axis).
For small $\phi$, the lemma does not provide a bound
(unless $R$ is very small), but see Section~\ref{etnls}.
}
\end{figure}
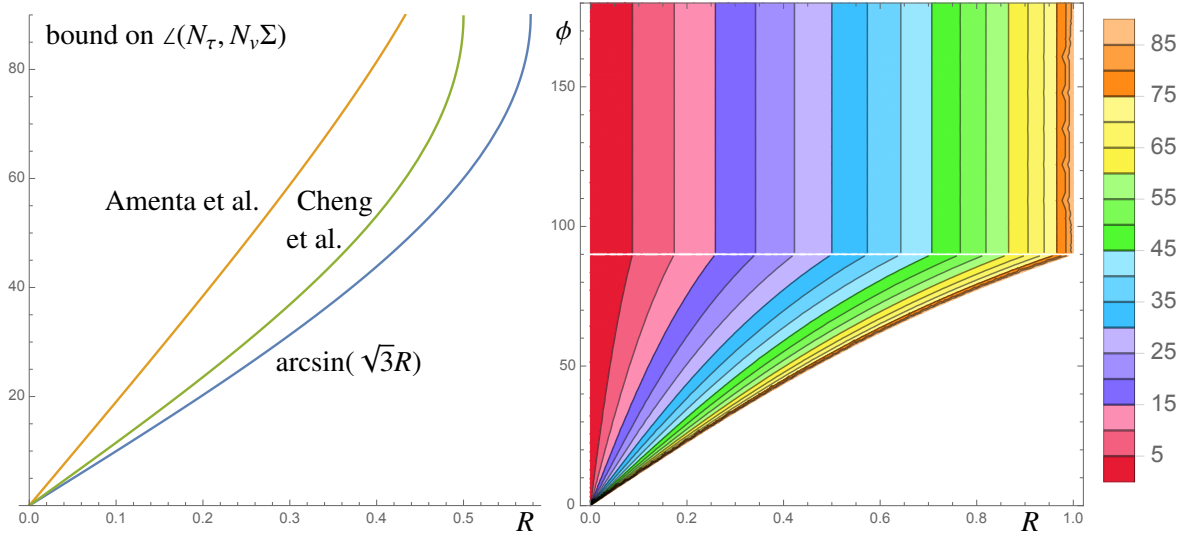

These bounds vary approximately linearly with the circumradius of $\tau$,
and inversely with the empty ball size or local feature size at $v$.
Whereas the interpolation error varies quadratically with
the radius of $\tau$'s smallest enclosing ball, and is therefore
very sensitive to $\tau$'s size but nearly insensitive to its shape,
the normal error varies (linearly) with $\tau$'s circumradius, which can be
much larger than $\tau$ if $\tau$ has a large angle (close to $180^\circ$).
It is well known that in surface meshes,
triangles with large angles are undesirable and sometimes even crippling to
applications, not because of problems with interpolation error, but
because of problems with very inaccurate normals.

Given a triangulation of $\Sigma$, one would like to have
a triangle normal lemma that applies to every point on $\Sigma$,
not just at the vertices.
Moreover, the Triangle Normal Lemma bounds are weak or nonexistent at
the vertices where the triangles have small plane angles.
Hence, we use the Normal Variation Lemmas to extend
the Triangle Normal Lemma bounds over the rest of $\nu(\tau)$---that is,
for every $x \in \tau$, we bound $\angle (N_\tau, N_{\tilde{x}}\Sigma)$.
Thus, a finely triangulated smooth manifold accurately approximates
the normal spaces of all the points on the manifold.
We call these results {\em extended triangle normal lemmas}.
Suppose that $R \leq \kappa \, \lfs(w)$ for every vertex $w$ of $\tau$.
Then for every point $x \in \tau$,
\[
\angle (N_\tau, N_{\tilde{x}}\Sigma) \leq
\max \left\{ \eta(\sqrt{2} \kappa) +
\arcsin \left( \kappa \cot \frac{\phi}{2} \right),
\eta(2 \kappa) +
\arcsin \left( \kappa \cot \left( 45^\circ - \frac{\phi}{4} \right) \right)
\right\}
\]
where $\eta(\delta) = \eta_1(\delta)$ in codimension~$1$, or
$\eta(\delta) = \eta_2(\delta)$ in codimension~$2$; and
$\phi$ is a ``proof parameter'' that can be set to
any angle in the range $(0^\circ, 60^\circ]$.
We recommend choosing $\phi = 49^\circ$ in codimension~$1$, and
$\phi = 48.5^\circ$ in higher codimensions.
Figure~\ref{etnlplot} graphs the bound for both cases.
We also give another version of this bound tailored for
restricted Delaunay triangles in an $\epsilon$-sample of $\Sigma$.
(See Section~\ref{etnls}.)

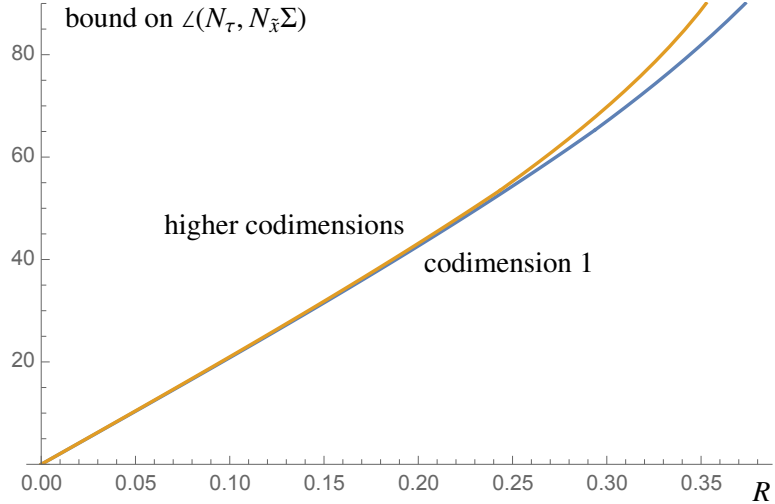
\begin{figure}
\centerline{
  \setlength{\unitlength}{4in}
  \begin{picture}(1,0.653)
  \put(0,0){\includegraphics[width=4in]{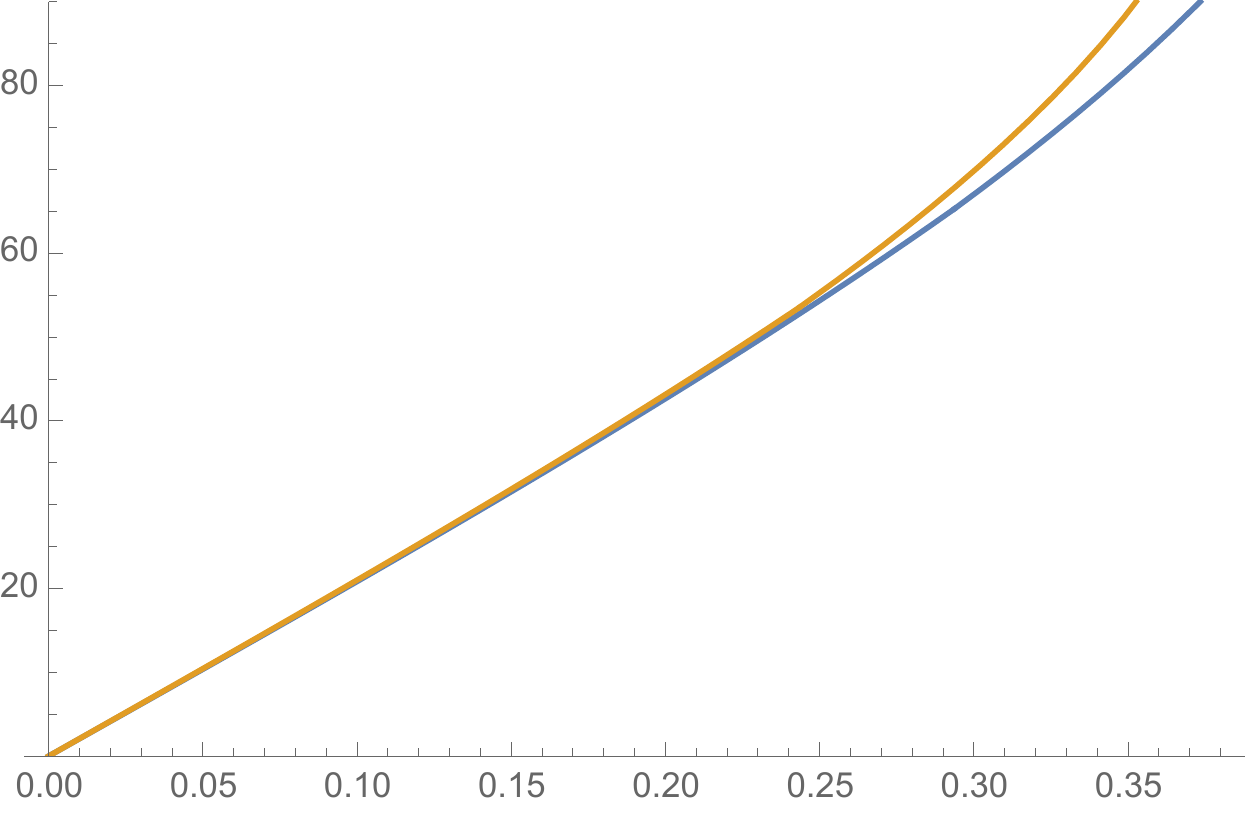}}
  \put(0.97,0){$R$}
  \put(0.07,0.62){bound on $\angle(N_\tau, N_{\tilde{x}}\Sigma)$}
  \put(0.54,0.3){codimension $1$}
  \put(0.2,0.35){higher codimensions}
  \end{picture}
}

%

\caption{\label{etnlplot}  \protect\small \sf
Upper bounds for $\angle(N_\tau, N_{\tilde{x}}\Sigma) =
\angle(\aff{\tau}, T_{\tilde{x}}\Sigma)$
as a function of the circumradius $R$ of $\tau$,
where $\tau$ is a triangle whose vertices lie on
a manifold $\Sigma$ and $x$ is any point on $\tau$.
We assume all three vertices $w$ of $\tau$ satisfy $\lfs(w) \geq 1$.
The blue curve is the upper bound in codimension $1$
(with the choice $\phi = 49^\circ$) and
the brown curve is the upper bound in higher codimensions
(with the choice $\phi = 48.5^\circ$),
for which the Normal Variation Lemma is weaker.
}
\end{figure}


Beyond the improved approximation bounds, we think that some of
the proof ideas in this paper are interesting in their own right.
Our proof of the Triangle Normal Lemma is strongly intuitive and
reveals a lot about {\em why} the bound is what it is.
Our proofs of the Normal Variation Lemmas exploit properties of medial balls
and medial-free balls in ways that allow us to obtain stronger bounds than
prior proofs, which were based on
integration of the curvature along a path on $\Sigma$.
These properties also find application in a forthcoming sequel paper that
improves the sampling bounds needed to guarantee that
a triangulation is homeomorphic to an underlying $2$-manifold.

Bounds on the interpolation and normal errors for surfaces have
much in common with analogous bounds for
piecewise linear interpolation over triangulations in the plane,
many of which were developed in an effort to analyze the finite element method
for solving partial differential equations~\cite{strang73}.
Consider a scalar field $f$ defined over a domain $\Omega \subset \R^d$, and
suppose that $f$'s directional second derivatives are, in all directions,
bounded so their magnitudes do not exceed some constant.
Let $g$ be an approximation of $f$ that is piecewise linear over $\Omega$,
with $g(v) = f(v)$ at every triangulation vertex $v$.
Waldron~\cite{waldron98} gives a sharp bound on
the pointwise interpolation error $f(p) - g(p)$ at
an arbitrary point $p \in \Omega$.
His bound is akin to our bound~(\ref{interpineq}) on $|x\tilde{x}|$---it is
proportional to the square of the min-containment radius of
the simplex that contains~$p$, it is sharp, and it holds in any dimension---but
the precise bound, the context, and the correctness proof are different.

In many applications (such as mechanical modeling of stress),
the interpolation error in the gradient, $\|\nabla f(p) - \nabla g(p)\|$, is
even more important than $|f(p) - g(p)|$.
The pointwise gradient interpolation error $\|\nabla f(p) - \nabla g(p)\|$ at
the worst point $p$ in a simplex scales linearly with the size of the simplex,
and is very sensitive to the shape of the simplex.
An early analysis by Bramble and Zl\'{a}mal~\cite{bramble70} for $\R^2$
seemed to implicate triangles with small angles (near $0^\circ$), but
a famous paper by Babu\v{s}ka and Aziz~\cite{babuska76} vindicated small angles
and placed the blame on large angles (near $180^\circ$).
A triangle's circumradius alone suffices to produce a reasonable rough bound
on the pointwise gradient interpolation error over the triangle, but
a stronger bound can be obtained by taking into account additional information
about the triangle's shape~\cite{shewchuk02b}.
Similarly, in this paper we show that a triangle's circumradius alone suffices
to produce a reasonable rough bound~(\ref{tnlbound1}) on the normal error, but
a stronger bound~(\ref{tnlbound2}) can be obtained by taking into account
more information about shape.

\section{A Surface Interpolation Lemma}
\label{prox}

Recall that, given a simplex $\tau$ whose vertices lie on a manifold $\Sigma$,
we desire an upper bound on the {\em interpolation error} $|x\tilde{x}|$
for a point $x \in \tau$.
To develop intuition, consider the lower bound first.
Suppose $\Sigma$ is a $k$-sphere embedded in $\R^d$, with radius $L$ and
centered at the origin, as illustrated in Figure~\ref{lower}.
Then the medial axis $M$ is a $(d - k - 1)$-flat passing through the origin;
for our purposes, the origin is the only medial axis point relevant here.
Let $\tau$ be a $j$-simplex whose vertices all lie on $\Sigma$.
Let $B_\tau \supset \tau$ be $\tau$'s diametric ball
(the smallest closed $d$-ball whose boundary
passes through all of $\tau$'s vertices).
Let $c$ and $R$ be the center and radius of $B_\tau$, respectively.
Observe that $\tau$'s circumcircle is a cross section of $\Sigma$.

\begin{figure}
\centerline{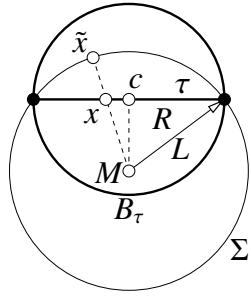}

\caption{\label{lower}  \protect\small \sf
A worst-case example for interpolation error.
}
\end{figure}

Consider a point $x \in \tau$ and
the point $\tilde{x}$ nearest $x$ on $\Sigma$.
As $x$ lies on the line segment connecting $\tilde{x}$ to
the center of $\Sigma$, and the length of that line segment is $L$,
it follows that the distance from $x$ to $\tilde{x}$---the interpolation error
that we wish to study---is $|x\tilde{x}| = L - \|x\|$.
Observe that the line segment connecting $c$ (the center of $B_\tau$) to
the origin (the center of $\Sigma$) is
perpendicular to the $j$-flat in which $\tau$ lies.
By Pythagoras' Theorem, $L^2 = \|c\|^2 + R^2$ and
$\|x\|^2 = \|c\|^2 + |xc|^2 = L^2 - R^2 + |xc|^2$, so
\begin{equation}
|x\tilde{x}| = L - \sqrt{L^2 - R^2 + |xc|^2}.
\label{intlower}
\end{equation}
In this example, $L = \ebs(\tilde{x}) = \lfs(\tilde{x})$, so
in Equation~(\ref{intlower})
we can replace $L$ with either of those expressions.

The following lemma shows that for any smooth manifold $\Sigma$,
the interpolation error can never be worse than in this example.
Moreover (and happily), the crucial characteristic of $\tau$ is not
its circumradius $R$, but the radius of its smallest enclosing ball.
(Note that in the lemma below, $B_\tau$ can be {\em any} enclosing ball.)

\begin{lemma}[Surface Interpolation Lemma]
\label{interplemma}
Let $\Sigma \subset \R^d$ be a smooth $k$-manifold, and
let $M$ be its medial axis.
Let $\tau$ be a simplex (of any dimension) whose vertices lie on $\Sigma$.
Let $B_\tau$ be a closed $d$-ball such that $B_\tau \supseteq \tau$
(e.g., $\tau$'s smallest enclosing ball or $\tau$'s diametric ball),
let $c$ be its center, and let $r$ be its radius.
For every point $x \in \tau$ such that $x \not\in M$,
if $r < \ebs(\tilde{x})$ then
\[
|x\tilde{x}| \leq
\ebs(\tilde{x}) - \sqrt{\ebs(\tilde{x})^2 - r^2 + |xc|^2}
\leq \ebs(\tilde{x}) - \sqrt{\ebs(\tilde{x})^2 - r^2}
= \frac{r^2}{2 \, \ebs(\tilde{x})} +
{\cal O}\left( \frac{r^4}{\ebs(\tilde{x})^3} \right),
\]
and if $r < \lfs(\tilde{x})$ then
\[
|x\tilde{x}| \leq
\lfs(\tilde{x}) - \sqrt{\lfs(\tilde{x})^2 - r^2 + |xc|^2}
\leq \lfs(\tilde{x}) - \sqrt{\lfs(\tilde{x})^2 - r^2}
= \frac{r^2}{2 \, \lfs(\tilde{x})} +
{\cal O}\left( \frac{r^4}{\lfs(\tilde{x})^3} \right).
\]
The first inequality in each line is sharp for balls that circumscribe $\tau$
(that is, when every vertex of $\tau$ lies on the boundary of $B_\tau$):
there exists a $\Sigma$ such that
$|x\tilde{x}| = \ebs(\tilde{x}) - \sqrt{\ebs(\tilde{x})^2 - r^2 + |xc|^2} =
\lfs(\tilde{x}) - \sqrt{\lfs(\tilde{x})^2 - r^2 + |xc|^2}$
for every simplex $\tau$ whose vertices lie on $\Sigma$,
every $x \in \tau \setminus M$, and
every ball $B_\tau$ that circumscribes $\tau$ and has radius
$r < \ebs(\tilde{x})$.
The second inequality in each line is sharp when $x = c$.
\end{lemma}

\begin{proof}
Let $x$ be a point on $\tau \setminus M$; then $\tilde{x}$ is uniquely defined.
If $x \in \Sigma$ then $|x\tilde{x}| = 0$ and the result follows immediately,
so assume that $x \not\in \Sigma$; thus $\tau$ has at least two vertices.
Let $B$ be the open medial ball tangent to $\Sigma$ at $\tilde{x}$ such that
$x$ lies on the line segment $\tilde{x}m$,
where $m \in M$ is the center of $B$, as illustrated in Figure~\ref{interpo}.
($B$ is the medial ball found by ``growing'' a ball tangent to $\Sigma$
at $\tilde{x}$ so its center moves linearly through $x$ and
stops at a medial axis point $m$.)
As $B \cap \Sigma = \emptyset$, no vertex of $\tau$ lies in $B$.
(Note that $B$ cannot degenerate to a halfspace because
a halfspace containing $x$ would contain at least one vertex of $\tau$.)
Let $L = |\tilde{x}m| \geq \ebs(\tilde{x}) \geq \lfs(\tilde{x})$ be
the radius of $B$.
As $x$ lies on $\tilde{x}m$, $|x\tilde{x}| = L - |xm|$.

\begin{figure}
\centerline{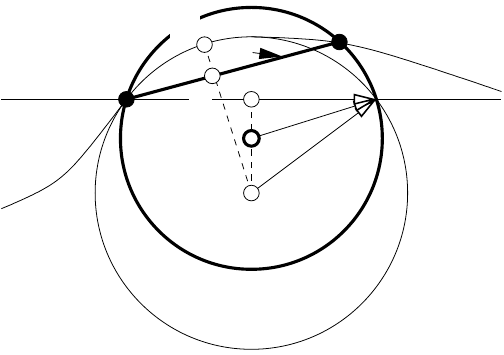}

\caption{\label{interpo}  \protect\small \sf
The Interpolation Lemma:  for every $x \in \tau \setminus M$,
$|x\tilde{x}| \leq \ebs(\tilde{x}) - \sqrt{\ebs(\tilde{x})^2 - r^2 + |xc|^2}$,
where $c$ and $r$ are the center and radius of a ball $B_\tau \supset \tau$.
}
\end{figure}


Let $S$ be the intersection of the boundaries of $B$ and $B_\tau$.
By the following reasoning, $S$ is a $(d - 2)$-sphere
(e.g., a circle in $\R^3$).
The two balls must intersect at more than one point,
as $x \in \tau \subseteq B_\tau$ and $x$ is in the open ball $B$.
By assumption, $r < \ebs(\tilde{x}) \leq L$, so
it is not possible that $B \subseteq B_\tau$.
Nor is it possible that $B_\tau$ is included in the closure of $B$, as
$\tau$'s vertices (there are at least two) lie in $B_\tau$ but not in $B$.

Let $\Pi$ be the unique hyperplane that includes $S$.
$\Pi$ divides $\R^d$ into a closed halfspace $H_\tau$ and
an open halfspace $H$, as illustrated.
The portion of $B_\tau$ in $H_\tau$ includes the portion of $B$ in $H_\tau$
(i.e., $B_\tau \cap H_\tau \supset B \cap H_\tau$), whereas
the portion of $B$ in $H$ includes the portion of $B_\tau$ in $H$
(i.e., $B \cap H \supset B_\tau \cap H$).
Every vertex of $\tau$ lies in $H_\tau$, because
$B_\tau$ contains every vertex of $\tau$ and $B$ contains no vertex of $\tau$.
Hence $\tau \subset H_\tau$. 

Recall that $m$ and $c$ are the centers of $B$ and $B_\tau$, respectively, and
observe that $\Pi$ is orthogonal to $cm$.
Moreover, the vector $c - m$ points
``out of'' the halfspace $H$ and ``into'' the halfspace $H_\tau$.
Let $z$ and $\rho$ be the center and radius of $S$.
Observe that $z \in \Pi$ and $z$ is collinear with $cm$.
By Pythagoras' Theorem,
$L^2 = \rho^2 + |zm|^2$ and $r^2 = \rho^2 + |zc|^2$.

Every point $x \in \tau$ lies in $H_\tau$, and
$z$ lies on the boundary of $H_\tau$, so
the angle separating the vectors $x - z$ and $c - m$ is at most $90^\circ$.
Hence
\begin{equation}
(x - z) \cdot (c - m) \geq 0.
\label{dotpos}
\end{equation}
It follows that
\begin{eqnarray*}
|xm|^2 - L^2 + r^2 - |xc|^2
&    = & |xm|^2 - |zm|^2 + |zc|^2 - |xc|^2 \\
&    = & - 2x \cdot m + 2z \cdot m - 2z \cdot c + 2x \cdot c \\
&    = & 2 (x - z) \cdot (c - m) \\
& \geq & 0.
\end{eqnarray*}
Therefore,
\begin{eqnarray}
|x\tilde{x}| & =    & L - |xm|  \nonumber \\
             & \leq & L - \sqrt{L^2 - r^2 + |xc|^2}  \label{Lexp} \\
             & \leq & \ebs(\tilde{x}) -
                      \sqrt{\ebs(\tilde{x})^2 - r^2 + |xc|^2}.  \label{ebsexp}
                      \\
             & \leq & \lfs(\tilde{x}) -
                      \sqrt{\lfs(\tilde{x})^2 - r^2 + |xc|^2}.  \label{lfsexp}
\end{eqnarray}
Inequalities~(\ref{ebsexp}) and~(\ref{lfsexp}) follow because (\ref{Lexp})~is
monotonically decreasing in $L$ (contrary to superficial appearances) and
$L \geq \ebs(\tilde{x}) \geq \lfs(\tilde{x})$.

We observe that the inequality~(\ref{Lexp}) holds with equality
if and only if~(\ref{dotpos}) holds with equality, which happens
if and only if $x \in \Pi$.
The inequalities~(\ref{ebsexp}) and~(\ref{lfsexp}) hold with equality when
$\Sigma$ is a $k$-sphere, in which case the medial ball $B$ is always
the open $d$-ball with the same center and radius as $\Sigma$.
Both inequalities hold with equality when
$\Sigma$ is a $k$-sphere and the boundary of $B_\tau$ circumscribes $\tau$,
in which case $S$ also circumscribes $\tau$, so $\tau \subset \Pi$ and
every point $x \in \tau$ lies on $\Pi$.
Hence, the inequalities are sharp as claimed.
\end{proof}

\section{Triangle Normal Lemmas}
\label{tnls}

Given a triangle $\tau$ whose vertices lie on a $k$-manifold $\Sigma$,
we derive bounds on how well $\tau$'s normal space
locally approximates the spaces normal to $\Sigma$ in the vicinity of $\tau$.
In this section, we derive a bound on
$\angle(N_\tau, N_v\Sigma) = \angle(\aff{\tau}, T_v\Sigma)$ where
$v$ is a vertex of $\tau$.
(In codimension~$1$, we can interpret this as the angle between normal vectors,
albeit a nonobtuse angle---we do not distinguish between a vector $n_v$ and
its negation $-n_v$.)
We first consider surfaces embedded in $\R^3$, then
we show that the same bound applies to
$k$-manifolds embedded in $\R^d$ for all $d > k \geq 2$
(for which the normal vectors are replaced by normal spaces).
In Section~\ref{etnls}, we give a bound on
$\angle(N_\tau, N_{\tilde{x}}\Sigma) =
\angle(\aff{\tau}, T_{\tilde{x}}\Sigma)$ applicable to every point
$x \in \tau$, not just at the vertices.
Hence, it applies to the normal spaces of all the points in $\nu(\tau)$.
Note that in the lemma, each occurrence of $\ebs(v)$ can be replaced by
$\lfs(v)$, as $\lfs(v) \leq \ebs(v)$.

\begin{lemma}[Triangle Normal Lemma for $\R^3$]
\label{lem:tnl}
Let $\Sigma$ be a smooth $2$-manifold without boundary embedded in $\R^3$. Let $\tau$ be a triangle whose vertices lie on $\Sigma$. Let $R$ be $\tau$'s circumradius. Let $v$ be a vertex of $\tau$ and let $\phi$ be $\tau$'s plane angle at $v$. Then
\[
\angle(N_\tau, N_v\Sigma) = \angle(\aff{\tau}, T_v\Sigma) \leq \arcsin{\left( \frac{R}{\ebs(v)} \max \left\{ \cot \frac{\phi}{2}, 1 \right\} \right)}.
\]
(Note that the argument $\cot \frac{\phi}{2}$ dominates if $\phi$ is acute and the argument $1$ dominates if $\phi$ is obtuse.)
In particular, if $v$ is the vertex at $\tau$'s largest plane angle (so $\phi \geq 60^\circ$) and $R < \ebs(v) / \sqrt{3} \doteq 0.577 \, \ebs(v)$, then
\[
\angle(N_\tau, N_v\Sigma) = \angle(\aff{\tau}, T_v\Sigma) \leq \arcsin \frac{\sqrt{3}R}{\ebs(v)}.
\]
\end{lemma}

\begin{proof}
Let $\theta = \angle(N_\tau, N_v\Sigma)$. Consider the two balls of radius $\ebs(v)$ tangent to $\Sigma$ at $v$. The plane $\aff{\tau}$ intersects these two balls in two circles of radius $\rho = \ebs{(v)} \sin \theta$, as Figure \ref{fig:ntnl1} shows. We consider these two circles $C_1$ and $C_2$ in the plane $\aff{\tau}$. Notice that since $C_1$ and $C_2$ are cross sections of surface-free balls, their insides are surface-free. In particular, $u$ and $w$ cannot lie strictly inside $C_1$ or $C_2$. We will use this fact to establish a relationship between the radius $\rho$ of these circles and the circumradius $R$ of $\tau$.

\begin{figure}[h!]
\begin{center}
\includegraphics[width=0.35\textwidth]{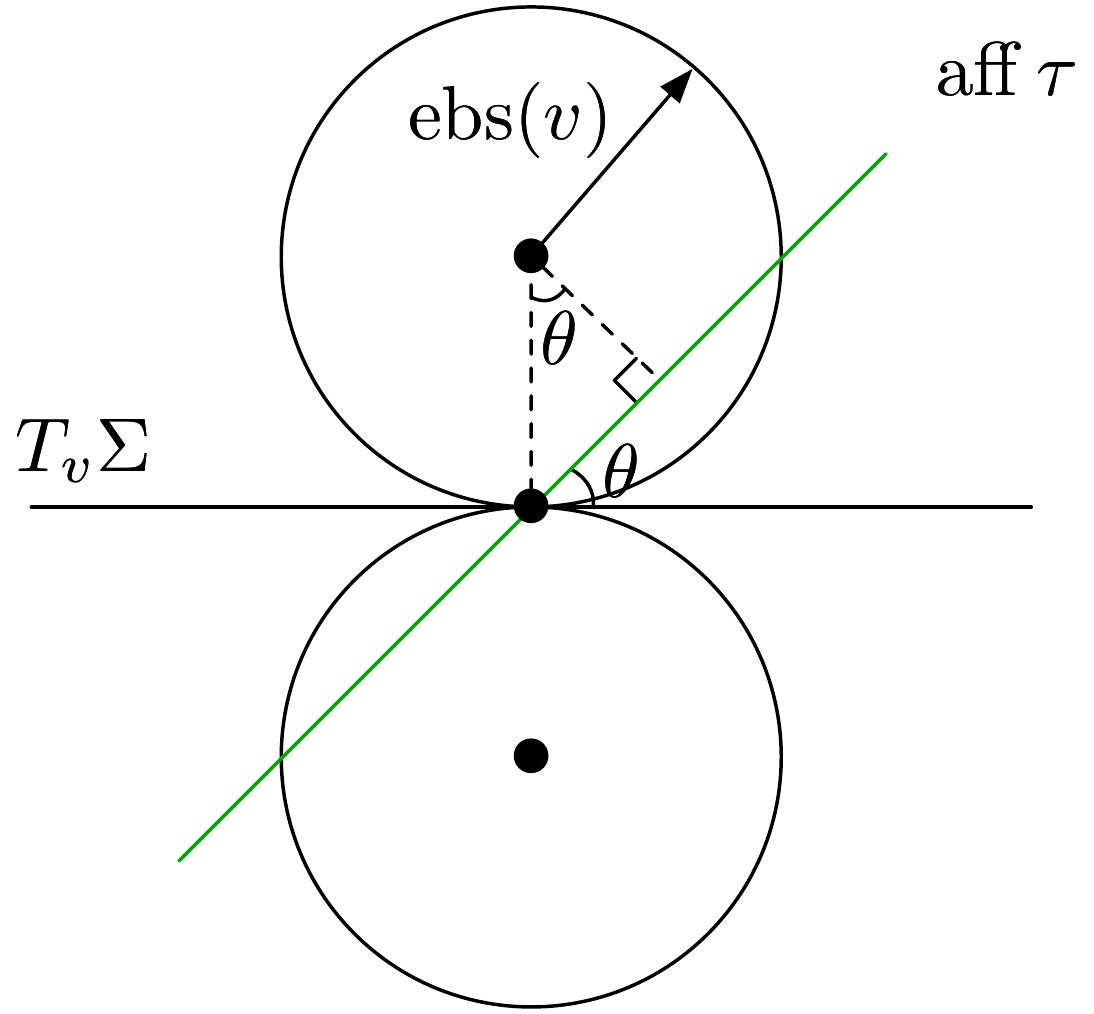}
\caption{\label{fig:ntnl1}  \protect\small \sf
The affine hull $\aff{\tau}$ intersects the surface-free balls of radius $\ebs(v)$ in two circles of radius $\ebs{(v)} \sin \theta$.}
\end{center}
\end{figure}

Let $c_1$ and $c_2$ be the centers of $C_1$ and $C_2$, respectively. Imagine that as $\theta$ increases, and $\aff{\tau}$ tilts further, $C_1$ grows in the direction $\vec{vc_1}$ while remaining in contact with $v$, and $C_2$ grows in the opposite direction. We distinguish two cases: (1)~either $\vec{vc_1}$ or $\vec{vc_2}$ points into $\tau$ or (2)~both $\vec{vc_1}$ and $\vec{vc_2}$ point to the exterior of $\tau$. See Figures~\ref{fig:ntnl2} and~\ref{fig:ntnl3}.

\begin{figure}[h!]
  \centering
  \includegraphics[width=0.5\linewidth]{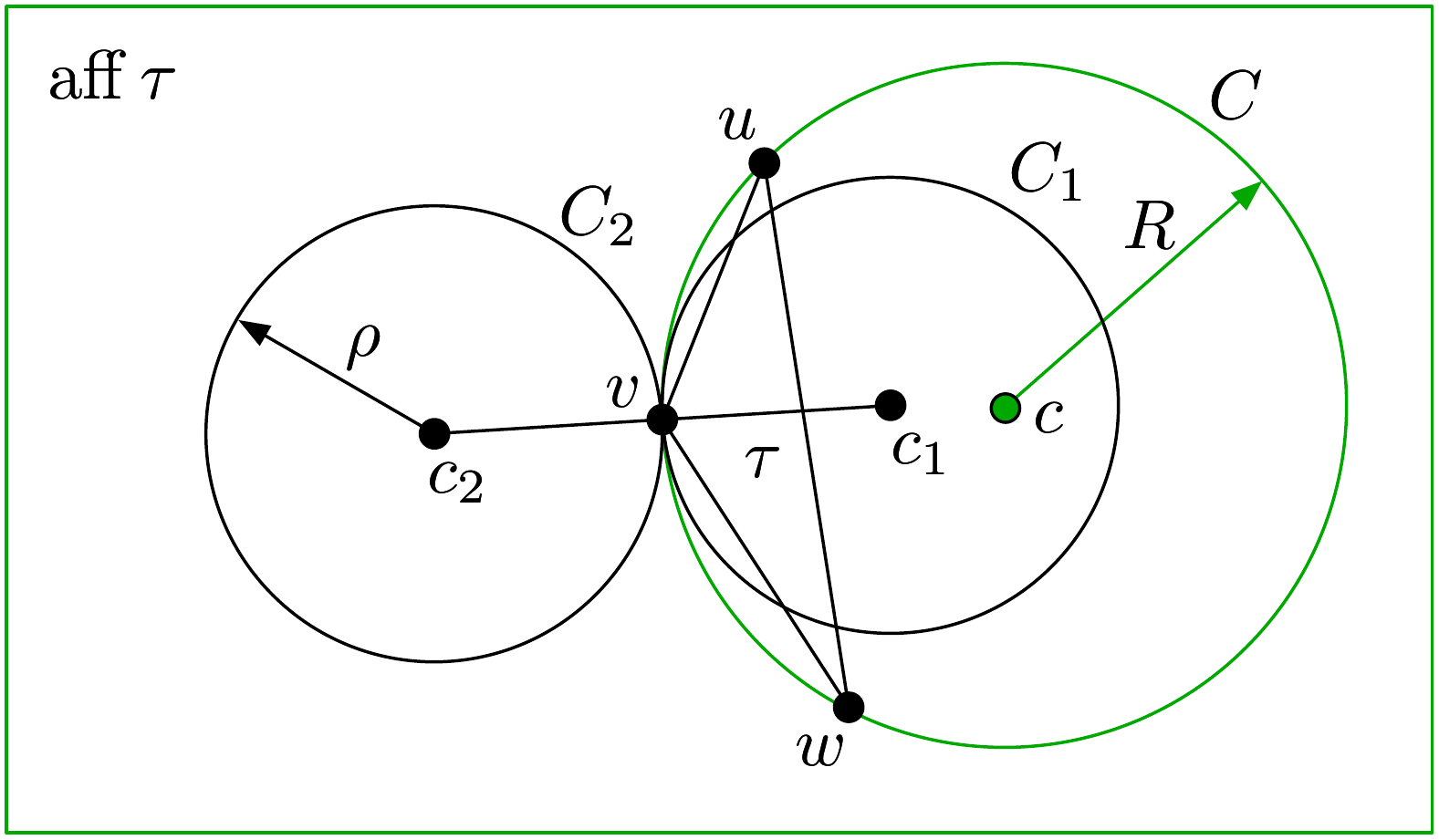}
  \caption{\label{fig:ntnl2}  \protect\small \sf
    Case 1, where one of the two circles grows into the interior of $\tau$. In this case, the radius of $C_1$ is at most $R$.
  }
\end{figure}

Let $\tau = \triangle uvw$. Let $C$ be the circumcircle of $\tau$ in the plane $\aff{\tau}$, and let $c$ be the center of $C$. In case~1, illustrated in Figure~\ref{fig:ntnl2}, one of $vc_1$ or $vc_2$ points into $\tau$; suppose it is $vc_1$. $C_1$ cannot grow indefinitely; eventually it intersects $u$ or $w$. The maximum angle is achieved when $C_1 = C$, whereupon $u$ and $w$ prevent further growth. Thus $R \geq \rho = \ebs{(v)} \sin{\theta}$ which implies that $\theta \leq \arcsin{\frac{R}{\ebs{(v)}}}$.


\begin{figure}[h!]
  \centering
  \includegraphics[width=0.5\linewidth]{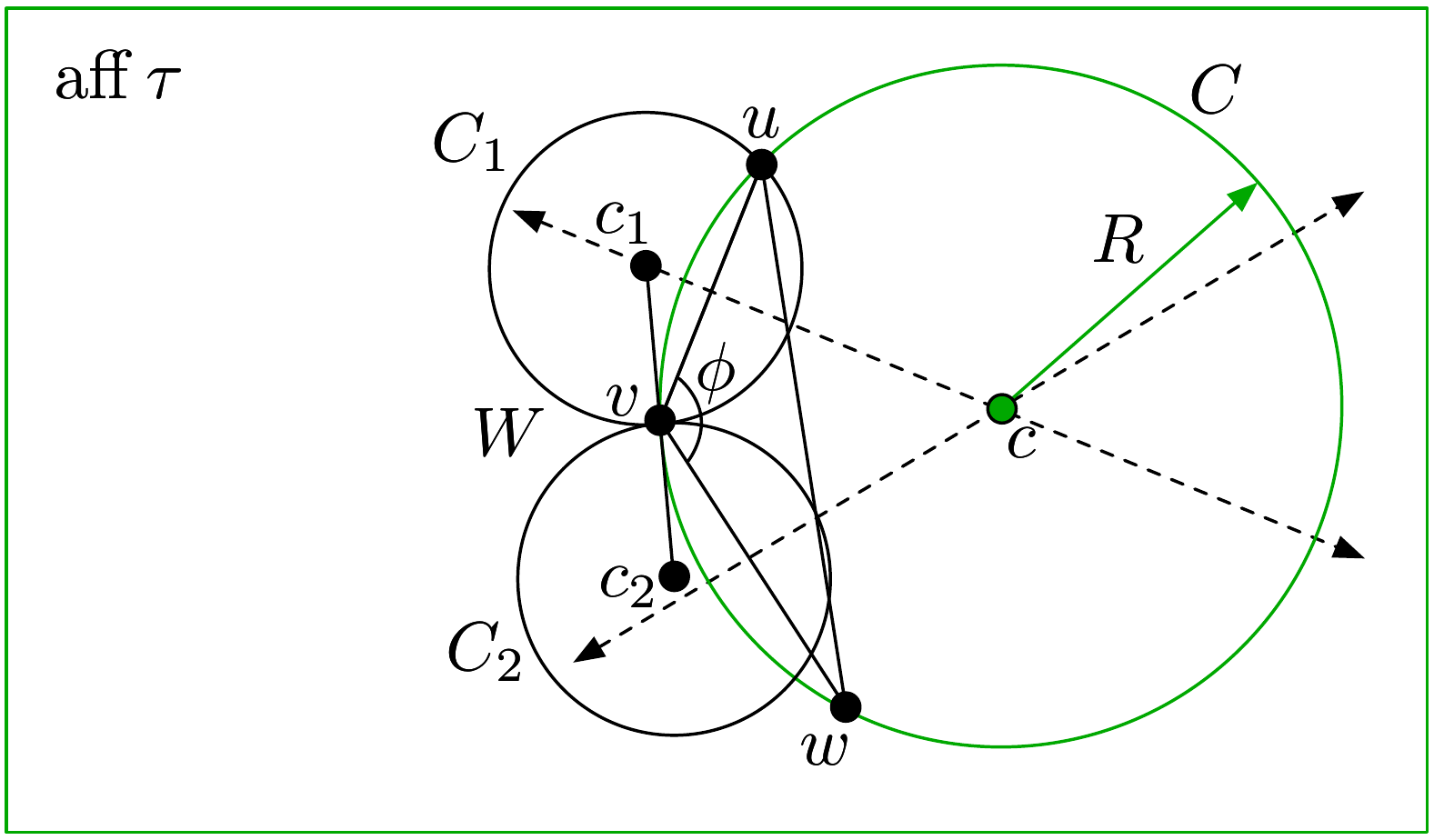}
  \caption{\label{fig:ntnl3}  \protect\small \sf
    Case 2, where both circles grow into the exterior of $\tau$. In this case, the bound depends on the angle $\phi$ at $v$.
  }
\end{figure}

In case~2, the line segment $c_1c_2$ does not intersect $\tau$ except at $v$, as Figure~\ref{fig:ntnl3} shows. The bisectors of $vu$ and $vw$ divide the plane into four wedges with apex $c$; let $W$ be the closed wedge that contains $v$. As $vu$ and $vw$ meet at $v$ at an angle $\phi$, the wedge angle where the bisectors meet at $c$ is $180^\circ - \phi$, as illustrated in Figure~\ref{fig:arcs}.

\begin{figure}[h!]
\centerline{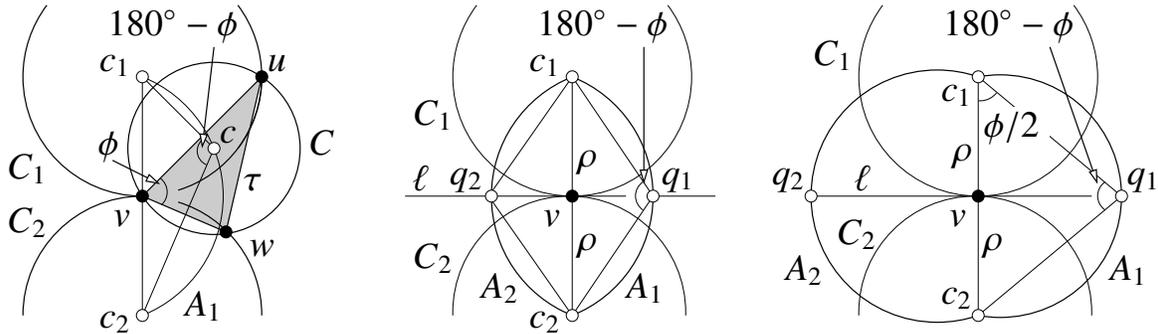}

\caption{\label{fig:arcs}  \protect\small \sf
Left:  the triangle angle of $\phi$ induces a wedge angle of $180^\circ - \phi$.  Center:  the circumcenter $c$ cannot lie inside the region enclosed by arcs $A_1$ and $A_2$, here illustrated for an acute $\phi$.  Right:  For an obtuse $\phi$.
}
\end{figure}


As $u$ is not inside the circle $C_1$, $|uc_1| \geq |vc_1|$. Similarly, $|wc_1| \geq |vc_1|$. It follows that $c_1 \in W$. Similarly, $c_2 \in W$. Therefore, $\angle{c_1cc_2} \leq 180^\circ - \phi$. By circle geometry, this inequality implies that we can draw two circular arcs with endpoints $c_1$ and $c_2$ such that $c$ cannot be strictly inside the region enclosed by the arcs. Specifically, let $\ell$ be the line that bisects $c_1c_2$. Let $q_1$ and $q_2$ be the two distinct points on $\ell$ such that $\angle c_1q_1c_2 = 180^\circ - \phi$ and $\angle c_1q_2c_2 = 180^\circ - \phi$, as illustrated in Figure~\ref{fig:arcs}. Both of these angles are bisected by $\ell$; that is, $\angle c_iq_jv = 90^\circ - \phi / 2$ for $i \in \{ 1, 2 \}$, $j \in \{1, 2 \}$. Thus we have four similar right triangles adjoining $v$ of the form $\triangle c_ivq_j$ with $\angle q_jc_iv = \phi / 2$.

Observe that $|vc_1| = |vc_2| = \rho = \ebs{(v)} \sin{\theta}$, hence $|vq_1| = |vq_2| = \rho \tan (\phi / 2)$. Consider the unique circular arc $A_1$ having endpoints $c_1$ and $c_2$ and passing through $q_1$, and its mirror image arc $A_2$ passing through $q_2$, as illustrated. By circle geometry, for every point $q$ on $A_1$ or $A_2$ (except $c_1$ or $c_2$), $\angle c_1qc_2 = 180^\circ - \phi$, and for every point $q$ enclosed between the two arcs, $\angle c_1qc_2 > 180^\circ - \phi$. It follows that the circumcenter $c$ cannot lie in the region enclosed by $A_1$ and $A_2$.

As $\sin \theta \leq \rho / \ebs(v)$, our goal is to determine the maximum possible value of $\rho$ for a fixed value of $R$. Equivalently, we wish to determine the minimum value of $R = |vc|$ for a fixed $\rho$. In other words, with $\rho$ fixed, what is the closest that $c$ can get to $v$? If $\phi \leq 90^\circ$, then the distance $|vc|$ is minimized for $c = q_1$ or $c = q_2$ (see Figure~\ref{fig:arcs}, center), in which case $R = |vq_1| = \rho \tan (\phi / 2)$. If $\phi \geq 90^\circ$, then $|vc|$ is minimized for $c = c_1$ or $c = c_2$ (see Figure~\ref{fig:arcs}, right), in which case $r = |vc_1| = \rho$. It follows that $r \geq \rho \min \{ \tan(\phi / 2), 1 \}$, hence $\sin \theta \leq \rho / \ebs(v) \leq R \max \{ \cot(\phi / 2), 1 \} / \ebs(v)$.
\end{proof}

Compare Lemma~\ref{lem:tnl} with two prior versions of
the Triangle Normal Lemma.
The following lemma gives the best known bound, which was proven by
Amenta, Choi, Dey, and Leekha~\cite{amenta02}.
(The derivation can also be found in
Dey~\cite{dey07} and Cheng et al.~\cite{cheng12}.)

\begin{lemma}
\label{lem:otnl}
Let $\Sigma$ be a smooth $2$-manifold without boundary embedded in $\R^3$. Let $\tau$ be a triangle whose vertices lie on $\Sigma$. Let $R$ be $\tau$'s circumradius. Let $v$ be the vertex of $\tau$ at $\tau$'s largest plane angle. If $R \leq 0.433 \, \lfs(v)$, then
\begin{equation*}
\angle(N_\tau, N_v\Sigma) = \angle(\aff{\tau}, T_v\Sigma) \leq \arcsin\left(\frac{R}{\lfs(v)}\right) + \arcsin\left(\frac{2}{\sqrt{3}}\sin\left(2\arcsin\left(\frac{R}{\lfs(v)}\right)\right)\right).
\end{equation*}
\end{lemma}

The year before, Cheng, Dey, Edelsbrunner, and Sullivan~\cite{cheng01} derived
a stronger bound of $\arcsin \frac{2R}{\lfs(v)}$, but
it seems to have escaped notice.
All three bounds are plotted in Figure~\ref{tnlplots} (left).
Lemma~\ref{lem:tnl} improves upon both prior results in three ways:
it is tighter for the case covered by Lemma~\ref{lem:otnl}
(improving the Cheng et al.\ bound by a factor of $1.15$ and
the Amenta et al.\ bound by a factor of $1.91$
for small values of $R / \lfs(v)$),
it applies to any vertex $v$ of $\tau$, and
it takes into account $\tau$'s angle at $v$.

Lemma~\ref{lem:tnl} extends straightforwardly to
higher-dimensional manifolds embedded in higher-dimensional Euclidean spaces
(but not to higher-dimensional simplices).
Given a triangle $\tau$ whose vertices lie on
a $k$-manifold $\Sigma \subset \R^d$,
we wish to know the worst-case angle deviation
$\angle (\aff \tau, T_v\Sigma)$ between $\tau$'s affine hull and
the tangent space at a vertex $v$ of $\tau$.

\begin{lemma}[Triangle Normal Lemma for $\R^d$]
\label{lem:tnlhigh}
Let $\Sigma$ be a smooth $k$-manifold without boundary embedded in $\R^d$,
with $k \geq 2$.
Let $\tau$ be a triangle whose vertices lie on $\Sigma$.
Let $R$ be $\tau$'s circumradius.
Let $v$ be a vertex of $\tau$ and let $\phi$ be $\tau$'s plane angle at $v$.
Then
\[
\angle(N_\tau, N_v\Sigma) = \angle(\aff{\tau}, T_v\Sigma) \leq
\arcsin{\left( \frac{R}{\ebs(v)} \max \left\{ \cot \frac{\phi}{2}, 1 \right\}
               \right)}.
\]
\end{lemma}

\begin{proof}
The dimension of $N_v\Sigma$ is less than or equal to
the dimension of $N_\tau$ (which is $d - 2$), so by definition,
\[
\angle (N_\tau, N_v\Sigma) = \max_{\ell_v \subset N_v\Sigma}
\min_{\ell_N \subset N_\tau} \angle (\ell_N, \ell_v)
\]
where $\ell_v$ and $\ell_N$ are lines.
Let $\ell_v \subset N_v\Sigma$ and $\ell_N \subset N_\tau$ be lines such that
$\angle (N_\tau, N_v\Sigma) = \angle (\ell_N, \ell_v)$,
translated so they pass through $v$ (without loss of generality).
If $\angle (\ell_N, \ell_v) = 0$ the result follows immediately, so
suppose that $\angle (\ell_N, \ell_v) > 0$.
Let $\Pi$ be the plane ($2$-flat) that includes both $\ell_v$ and $\ell_N$.
Let $\ell_\tau \subset \Pi$ be the line through $v$ perpendicular to $\ell_N$
in $\Pi$.
As $\ell_N$ is chosen from the flat $N_\tau$ to minimize
its angle with $\ell_v$, the line $\ell_\tau$ is orthogonal to $N_\tau$, and
therefore $\ell_\tau$ lies in the complementary flat $\aff \tau$.
Let $\Xi \subset \R^d$ be the unique $3$-flat that includes
$\tau$ and $\ell_N$.
As $\Xi$ includes $\aff \tau$, $\ell_\tau \subset \Xi$; and
as $\Xi$ also includes $\ell_N$, $\Pi \subset \Xi$, hence $\ell_v \subset \Xi$.

We reiterate the proof of Lemma~\ref{lem:tnl} to bound
$\angle (\ell_N, \ell_v)$, with $\Xi$ replacing $\R^3$ and
$\ell_v$ replacing $N_v\Sigma$ in the proof.
The proof of Lemma~\ref{lem:tnl} relies entirely on the fact that
$\tau$'s vertices cannot be inside the two open balls of radius $\ebs(v)$
that are centered on $\ell_v$ and touching $v$.
In the present setting in $\R^d$,
every open ball of radius $\ebs(v)$ tangent to $\Sigma$ at $v$ is surface-free;
two of those balls have centers on $\ell_v$.
The intersections of these balls with $\Xi$ are
surface-free $3$-balls of radius $\ebs(v)$, so
the constraints harnessed by the proof of Lemma~\ref{lem:tnl} hold
in the subspace $\Xi$.
Therefore, the bound of Lemma~\ref{lem:tnl} holds for $k$-manifolds in $\R^d$
as well.
%
\end{proof}


\section{Normal Variation Lemmas}
\label{variation}

Recall that, given two nearby points $p, q \in \Sigma$,
we seek an upper bound on the {\em normal variation},
the angle $\angle (n_p, n_q)$ separating their normal vectors
(in codimension~$1$) or
the angle $\angle (N_p\Sigma, N_q\Sigma)$ separating their normal spaces
(in codimension~$2$ or higher).

\begin{lemma}[Normal Variation Lemma for Codimension $1$]
\label{ptnormal}
Let $\Sigma \subset \mathbb{R}^d$ be
a bounded, smooth $(d - 1)$-manifold without boundary.
Consider two points $p, q \in \Sigma$ and let $\delta = |pq| / \lfs(p)$.
Let $n_p$ and $n_q$ be
outward-directed vectors normal to $\Sigma$ at $p$ and $q$, respectively.

If $\delta < \sqrt{4 \sqrt{5} - 8} \doteq 0.9717$, then
$\angle (n_p, n_q) \leq \eta_1(\delta)$ where
\begin{equation}
\eta_1(\delta) = \arccos \left( 1 - \frac{\delta^2}{2 \sqrt{1 - \delta^2}}
                                \right)
\approx
\delta + \frac{7}{24} \delta^3 + \frac{123}{640} \delta^5 +
\frac{1\mbox{,}083}{7\mbox{,}168} \delta^7 + O(\delta^9).
\label{codim1bound}
\end{equation}
Moreover, if $\delta_N$ is the component of $\delta$
parallel to $p$'s normal line $N_p\Sigma$---that is,
$\delta_N$ is the distance from $q$ to the tangent space $T_p\Sigma$
divided by $\lfs(p)$---we have the bound
(which is stronger when $\delta_N \neq 0$)
\begin{equation}
\angle (n_p, n_q) \leq
\arccos \left(1 -
\frac{\delta^2 - \delta^4 / 2 - 2 \delta_N^2}
     {\sqrt{(1 - \delta^2) \, \left( (2 - \delta^2)^2 - 4 \delta_N^2 \right)}}
\right).
\label{codim1bound2}
\end{equation}
\end{lemma}

Recall that the right-hand side of Inequality~(\ref{codim1bound}) is
plotted in green in Figure~\ref{nvlplots}.
Two isocontour plots of the right-hand side of Inequality~(\ref{codim1bound2})
appear in Figure~\ref{nvlcodim1}.
In most circumstances where a normal variation lemma is applied,
$|pq|$ is known but the normal component $\delta_N$ is not.
It is clear from the plot on the left that for any given value of $\delta$,
the bound~(\ref{codim1bound2}) is weakest at $\delta_N = 0$;
this substitution yields the bound~(\ref{codim1bound}).
Hence the green curve in Figure~\ref{nvlplots} also represents
the horizontal midline of the isocontour plot.

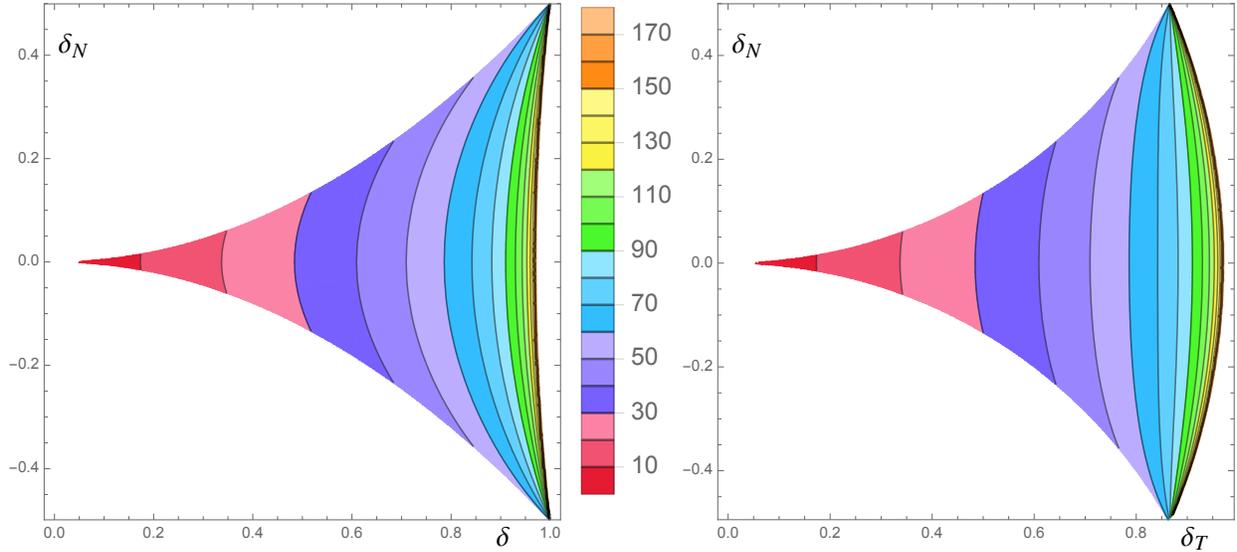
\begin{figure}
\centerline{
  \setlength{\unitlength}{2.9in}
  \begin{picture}(1,0.981)
  \put(0,0){\includegraphics[width=2.9in]{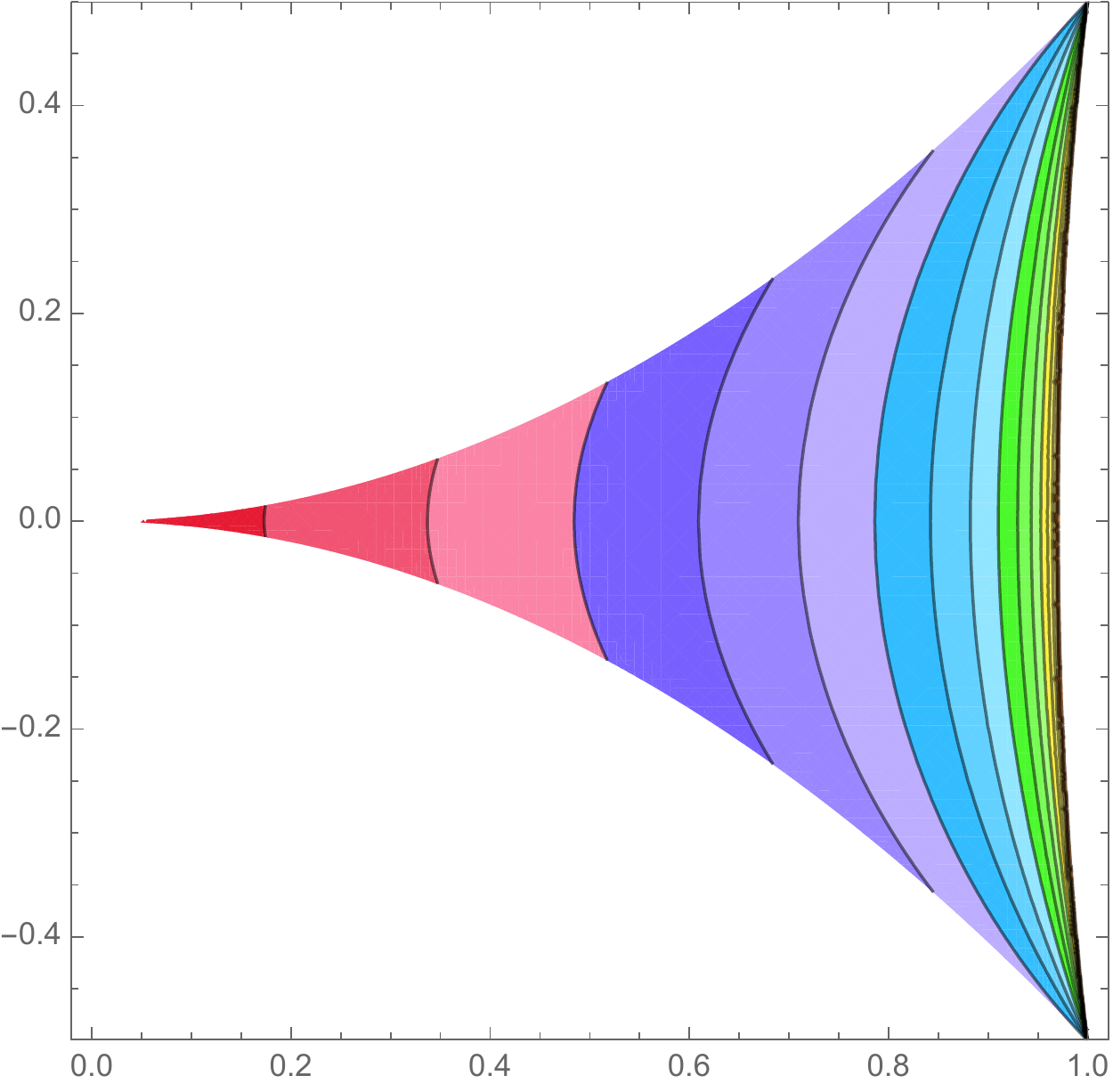}}
  \put(0.88,0){$\delta$}
  \put(0.09,0.89){$\delta_N$}
  \end{picture}
  \includegraphics[width=0.55in]{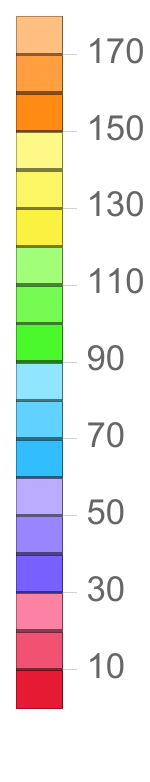}
  \begin{picture}(1,0.981)
  \put(0,0){\includegraphics[width=2.9in]{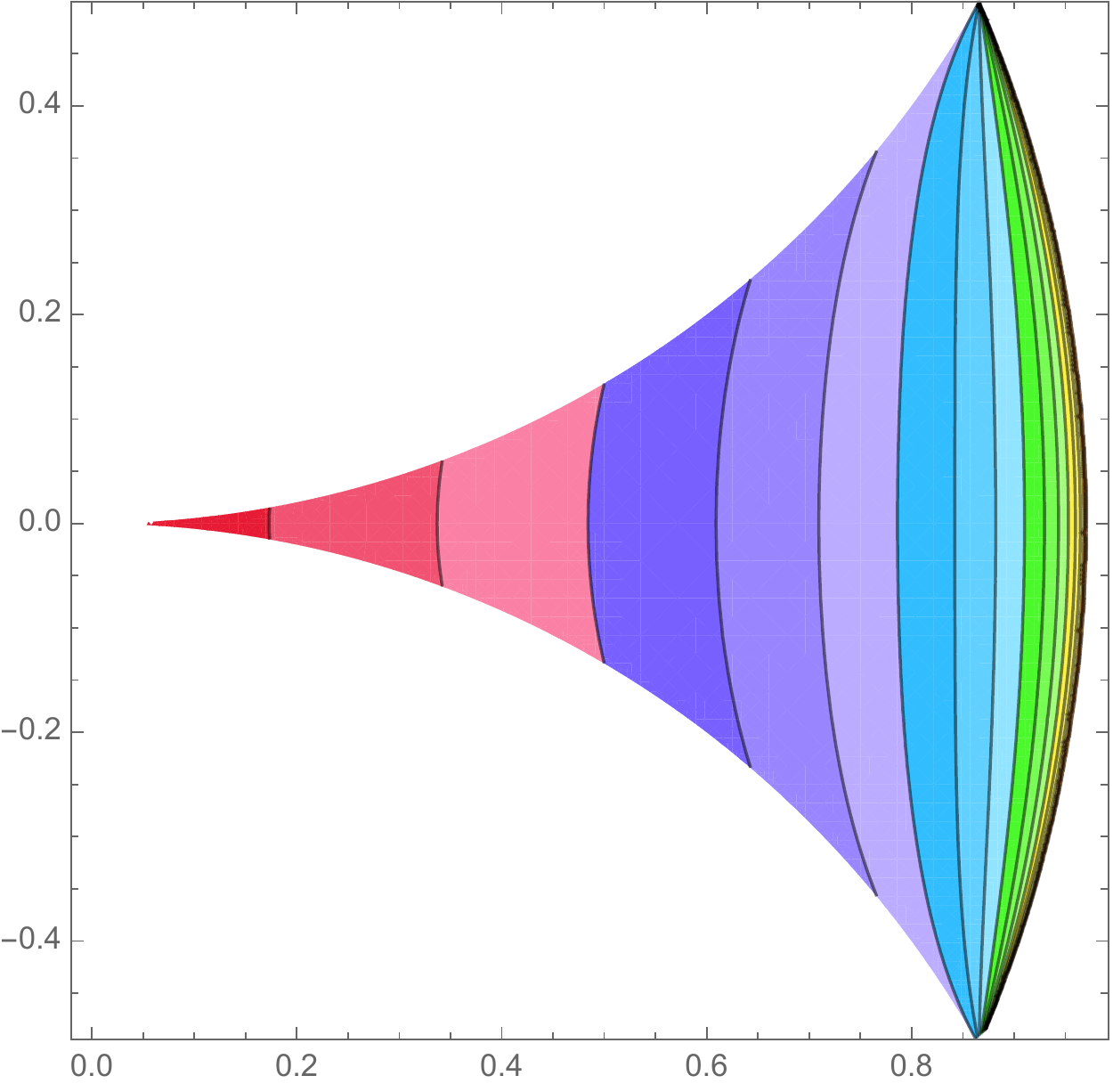}}
  \put(0.9,0){$\delta_T$}
  \put(0.09,0.89){$\delta_N$}
  \end{picture}
}


\caption{\label{nvlcodim1}  \protect\small \sf
Left:  The upper bound (in degrees) for $\angle (n_p, n_q)$
as specified by Inequality~(\ref{codim1bound2}), as a function of
$\delta = |pq| / \lfs(p)$ (on the horizontal axis) and
the normal component $\delta_N$ of $\delta$ (on the vertical axis);
i.e., $\delta_N$ is the distance from $q$ to the tangent space $T_p\Sigma$
divided by $\lfs(p)$.
Right:  A similar plot with one change:
the horizontal axis is the tangential component $\delta_T$ of $\delta$;
i.e., the distance from $q$ to the normal line $N_p\Sigma$ divided by
$\lfs(p)$.
Hence, this plot reflects the Euclidean geometry of the space,
with $p$ at the origin, $T_p\Sigma$ on the horizontal midline,
$q$ somewhere in the colored region, and
the two surface-free balls of radius $\lfs(p)$ (white) blocking
$q$ from occupying certain regions (compare with Figure~\ref{master}).
}
\end{figure}

\begin{proof}
Let $F$ be the open ball with center $p$ and radius $\lfs(p)$.
By the definition of $\lfs$,
$F$ does not intersect the medial axis $M$ of $\Sigma$.
The line $N_p\Sigma$ normal to $\Sigma$ at $p$ intersects the boundary of $F$
at two opposite poles $o$ and $o'$.
By assumption, $|pq| < \lfs(p)$, so $q \in F$ and
the normal line $N_q\Sigma$ intersects the boundary of $F$
at two points $z$ and $z'$.

\begin{figure}
\centerline{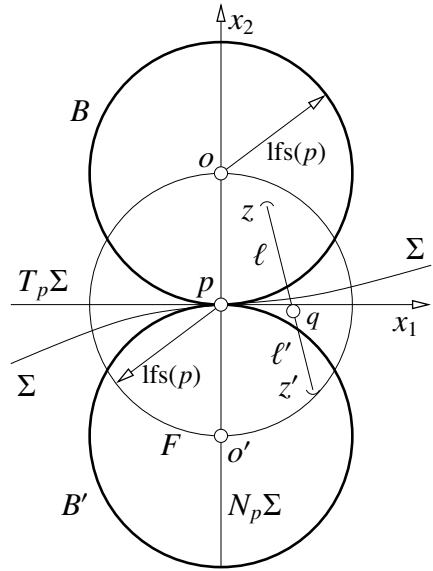}

\caption{\label{master}
The medial-free ball $F$ and surface-free balls $B$ and $B'$ associated with
a point $p \in \Sigma$.
}
\end{figure}

Let $B$ and $B'$ be the two open balls of radius $\lfs(p)$ tangent to $\Sigma$
at $p$, illustrated in Figure~\ref{master};
the centers of these balls are $o$ and $o'$, respectively.
Neither ball intersects $\Sigma$ nor contains $q$.
Let $Z$ be the open ball centered at $z$ with its boundary passing through $q$,
and define $Z'$ likewise with its center at $z'$.
Each of $Z$ and $Z'$ is a subset of a medial ball tangent to $\Sigma$ at $q$,
so neither ball intersects $\Sigma$ nor contains $p$.
Without loss of generality,
suppose that $B'$ and $Z'$ are enclosed by $\Sigma$, whereas 
$B$ and $Z$ are outside the region enclosed by $\Sigma$.
Therefore, $B$ is disjoint from $Z'$, and $B'$ is disjoint from $Z$.
(However, $B$ may intersect $Z$, and $B'$ may intersect $Z'$.)
This property is the key to obtaining a bound on $\angle (n_p, n_q)$.


We create a $d$-axis coordinate system with $p$ at the origin.
For simplicity, we will scale the coordinate system so that $\lfs(p) = 1$;
hence $B$, $B'$, and $F$ all have radius $1$.
The $x_2$-axis is the normal line $N_p\Sigma$,
which passes through $o$, $p$, and $o'$ and is directed so that
$o = (0, 1, 0, \ldots, 0)$, $o' = (0, -1, 0, \ldots, 0)$, and
$p = (0, 0, \ldots, 0)$, as illustrated in Figure~\ref{master}.
The remaining axes span the tangent space $T_p\Sigma$.
We choose an $x_1$-axis on $T_p\Sigma$ such that its positive branch passes
through the orthogonal projection of $q$ onto $T_p\Sigma$;
that is, $q_1 \geq 0$ and $q_3 = q_4 = \ldots = q_d = 0$.
We choose an $x_3$-axis on $T_p\Sigma$ such that
the normal line $N_q\Sigma$ lies in the $x_1$-$x_2$-$x_3$-space
(which is now the affine hull of $N_p\Sigma \cup N_q\Sigma$).
Hence, $z_4 = z_5 = \ldots = z_d = 0$ and $z'_4 = z'_5 = \ldots = z'_d = 0$.
All the important features of the problem lie on the three-dimensional
cross-section of $\R^d$ specified by these three coordinates.

Let $\ell = |qz|$ and $\ell' = |qz'|$ be
the radii of $Z$ and $Z'$, respectively.
The unit ball $F$ has a diameter $e$ that passes through $q$
(and through the origin $p$, like all diameters of $F$).
The point $q$ subdivides $e$ into
a line segment of length $1 + \|q\|$ and a line segment of length $1 - \|q\|$.
As this diameter and the line segment $zz'$ intersect each other at $q$,
they are both chords of a common circle on the boundary of $F$,
illustrated in Figure~\ref{ichords}.
By the well-known Intersecting Chords Theorem,
\begin{equation}
\ell \ell' = (1 + \|q\|) \, (1 - \|q\|) = 1 - \|q\|^2,
\label{intchords}
\end{equation}
where $\|q\|^2 = q_1^2 + q_2^2$ (as $q$'s other coordinates are zero).
Note that $\|q\|$ is the distance from $p$ to $q$.

\begin{figure}
\centerline{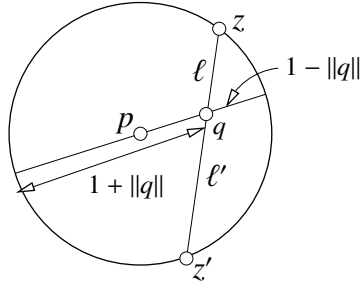}

\caption{\label{ichords}
The Intersecting Chords Theorem:  $\ell \ell' = (1 + \|q\|) \, (1 - \|q\|)$.
}
\end{figure}


The balls $Z$ and $B'$ (with centers $z$ and $o'$ and radii $\ell$ and $1$)
are disjoint and $z$ lies on the unit sphere, so
\begin{eqnarray}
\ell + 1 & \leq & |zo'|  \nonumber \\
         & =    & \sqrt{z_1^2 + (z_2 + 1)^2 + z_3^2}  \nonumber \\
         & =    & \sqrt{2 + 2z_2}.
\label{disjoint1}
\end{eqnarray}
Symmetrically, $Z'$ and $B$ are disjoint, so
\begin{equation}
\ell' + 1 \leq \sqrt{2 - 2z'_2}.
\label{disjoint2}
\end{equation}

If one of the inequalities~(\ref{disjoint1}) or~(\ref{disjoint2}) holds
with equality, we call this event a {\em tangency}.
A tangency between $Z$ and $B'$ implies that
\begin{equation}
z_2 = \frac{(\ell + 1)^2}{2} - 1,  
\label{tangency1}
\end{equation}
whereas a tangency between $Z'$ and $B$ implies that
\begin{equation}
z'_2 = 1 - \frac{(\ell' + 1)^2}{2}.  
\label{tangency2}
\end{equation}

Our goal is to find an upper bound on $\angle (n_p, n_q)$.
This angle is the tilt of the line segment $zq$ relative to the $x_2$-axis, so
\[
\cos \angle (n_p, n_q) =
\frac{z_2 - q_2}{|zq|} = \frac{z_2 - q_2}{\ell}.
\]

To find a bound, we seek to determine the configuration(s) in which
the angle is maximized---hence, the cosine is minimized---subject to
Inequalities~(\ref{disjoint1}) and~(\ref{disjoint2}).
We will see that the maximum is obtained when both inequalities hold with
equality, a configuration we call a {\em dual tangency},
illustrated in Figure~\ref{dualtan}.

\begin{figure}
\centerline{
  \setlength{\unitlength}{2.8in}
  \begin{picture}(1,0.967)
  \put(0,0){\includegraphics[width=2.8in]{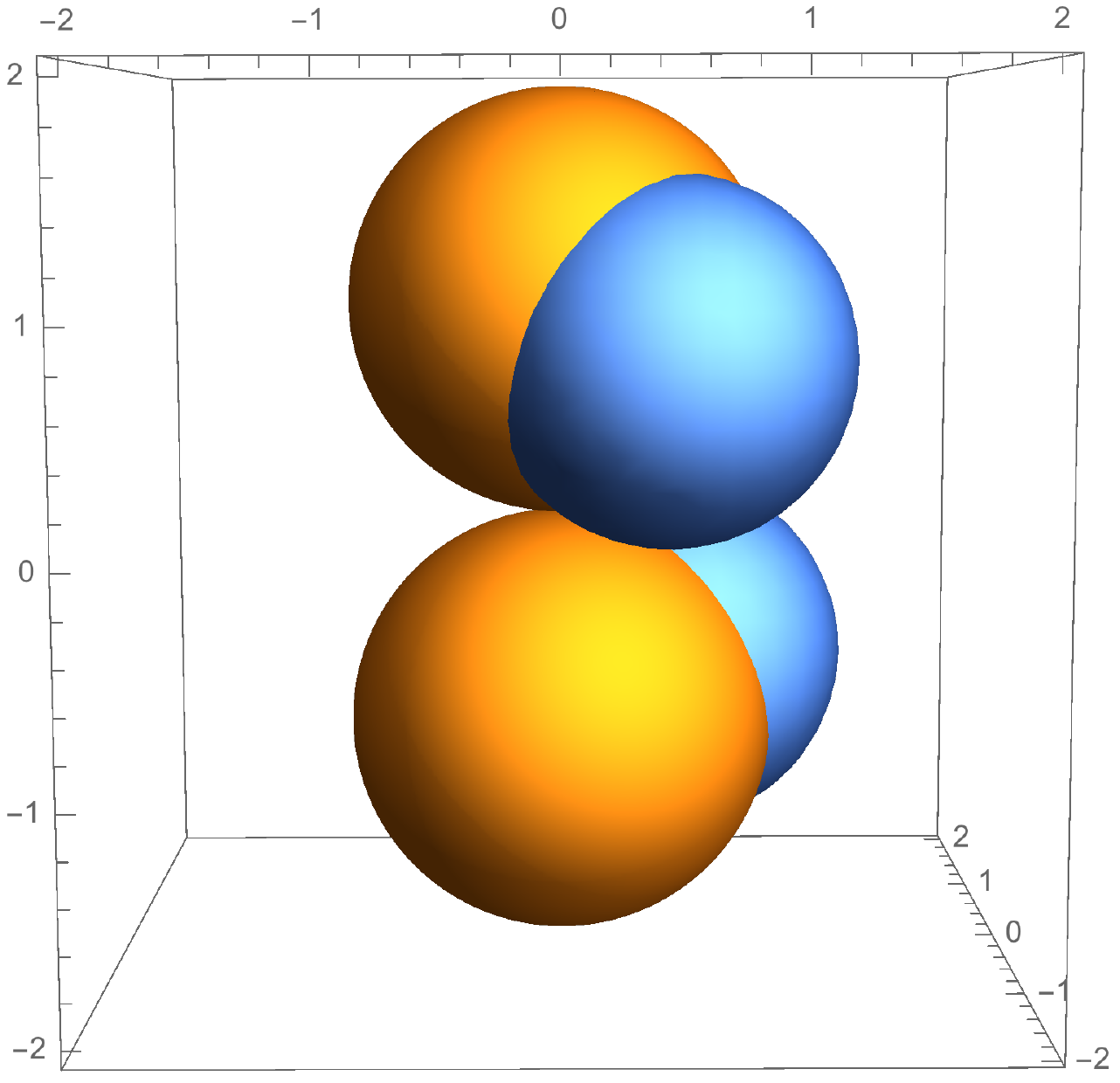}}
  \put(0.81,0.95){$x_1$}
  \put(-0.03,0.8){$x_2$}
  \put(0.81,0.11){$x_3$}
  \end{picture}
  \begin{picture}(1,1.041)
  \put(0,0){\includegraphics[width=2.8in]{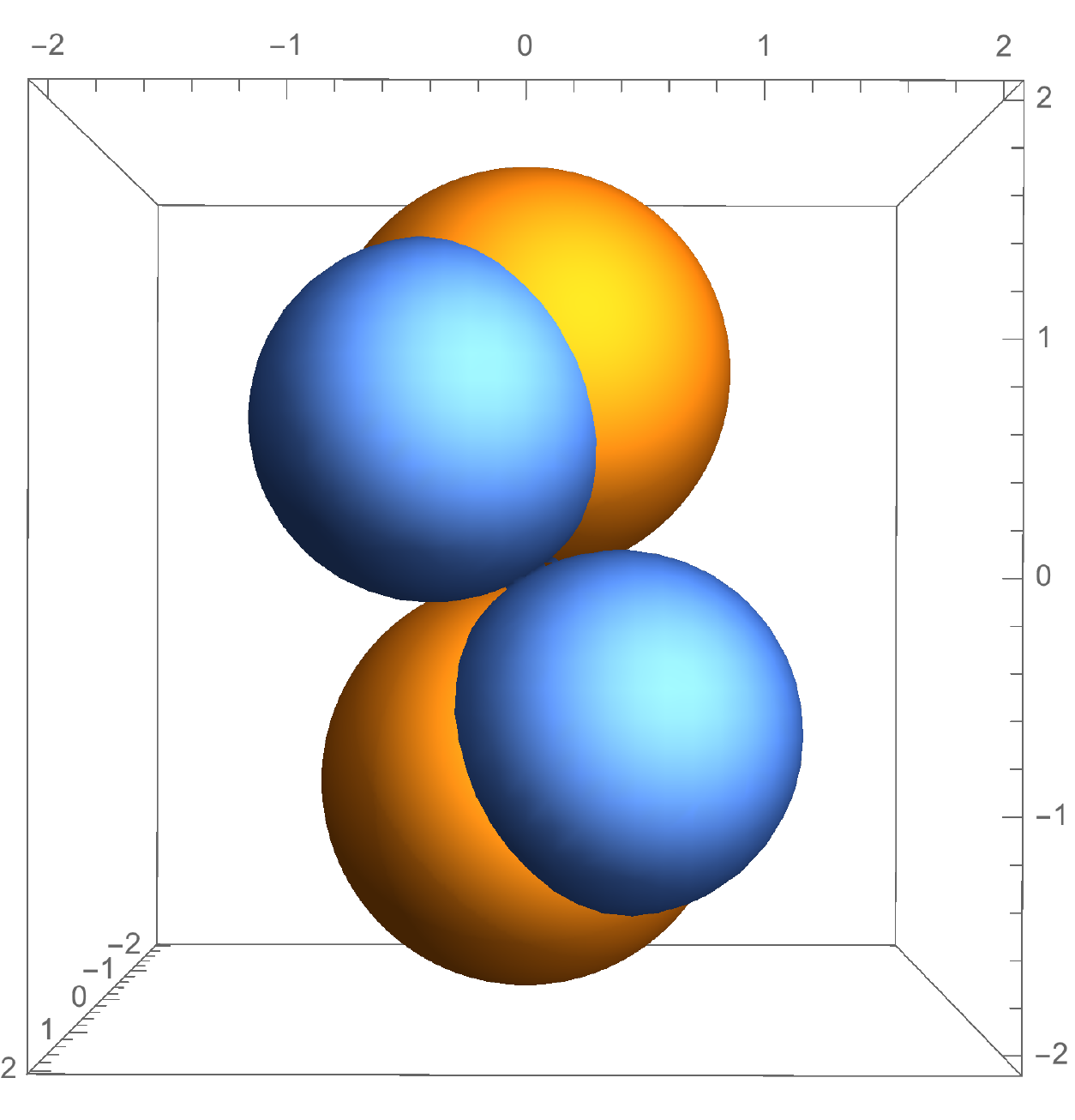}}
  \put(0.81,1){$x_3$}
  \put(0.97,0.84){$x_2$}
  \put(0.13,0.09){$x_1$}
  \end{picture}
}
\centerline{
  \setlength{\unitlength}{2.8in}
  \begin{picture}(1,0.967)
  \put(0,0){\includegraphics[width=2.8in]{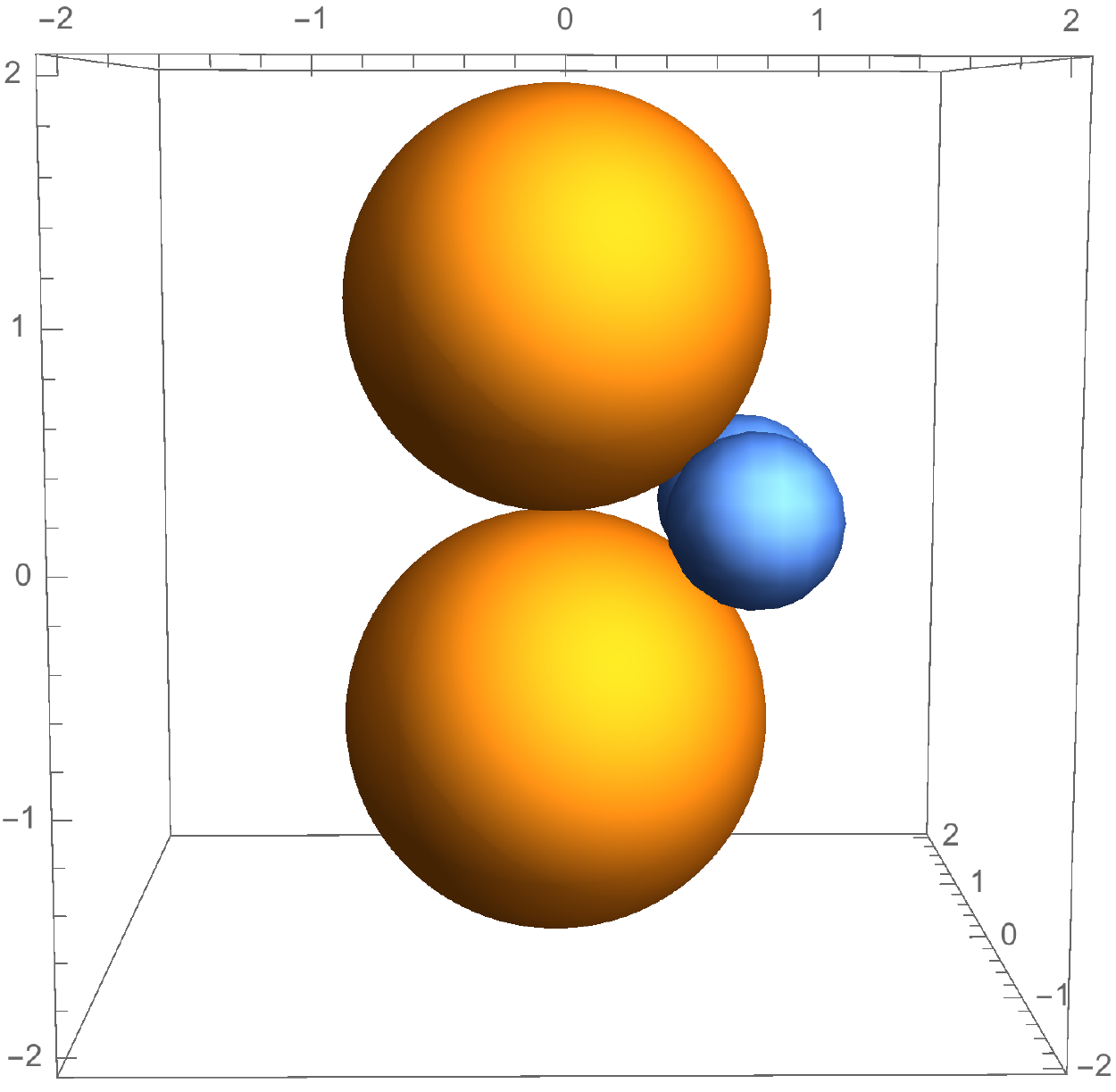}}
  \put(0.81,0.95){$x_1$}
  \put(-0.03,0.8){$x_2$}
  \put(0.81,0.11){$x_3$}
  \end{picture}
  \begin{picture}(1,1.041)
  \put(0,0){\includegraphics[width=2.8in]{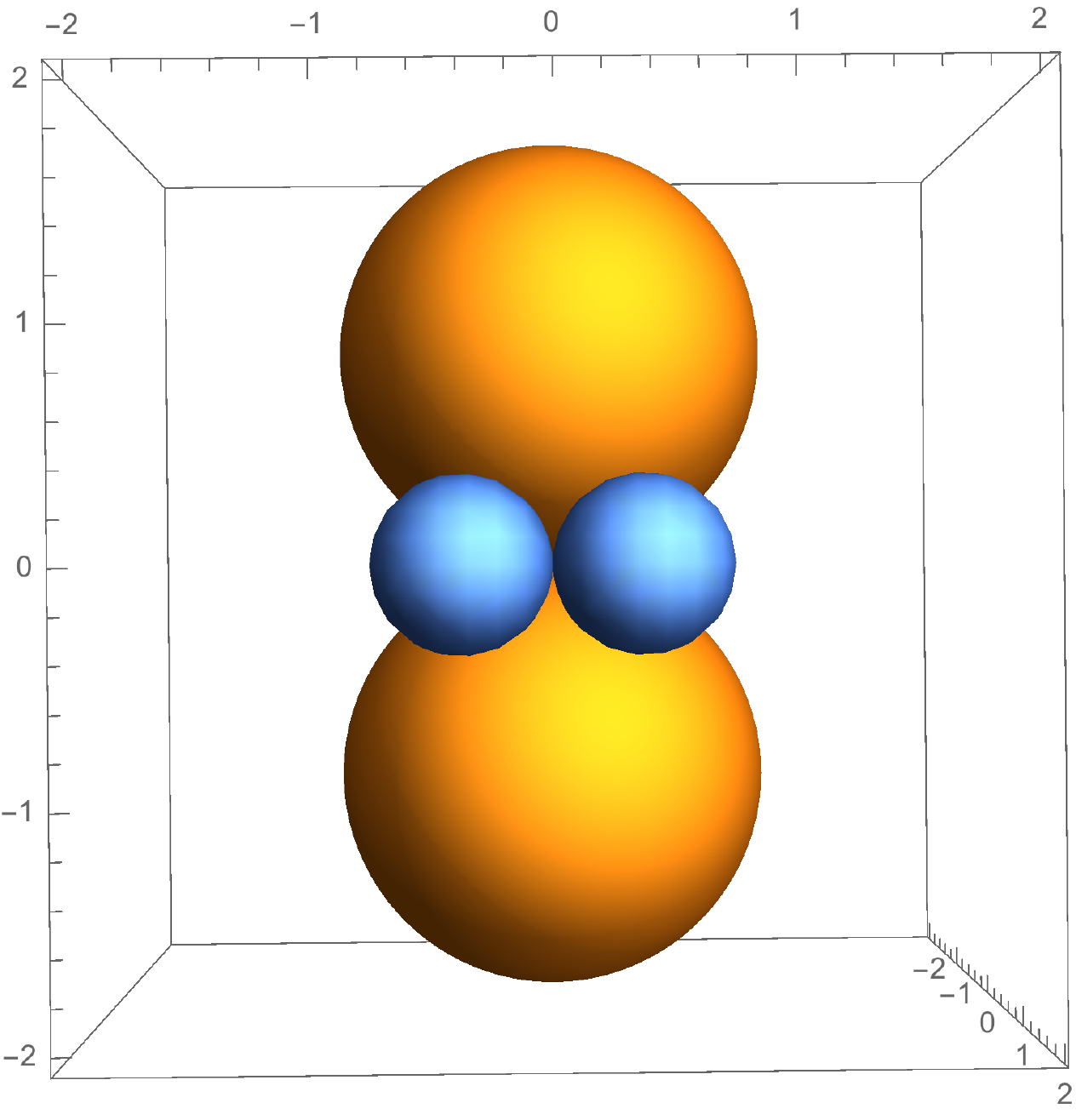}}
  \put(0.81,1.02){$x_3$}
  \put(0.07,0.84){$x_2$}
  \put(0.81,0.09){$x_1$}
  \end{picture}
}


\caption{\label{dualtan}  \protect\small \sf
Dual tangency configurations for $\delta = 0.5$ (top two images) and
$\delta = 0.9101$ (bottom two images).
In the former configuration, $\angle (n_p, n_q) \doteq 31.17^\circ$, and
in the latter configuration, $\angle (n_p, n_q) \doteq 90^\circ$.
The orange balls are $B$ and $B'$, with $p$ at their point of tangency, and
the blue balls are $Z$ and $Z'$, with $q$ at their point of tangency.
The manifold $\Sigma$ passes through $p$ and $q$ but does not intersect
the interior of any of these balls.
}
\end{figure}
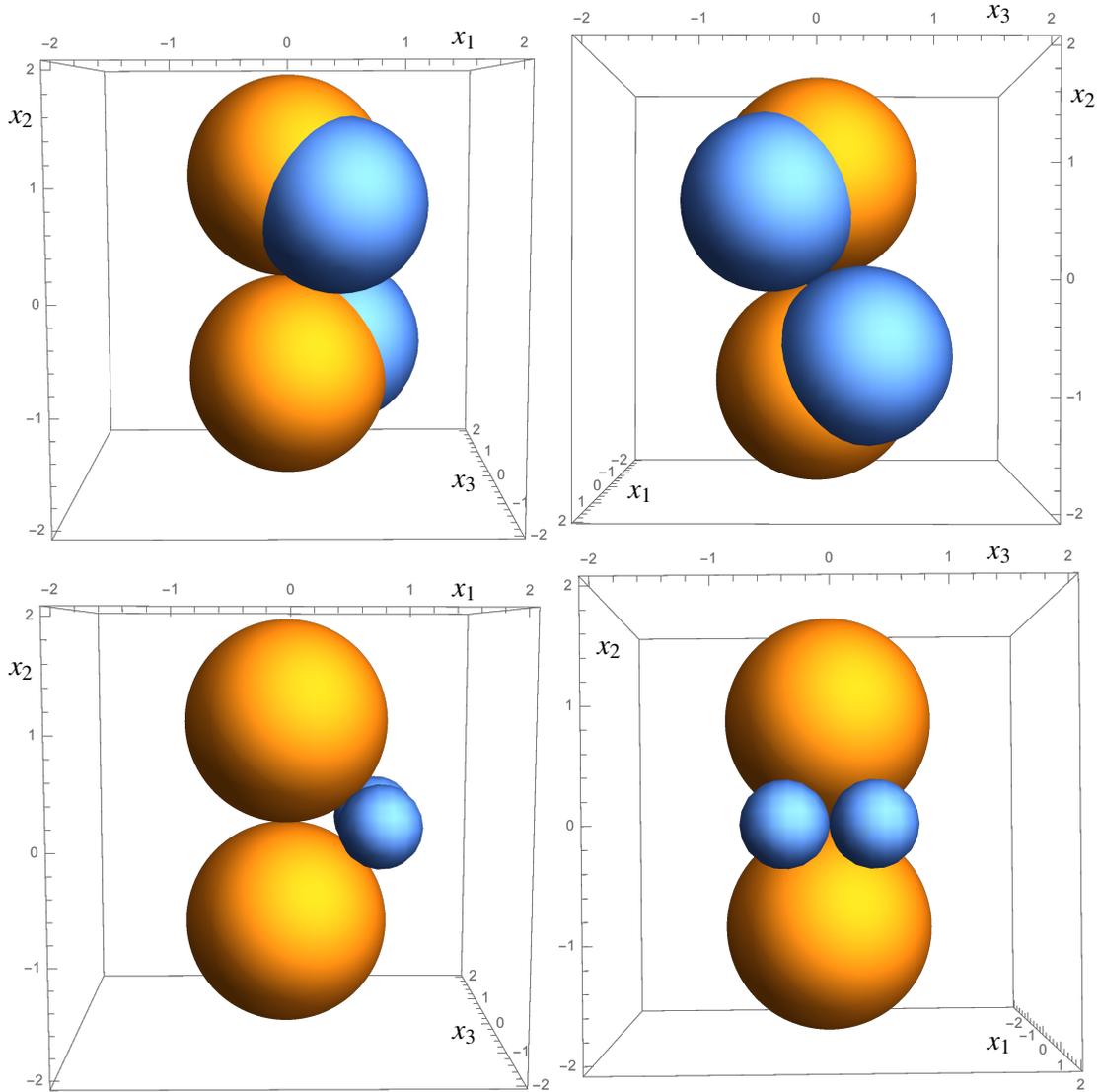

In a configuration where neither tangency is engaged
(i.e., both inequalities are strict), we can increase
$\angle (n_p, n_q)$ and decrease its cosine by freely tilting
the line segment $zz'$ while maintaining the constraints that
$zz'$ passes through $q$, and both $z$ and $z'$ lie on the boundary of $F$.
(Note that in our coordinate system,
$q$, $p$, $B$, $B'$, $F$, and $n_p$ are all fixed, but
we can adjust $n_q$ subject to the inequalities.)
Therefore, if the maximum possible angle is not $180^\circ$,
a configuration that maximizes the angle must engage at least one tangency.
As $Z$ and $Z'$ play symmetric roles, we can assume without loss of generality
that $Z$ is tangent to $B'$ and Equation~(\ref{tangency1}) holds, giving
\begin{equation}
\cos \angle (n_p, n_q) = 1 + \frac{\ell^2 - 1 - 2 q_2}{2\ell}.
\label{cos}
\end{equation}

The derivative
$\frac{\partial}{\partial\ell} \cos \angle (n_p, n_q) =
(\ell^2 + 1 + 2 q_2) / (2\ell^2)$ is positive for all $q_2 \geq -1 / 2$;
we have $q_2 \in (-1 / 2, 1 / 2)$ because $q \in F$, $q \not\in B'$, and
$q \not\in B$.
Therefore, the cosine~(\ref{cos}) increases monotonically with $\ell$.
We see from Equation~(\ref{intchords}) that $\ell$ increases monotonically as
$\ell'$ decreases.
Inequality~(\ref{disjoint2}) places an upper bound on $\ell'$,
which together with~(\ref{intchords}) places a lower bound on $\ell$,
which places a lower bound on the cosine~(\ref{cos}) and
an upper bound on the angle $\angle (n_p, n_q)$ itself.
A configuration attains this upper bound on $\angle (n_p, n_q)$ when
Inequality~(\ref{disjoint2}) holds with equality---in a dual tangency, 
where $Z'$ is tangent to $B$ in addition to $Z$ being tangent to $B'$,

A dual tangency uniquely determines the values of $\ell$ and $\ell'$.
As $q \in zz'$, we can write
\begin{equation}
\ell (z'_2 - q_2) = \ell' (q_2 - z_2).
\label{proportions}
\end{equation}
The identities~(\ref{intchords}), (\ref{tangency1}), (\ref{tangency2}),
and~(\ref{proportions}) form a system of four (nonlinear) equations in
the four variables $\ell$, $\ell'$, $z_2$, and $z'_2$.
According to Mathematica (and verified by substitution),
these equations are simultaneously satisfied by
\begin{equation}
\ell = \sqrt{\frac{(1 - \|q\|^2) \, (2 + 2q_2 - \|q\|^2)}{2 - 2q_2 - \|q\|^2}}
  \hspace{.2in} \mbox{and} \hspace{.2in}
\ell' =
  \sqrt{\frac{(1 - \|q\|^2) \, (2 - 2q_2 - \|q\|^2)}{2 + 2q_2 - \|q\|^2}}.
\label{ell}
\end{equation}

As this configuration places a lower bound on $\ell$,
substituting the identity~(\ref{ell}) into~(\ref{cos}) shows that
\begin{equation}
\cos \angle (n_p, n_q) \geq 1 - \frac{\|q\|^2 - \|q\|^4 / 2 - 2 q_2^2}
{\sqrt{(1 - \|q\|^2) \, \left( (2 - \|q\|^2)^2 - 4 q_2^2 \right)}}.
\label{cos2}
\end{equation}
Recall the parameter $\delta = |pq| / \lfs(p)$.
As we chose and scaled our coordinate system so that $p$ is the origin and
$\lfs(p) = 1$, $\|q\| = \delta$ and $q_2 = \delta_N$.
Inequality~(\ref{codim1bound2}) follows.

This expression provides a strong upper bound when
the value of $q_2$ (the distance from $q$ to $T_p\Sigma$) is known, but
$q_2$ is not usually available in circumstances where
the Normal Variation Lemma is invoked.
To find a bound independent of $q_2$, we seek the value of
$q_2 \in (-\|q\|^2 / 2, \|q\|^2 / 2)$ that minimizes
the right-hand side of~(\ref{cos2}).
The left plot in Figure~\ref{nvlcodim1} makes it clear that
for all $\|q\| < 1$, this value is $q_2 = 0$.
To verify this formally, observe that~(\ref{cos2}) is
symmetric about $q_2 = 0$ (as it is a function of $q_2^2$) and
\[
\frac{\partial}{\partial q_2} \cos \angle (n_p, n_q) =
2 q_2
\frac{3(1 - \|q\|^2)^2 + 4(1 - \|q\|^2) + (1 - 4q_2^2)}
{\sqrt{1 - \|q\|^2} \, \left( (2 - \|q\|^2)^2 - 4 q_2^2 \right)^{3/2}}.
\]
The numerator and denominator are positive
for all $\|q\| < 1$ and $q_2 \in (-0.5, 0.5)$, so
the derivative is zero at $q_2 = 0$, positive for $q_2 > 0$, and
negative for $q_2 < 0$, showing that the cosine is minimized at $q_2 = 0$.
Setting $q_2 = 0$ shows that
\[
\cos \angle (n_p, n_q) \geq 1 - \frac{\|q\|^2}{2 \sqrt{1 - \|q\|^2}},
\]
proving Inequality~(\ref{codim1bound}).
\end{proof}

We conjecture (but are not certain) that
Inequality~(\ref{codim1bound}) is sharp:
for every legal $\delta$, there exists a surface $\Sigma$ and
points $p, q \in \Sigma$ for which the bound holds with equality.
Proving this conjecture would entail finding a surface $\Sigma$ that is
compatible with the four balls $B$, $B'$, $Z$, and $Z'$ in the dual tangency
described in the proof of Lemma~\ref{ptnormal} and
illustrated in Figure~\ref{dualtan}---meaning that $\Sigma$
intersects none of the four balls but passes through
the four points of tangency $p$, $q$, $z$, and $z'$---such that
no point of $\Sigma$'s medial axis lies in the ball $F$.

Figure~\ref{dualtan} reveals that in the worst-case configuration,
$n_q$ is tilted along the $x_3$-axis (so $z_3 = -z'_3 \neq 0$), but
not along the $x_1$-axis (i.e., $z_1 = z'_1 = q_1$).
In other words, $\Sigma$ undergoes a helical twisting as one walks
from $p$ to $q$.
By contrast, a tilt along the $x_1$-axis cannot be as large.

The proof of the Normal Variation Lemma for higher codimensions is
similar in many respects, but it takes a different turn because
adding an extra dimension to the normal space enables a novel configuration
(not possible in codimension~$1$) such that the largest angle no longer occurs
when $q \in T_p\Sigma$.

\begin{lemma}[Normal Variation Lemma for Codimension $2$ and Higher]
\label{ptnormal2}
Let $\Sigma \subset \mathbb{R}^d$ be
a bounded, smooth $k$-manifold without boundary for any $k < d$.
Consider two points $p, q \in \Sigma$ and let $\delta = |pq| / \lfs(p)$.

If $\delta < \sqrt{\left( \sqrt{5} - 1 \right) / 2} \doteq 0.7861$, then
$\angle (N_p\Sigma, N_q\Sigma) = \angle (T_p\Sigma, T_q\Sigma) \leq
\eta_2(\delta)$ where
\begin{equation}
\eta_2(\delta) = \arccos
  \sqrt{1 - \frac{\delta^2}{\sqrt{1 - \delta^2}}}
\approx
\delta + \frac{5}{12} \delta^3 + \frac{57}{160} \delta^5 +
\frac{327}{896} \delta^7 + O(\delta^9).
\label{codim2bound}
\end{equation}

Moreover, if $\delta_N$ is the component of $\delta$
parallel to $p$'s normal space $N_p\Sigma$---that is,
$\delta_N$ is the distance from $q$ to the tangent space $T_p\Sigma$
divided by $\lfs(p)$---we have the (stronger) bound
\begin{equation}
\angle (N_p\Sigma, N_q\Sigma) \leq \arccos
\sqrt{\left( 1 - \frac{\delta^2}{2 \sqrt{1 - \delta^2}} \right)^2 -
      \frac{\delta_N^2}{1 - \delta^2}}.
\label{codim2bound2}
\end{equation}
In the special case where $q \in T_p\Sigma$ (that is, $\delta_N = 0$),
this bound reduces to
the codimension-$1$ bound $\eta_1(\delta)$ from Lemma~\ref{ptnormal}.
\end{lemma}

An isocontour plot of the right-hand side of Inequality~(\ref{codim2bound2})
appears in Figure~\ref{nvlcodim2}.
For any given value of $\delta$,
the bound~(\ref{codim2bound2}) is weakest along the upper (or lower) boundary
of the plot, at $\delta_N = \delta^2 / 2$;
this substitution yields the bound~(\ref{codim2bound}).
The upper boundary is also plotted as the red curve in Figure~\ref{nvlplots}.
Interestingly, the horizontal midline of this plot is
the green curve in Figure~\ref{nvlplots}:
when $\delta_N = 0$, the symmetry of the configuration yields the
codimension-$1$ bound $\eta_1(\delta)$.
The bound gets worse from there as $\delta_N$ increases.

We are certain that this bound can be tightened for larger values of $\delta_N$
(but not for $\delta_N = 0$), but
we have not been able to derive a better explicit bound.
It would be nice if the codimension~$1$ bound held for all $\delta_N$, but
we think it very unlikely;
we know a configuration in $\R^4$ that defies the codimension~$1$ bound and
which we think (but don't know for sure) can be realized by
a $2$-manifold fitting the specified constraints.

\begin{figure}
\centerline{
  \setlength{\unitlength}{2.9in}
  \begin{picture}(1,0.983)
  \put(0,0){\includegraphics[width=2.9in]{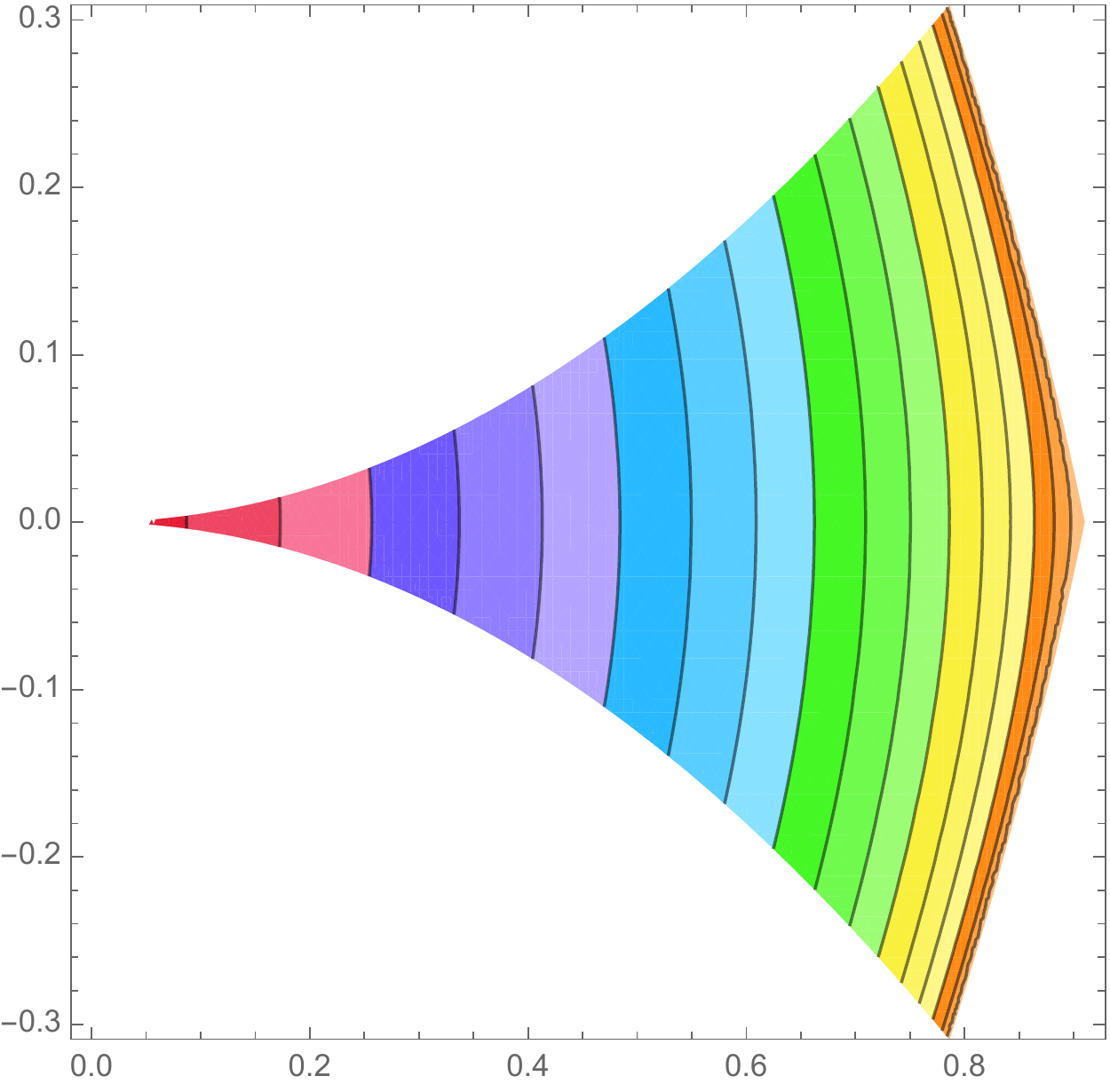}}
  \put(0.95,0){$\delta$}
  \put(0.09,0.89){$\delta_N$}
  \end{picture}
  \includegraphics[width=0.45in]{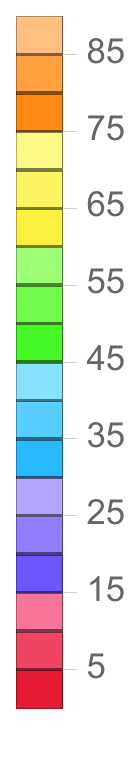}
  \begin{picture}(1,0.983)
  \put(0,0){\includegraphics[width=2.9in]{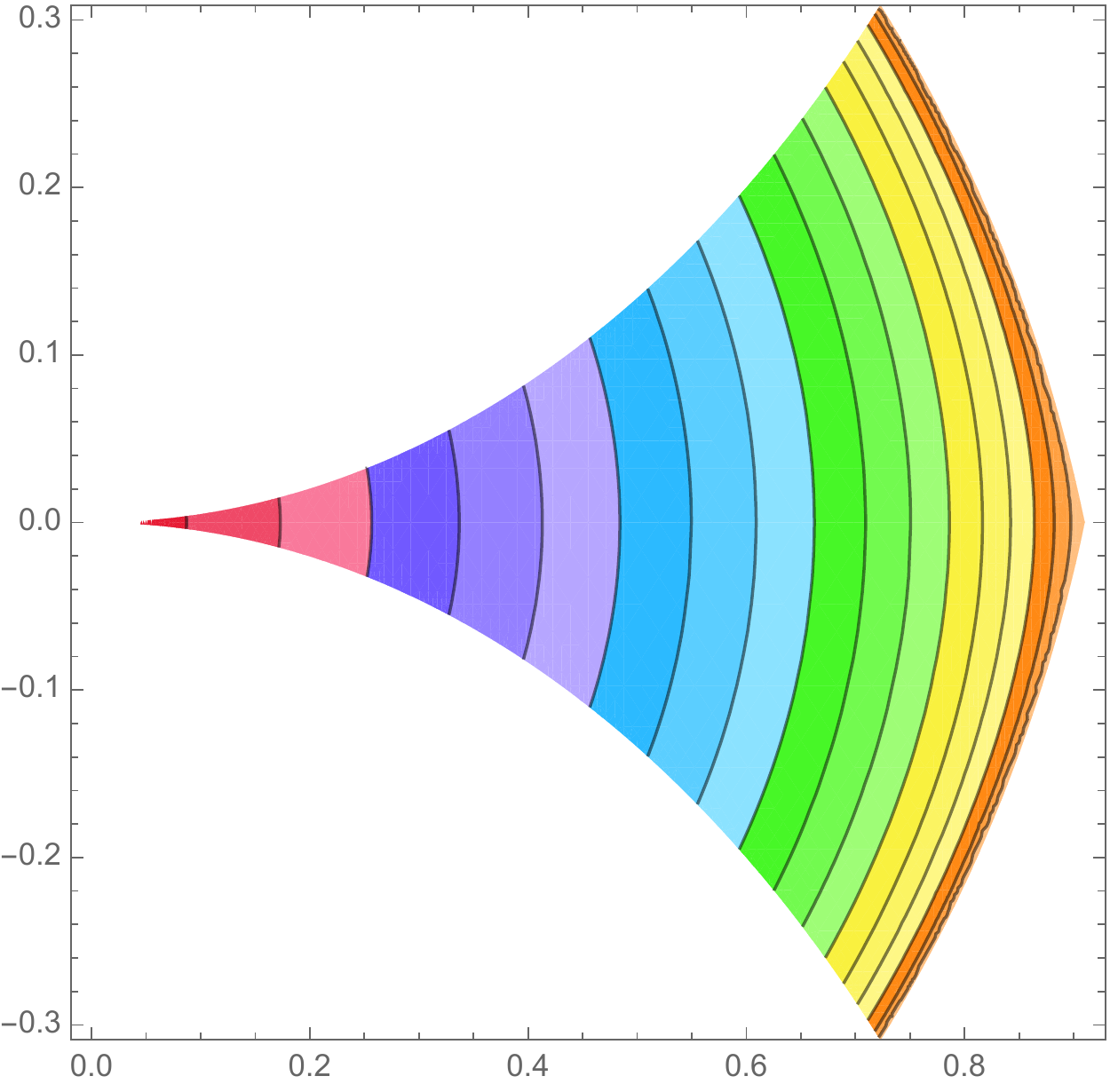}}
  \put(0.95,0){$\delta_T$}
  \put(0.09,0.89){$\delta_N$}
  \end{picture}
}


\caption{\label{nvlcodim2}  \protect\small \sf
Left:  Upper bound in degrees for $\angle (N_p\Sigma, N_q\Sigma)$
as specified by Inequality~(\ref{codim2bound2}), as a function of
$\delta = |pq| / \lfs(p)$ (on the horizontal axis) and
the normal component $\delta_N$ of $\delta$ (on the vertical axis);
i.e., $\delta_N$ is the signed distance from $q$ to
the tangent space $T_p\Sigma$ divided by $\lfs(p)$.
Right:  A similar plot with one change:
the horizontal axis is the tangential component $\delta_T$ of $\delta$;
i.e., the distance from $q$ to the normal space $N_p\Sigma$ divided by
$\lfs(p)$.
Hence, this plot reflects the Euclidean geometry of the space,
with $p$ at the origin, $T_p\Sigma$ on the horizontal midline,
$q$ somewhere in the colored region, and
the smallest possible medial torus (white) blocking $q$ from certain regions.
}
\end{figure}
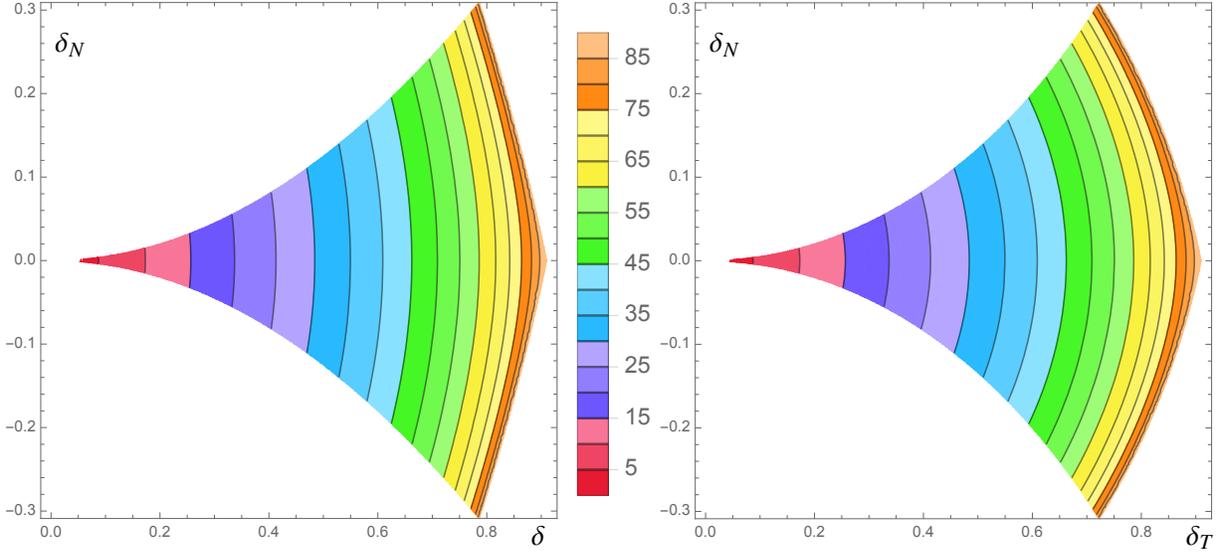

\begin{proof}
%
Let $F$ be the open ball with center $p$ and radius $\lfs(p)$.
$F$ does not intersect the medial axis.
As in the proof of Lemma~\ref{ptnormal},
we choose a coordinate system with $p$ at the origin and
scale the coordinate system so that $\lfs(p) = 1$, so
$F$ is the unit ball centered at the origin.

Let $\dot{\mathbb{B}}$ be the intersection of $p$'s normal space $N_p\Sigma$
with the unit hypersphere $\partial F$ (the boundary of $F$);
$\dot{\mathbb{B}}$ is a unit $(d - k - 1)$-sphere.
For every point $c \in \dot{\mathbb{B}}$,
the open unit ball with center $c$ is tangent to $\Sigma$ at $p$ and
does not intersect $\Sigma$.
Let $\mathbb{B}$ be the union of these (infinitely many) open unit balls
(which constitute all the unit balls tangent to $\Sigma$ at $p$).
The boundary of $\mathbb{B}$ is a torus with inner radius zero (a horn torus).
We call $\mathbb{B}$ itself the (open) {\em solid torus}
and $\dot{\mathbb{B}}$ the {\em torus skeleton}.
Geometrically, $\mathbb{B}$ is
the Minkowski sum of $\dot{\mathbb{B}}$ and an open $d$-ball.
Topologically, $\mathbb{B}$ is the $d$-dimensional product of
a $(d - k - 1)$-sphere and an open $(k + 1)$-ball.
Like the balls it is composed of, $\mathbb{B}$ does not intersect $\Sigma$ nor
contain $q$.

By assumption, $|pq| < \lfs(p)$, so $q \in F$ and
$q$'s normal space $N_q\Sigma$ intersects $\partial F$ in
a $(d - k - 1)$-sphere $S$ (like $\dot{\mathbb{B}}$, but smaller).
Consider an open ball $Z$ with center $z \in S$ such that
$Z$'s boundary passes through $q$.
$Z$ is a subset of a medial ball tangent to $\Sigma$ at $q$, so
$Z$ does not intersect $\Sigma$ nor contain $p$.

The key property for obtaining a bound is that
$Z$ cannot intersect every open unit ball centered on $\dot{\mathbb{B}}$.
If it did, then it would effectively block the hole in the solid torus
$\mathbb{B}$, so that
$\Sigma$ cannot thread through $\mathbb{B}$ at $q$ without somewhere
intersecting $Z$ or $\mathbb{B}$.
This property applies to every ball $Z$ centered on $S$ and just touching $q$.
To obtain a tractable proof, we focus on two particular balls that help
determine the angle $\angle (N_p\Sigma, N_q\Sigma)$.
(Unfortunately, these two balls do not suffice to give a sharp bound, but
we have not been able to derive
better closed-form bounds that take advantage of the other balls.)

We choose a $d$-axis coordinate system with $p$ at the origin such that
the $x_1$-axis lies on $p$'s tangent space $T_p\Sigma$,
the $x_2$-axis lies on $p$'s normal space $N_p\Sigma$, and
$q$ lies in the upper right quadrant of the $x_1$-$x_2$-plane; that is,
$q_1 > 0$, $q_2 \geq 0$, and $q_3 = q_4 = \ldots = q_d = 0$.
Each remaining axis lies in $T_p\Sigma$ or $N_p\Sigma$, so
every axis can be categorized as tangential or normal with respect to $p$.
Let $z_T^2$ be the sum of squares of the tangential components of $z$
except $z_1$, and
let $z_N^2$ be the sum of squares of the normal components of $z$ except $z_2$;
thus $\|z\|^2 = z_1^2 + z_2^2 + z_T^2 + z_N^2$.
(The signs of $z_T$ and $z_N$ are irrelevant.)

By definition,
$\angle (N_p\Sigma, N_q\Sigma) = \max_{\ell_q \subset N_q\Sigma}
\min_{\ell_p \subset N_p\Sigma} \angle (\ell_p, \ell_q)$.
Let $\ell_q \subset N_q\Sigma$ be a line through $q$ that satisfies
$\angle (N_p\Sigma, N_q\Sigma) = \angle (N_p\Sigma, \ell_q)$.
Let $z$ and $z'$ be the two points where $\ell_q$ intersects $\partial F$, and
observe that $z, z' \in S$ (as $S = N_q\Sigma \cap \partial F$).
Let $Z$ and $Z'$ be the open balls centered on $z$ and $z'$, respectively,
with the boundaries of both balls passing through $q$.
Let $\ell = |qz|$ and $\ell' = |qz'|$ be their radii.

As $q_N = 0$, we can determine the angle $\angle (N_p\Sigma, N_q\Sigma)$ from
the identity
\begin{equation}
\cos \angle (N_p\Sigma, N_q\Sigma) = \cos \angle (N_p\Sigma, \ell_q) =
\frac{\sqrt{(z_2 - q_2)^2 + z_N^2}}{\ell},
\label{cosid}
\end{equation}
because the denominator is the length of the line segment $qz$ and
the numerator is the length of the projection of $qz$ onto $N_p\Sigma$.
To find an upper bound on $\angle (N_p\Sigma, N_q\Sigma)$,
we seek a lower bound on the cosine~(\ref{cosid}); to find that,
we will search for legal values of $z_2$, $z_N$, and $\ell$ that minimize
the right-hand side (i.e., a worst-case configuration).
First, we must understand the constraints on these values.

Let $o$ be the point on the torus skeleton $\dot{\mathbb{B}}$ farthest from
$z$.
What is the distance $|zo|$?
First consider the projection $\bar{z}$ of $z$ onto $N_p\Sigma$.
The origin lies between $\bar{z}$ and the farthest point on $\dot{\mathbb{B}}$,
so the distance from $\bar{z}$ to the farthest point is $\|\bar{z}\| + 1$.
With Pythagoras' Theorem we add the tangential component:
\begin{eqnarray*}
|zo|^2 & = & |z\bar{z}|^2 + (\|\bar{z}\| + 1)^2  \\
       & = & z_1^2 + z_T^2 + \left( \sqrt{z_2^2 + z_N^2} + 1 \right)^2  \\
       & = & z_1^2 + z_T^2 + z_2^2 + z_N^2 + 2 \sqrt{z_2^2 + z_N^2} + 1  \\
       & = & 2 + 2 \sqrt{z_2^2 + z_N^2}.
\end{eqnarray*}
The last step follows because $z$ lies on $\partial F$.

As $Z$ has radius $\ell$ and is disjoint from the unit ball centered at $o$,
$\ell + 1 \leq |zo|$. 
We rewrite this constraint as
\begin{equation}
z_2^2 + z_N^2 \geq \left( \frac{(\ell + 1)^2}{2} - 1 \right)^2.
\label{tangencyco1}
\end{equation}
If Inequality~(\ref{tangencyco1}) holds with equality,
we call this event a {\em tangency} between $Z$ and $\mathbb{B}$.
Likewise, the ball $Z'$ entails the following inequality, and
a tangency between $Z'$ and $\mathbb{B}$ means that
it holds with equality.
\begin{equation}
{z'_2}^2 + {z'_N}^2 \geq \left( \frac{(\ell' + 1)^2}{2} - 1 \right)^2.
\label{tangencyco2}
\end{equation}

Recall from the proof of Lemma~\ref{ptnormal} that,
by the Intersecting Chords Theorem, $\ell \ell' = 1 - \|q\|^2$
where $\|q\| = q_1^2 + q_2^2$ is the distance from $p$ to $q$.
As $q \in zz'$, we write two more useful identities:
\begin{eqnarray}
\ell z'_N & = & - \ell' z_N,  \label{proportionsN}  \\
\ell (z'_2 - q_2) & = & \ell' (q_2 - z_2).  \label{proportions2}
\end{eqnarray}
Thus we have a system of three equations and two inequalities in six variables:
$\ell$, $\ell'$, $z_2$, $z'_2$, $z_N$, and $z'_N$.
Among the multiple solutions of this system,
we seek one that minimizes the objective~(\ref{cosid}).

In a configuration where neither tangency is engaged, we can increase
$\angle (N_p\Sigma, N_q\Sigma)$ and decrease its cosine~(\ref{cosid}) by
freely tilting the line segment $zz'$ while maintaining the constraints that
$zz'$ passes through $q$ and $z, z' \in \partial F$.
Therefore, if there is a meaningful bound at all,
an optimal (i.e, worst-case) configuration must engage at least one tangency.
As $Z$ and $Z'$ play symmetric roles, we can assume without loss of generality
that $Z$ is tangent to $\mathbb{B}$ and
Inequality~(\ref{tangencyco1}) holds with equality.
Substituting that identity into~(\ref{cosid}) yields
\begin{equation}
\cos \angle (N_p\Sigma, N_q\Sigma) =
\frac{\sqrt{\left( \frac{(\ell + 1)^2}{2} - 1 \right)^2 + q_2^2 - 2 q_2 z_2}}
     {\ell} =
\sqrt{\left( 1 - \frac{\ell^2 - 1}{2 \ell} \right)^2 +
      \frac{q_2^2 - 2 q_2 z_2}{\ell^2}}.
\label{cosidl}
\end{equation}

As in the proof of Lemma~\ref{ptnormal2}, symmetry will play a role:
the ``optimal'' (i.e., worst-case) solution will turn out to have
$\ell = \ell'$.
To expose this symmetry, we define a parameter
\[
\gamma = \frac{\ell'}{\ell}.  
\]
By Identities~(\ref{proportionsN}) and~(\ref{proportions2}),
we can eliminate the primed variables with the substitutions
$\ell' = \gamma \ell$, $z'_N = - \gamma z_N$ and
$z'_2 = q_2 + \gamma (q_2 - z_2)$.
(A solution with $\gamma = 1$ would imply that
$\ell = \ell'$ and $z'_N = - z_N$.)
Inequality~(\ref{tangencyco2}) becomes
\begin{eqnarray}
\left( q_2 + \gamma (q_2 - z_2) \right)^2 + \gamma^2 z_N^2 \geq
\left( \frac{(\gamma \ell + 1)^2}{2} - 1 \right)^2.
\label{tangencysub}
\end{eqnarray}
To eliminate the variable $z_N$,
we multiply Inequality~(\ref{tangencyco1}) by $\gamma^2$
(recalling that the inequality is now assumed to be an equality) and
subtract Inequality~(\ref{tangencysub}) (which is still an inequality), giving
\begin{eqnarray}
(2 \gamma^2 + 2) q_2 z_2 - (\gamma + 1)^2 q_2^2 \leq \omega
\hspace{.2in}  \mbox{where}  \hspace{.2in}
\omega = \gamma^2 \left( \frac{(\ell + 1)^2}{2} - 1 \right)^2 -
         \left( \frac{(\gamma \ell + 1)^2}{2} - 1 \right)^2.
\end{eqnarray}
Rearranging, we have
\begin{equation}
q_2 z_2 \leq \frac{(\gamma + 1)^2 q_2^2 + \omega}{2 \gamma^2 + 2}.
\end{equation}
Substituting this into~(\ref{cosidl}) gives
\begin{equation}
\cos \angle (N_p\Sigma, N_q\Sigma) \geq
\sqrt{\left( 1 - \frac{\ell^2 - 1}{2 \ell} \right)^2 -
      \frac{2 \gamma q_2^2 + \omega}{(\gamma^2 + 1) \, \ell^2}}.
\label{cosg}
\end{equation}
The right-hand side is a function of $\gamma$, $\ell$, and the point $q$.
However, the definition $\gamma = \ell' / \ell$ and Equation~(\ref{intchords})
together imply that $\ell = \sqrt{(1 - \|q\|)^2 / \gamma}$, so
we can write the right-hand side as a function $f(\gamma, q)$.
We claim that for all valid $q$, $f(\gamma, q)$ is minimized at $\gamma = 1$.
It is straightforward but tedious (and best done with Mathematica) to verify
that $f(\gamma, q) = f(1 / \gamma, q)$ and that
$\frac{\partial}{\partial \gamma} f(\gamma, q)$ is zero at $\gamma = 1$,
positive for $\gamma > 1$, and negative for $\gamma \in (0, 1)$.
Specifically, with the abbreviation $\mathring{q} = 1 - \|q\|^2$, we have
\[
\frac{\partial}{\partial \gamma} f(\gamma, q) =
(\gamma - 1)
\frac{(\gamma + 1)^3 \mathring{q}^2 +
      \left( (2 \gamma^2 + 8\gamma + 2) \mathring{q} + 4\gamma \right)
      \sqrt{\gamma \mathring{q}}}
     {4 \gamma (\gamma + 1)^2 \sqrt{\gamma \mathring{q}}
      \sqrt{\gamma (1 - 4 q_2^2) + 2 \gamma \mathring{q} +
            (1 - \gamma + \gamma^2) \mathring{q}^2 +
            \frac{4 \left( (\gamma^2 + 1) \mathring{q} -
                  2 \gamma \right) \sqrt{\gamma \mathring{q}}}
                 {\gamma + 1}}}.
\]
The numerator and denominator are positive for
$\gamma > 0$, $\mathring{q} > 0$, and $q_2 \in [0, 0.5]$, so
the sign of $\frac{\partial}{\partial \gamma} f$ depends solely on
the sign of $\gamma - 1$, confirming that
the right-hand side of~(\ref{cosg}) is minimized at $\gamma = 1$.

For $\gamma = 1$, we have $\ell = \ell' = \sqrt{1 - \|q\|^2}$ and $\omega = 0$,
so Inequality~(\ref{cosg}) becomes
\begin{equation}
\cos \angle (N_p\Sigma, N_q\Sigma) \geq
\sqrt{\left( 1 - \frac{\|q\|^2}{2 \sqrt{1 - \|q\|^2}} \right)^2 -
      \frac{q_2^2}{1 - \|q\|^2}},
\label{cosq}
\end{equation}
Recall the parameter $\delta = |pq| / \lfs(p)$.
As we chose and scaled our coordinate system so that $p$ is the origin and
$\lfs(p) = 1$, $\|q\| = \delta$.
Inequality~(\ref{codim2bound2}) follows.

Clearly, larger values of $q_2^2$ make the right-hand side smaller
(and the bound weaker).
It is smallest when $q_2$ reaches its maximum allowable value of $\|q\|^2 / 2$.
(This maximum is imposed by the fact that $q \not\in \mathbb{B}$.)
Hence, the following bound holds for all valid values of $q_2$.
\[
\cos \angle (N_p\Sigma, N_q\Sigma) \geq
\sqrt{1 - \frac{\|q\|^2}{\sqrt{1 - \|q\|^2}}},
\]
proving Inequality~(\ref{codim2bound}).
\end{proof}

\section{Extended Triangle Normal Lemmas}
\label{etnls}

The Triangle Normal Lemmas in Section~\ref{tnls} bound
$\angle(N_\tau, N_v\Sigma) = \angle(\aff{\tau}, T_v\Sigma)$
only at a vertex $v$ of $\tau$.
Moreover, for vertices where $\tau$ has a small plane angle, the bound is poor.
Here, we derive a bound on $\angle (N_\tau, N_{\tilde{x}}\Sigma)$
for every $x \in \tau$.
The method to accomplish this is not new:
a triangle normal lemma establishes a strong bound at a vertex where
a triangle has a large plane angle, and a normal variation lemma
extends the bound from that anchor over the rest of the triangle.
We improve on this formulation a bit by taking advantage of the fact that
our Triangle Normal Lemma's bound varies with the plane angle at a vertex:
we choose $\tau$'s vertex nearest $\tilde{x}$ as the anchor if its angle is
at least $49^\circ$; otherwise,
we choose the vertex with the largest plane angle as the anchor.

We begin with several technical lemmas that help us obtain better bounds.
Both lemmas help to constrain where $\tilde{x}$ can lie.

\begin{lemma}
\label{lem:projinsphere}
Let $\Sigma \subset \R^d$ be a smooth $k$-manifold.
Let $\tau$ be a simplex (of any dimension) whose vertices lie on~$\Sigma$.
Let $B_\tau$ be a closed $d$-ball such that $B_\tau \supseteq \tau$
(e.g., $\tau$'s smallest enclosing ball or a circumscribing ball).
Let $r$ be the radius of $B_\tau$, let $v$ be a vertex of $\tau$, and
suppose that $r \leq \lfs(v) / 2$.
Then for every point $x \in \tau$ that is not a vertex of $\tau$,
$\tilde{x} = \nu(x)$ is in the interior of~$B_\tau$.
\end{lemma}

\begin{proof}
Consider a point $x \in \tau$ that is not a vertex of $\tau$.
As $\tau$'s vertices lie in $B_\tau$, $x$ is in the interior of $B_\tau$.
If $\tilde{x} = x$ the lemma follows immediately, so
suppose that $\tilde{x} \neq x$ and thus $x \not\in \Sigma$.
Let $B$ be the open medial ball tangent to $\Sigma$ at $\tilde{x}$ such that
$x$ lies on the line segment $\tilde{x}m$, where $m$ is the center of $B$,
as illustrated in Figure~\ref{projinside}.
As $B$ is a medial ball, $m$ lies on the medial axis of $\Sigma$.

\begin{figure}
\centerline{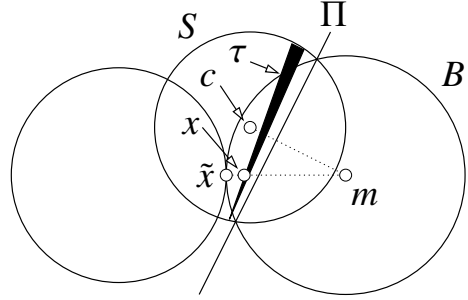}

\caption{\label{projinside}
For every point $x \in \tau$ except $\tau$'s vertices,
$\tilde{x}$ is in the interior of $B_\tau$.
}
\end{figure}

Recall that $B$ is open and $B_\tau$ is closed.
If the entire closure of $B$ is in the interior of $B_\tau$, then
$\tilde{x}$ is in the interior of $B_\tau$ and the lemma follows immediately;
so assume it is not.
Let $C$ be the intersection of the boundaries of $B$ and $B_\tau$.
$C$ cannot be the boundary of $B$, because
we have just assumed that $B_\tau$ does not include the closure of $B$.
We show that $C \neq \emptyset$ by ruling out the alternatives:
we cannot have $B$ and $B_\tau$ disjoint because
$x \in B$ and $x$ is in the interior of $B_\tau$;
we cannot have $B_\tau \subset B$, as $\tau$'s vertices are not in $B$; and
we have already ruled out $\mathrm{closure}(B) \subset B_\tau$.
Hence $C$ is either a $(d - 2)$-sphere (e.g., a circle in $\R^3$) or
a single point (with $B$ and $B_\tau$ tangent to each other at that point,
one inside the other).

If $C$ is a $(d - 2)$-sphere, let $\Pi$ be
the unique hyperplane that includes that $(d - 2)$-sphere, as illustrated;
if $C$ contains a single point,
let $\Pi$ be the hyperplane tangent to $B_\tau$ and $B$ at that point.
Let $\bar{\Pi}_\tau$ be the closed halfspace bounded by $\Pi$ that includes
$B_\tau \setminus B$, and
let $\Pi_\tau$ be the open version of the same halfspace.
The portion of $B$ in $\bar{\Pi}_\tau$ is in the interior of $B_\tau$, and
the portion of $B$'s boundary in $\Pi_\tau$ is in the interior of $B_\tau$.
The portion of $B_\tau$ in the open halfspace complementary to $\bar{\Pi}_\tau$
is a subset of $B$.
Every vertex of $\tau$ lies in $B_\tau$ but not in $B$, hence
$\tau$'s vertices lie in $\bar{\Pi}_\tau$.
Therefore, $\tau \subset \bar{\Pi}_\tau$ and $x \in \bar{\Pi}_\tau$.

By assumption, the radius of $B_\tau$ satisfies $r \leq \lfs(v) / 2$, so
$|vm| \geq \lfs(v) \geq 2r$.
As $v$ lies in $B_\tau$ and $|vm|$ is at least twice the radius of $B_\tau$,
it follows that $m$ is not in the interior of $B_\tau$.
But $m \in B$, so $m \not\in \bar{\Pi}_\tau$.

Given the facts that $x$ lies on the line segment $m\tilde{x}$,
$m \not\in \bar{\Pi}_\tau$, $x \in \bar{\Pi}_\tau$, and
$\tilde{x} \neq x$, it follows that $\tilde{x} \in \Pi_\tau$.
As $\tilde{x}$ is also on $B$'s boundary,
$\tilde{x}$ is in the interior of $B_\tau$.
\end{proof}

Lemma~\ref{lem:projinsphere} implies that $\tilde{x}$ is in
{\em every} ball $B_\tau \supseteq \tau$ with radius $\lfs(v) / 2$ (or less).
The intersection of these balls, illustrated in Figure~\ref{lens}, is
typically a narrow region, especially if $\tau$ is small.
The next lemma also places a restriction on the position of $\tilde{x}$.

\begin{figure}
\centerline{\includegraphics[width=3.5in]{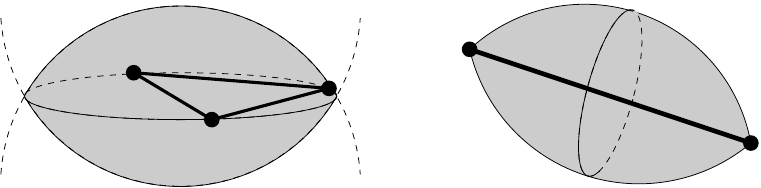}}

\caption{\label{lens}  \protect\small \sf
For the triangle $\tau$ at left, the dark lens-shaped region is
the intersection of $\tau$'s two enclosing balls of radius $\lfs(v) / 2$, where
$v$ is any vertex of $\tau$.
For every point $x \in \tau$, $\tilde{x}$ lies in this lens.
Likewise, for the segment at right,
the lemon-shaped region is the intersection of
its infinitely many enclosing balls of radius $\lfs(v) / 2$;
this lemon contains $\tilde{x}$ for every point $x$ on the segment.
}
\end{figure}

\begin{lemma}
\label{proj2nearest}
Let $\Sigma \subset \R^d$ be a smooth $k$-manifold.
Let $\tau$ be a simplex (of any dimension) whose vertices lie on~$\Sigma$.
Let $r$ be the min-containment radius of $\tau$
(i.e., the radius of $\tau$'s smallest enclosing ball).
Then for every point $x \in \tau$,
the distance from $\tilde{x}$ to the nearest vertex of $\tau$ is
at most $\sqrt{2} r$.
Moreover, if $r < \ebs(\tilde{x})$,
the distance from $\tilde{x}$ to the nearest vertex of $\tau$ is at most
\begin{equation}
\sqrt{2 \, \ebs(\tilde{x})
      \left( \ebs(\tilde{x}) - \sqrt{\ebs(\tilde{x})^2 - r^2} \right)}
\in [r, \sqrt{2} r).
\label{proj2nearbound}
\end{equation}
\end{lemma}

\begin{proof}
Let $y \in \tau$ be the point nearest $\tilde{x}$ on $\tau$.
As $x$ is also on $\tau$, $|y\tilde{x}| \leq |x\tilde{x}|$.
Let $\sigma$ be the unique face of $\tau$
(i.e., a vertex, edge, triangle, etc.)\ whose relative interior contains $y$.
Observe that the line segment $y\tilde{x}$ is orthogonal to $\sigma$,
as Figure~\ref{proj2near} illustrates.
(If $\sigma$ is a vertex, it is a trivial ``orthogonality.'')
Let $w$ be the vertex of $\sigma$ nearest $y$;
$y\tilde{x}$ is orthogonal to $yw$.
By Pythagoras' Theorem,
$|w\tilde{x}|^2 = |yw|^2 + |y\tilde{x}|^2 \leq |yw|^2 + |x\tilde{x}|^2$.

\begin{figure}
\centerline{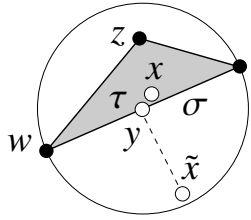}

\caption{\label{proj2near}
Given a simplex $\tau$ with min-containment radius $r$ and
a point $x \in \tau$, the distance from $\tilde{x}$ to
the nearest vertex of $\tau$ is at most $\sqrt{2} r$.
}
\end{figure}

As $\tau$'s smallest enclosing ball has radius $r$, $|yw| \leq r$.
Likewise, let $z$ be the vertex of $\tau$ nearest $x$; then $|xz| \leq r$.
As $z$ lies on $\Sigma$ and $\tilde{x}$ is the point nearest $x$ on $\Sigma$,
$|x\tilde{x}| \leq |xz| \leq r$.
Hence $|w\tilde{x}| \leq \sqrt{r^2 + r^2} = \sqrt{2} r$, and
the distance from $\tilde{x}$ to the nearest vertex of $\tau$
(which may or may not be $w$) is at most $\sqrt{2} r$ as claimed.

Alternatively, if $r < \ebs(\tilde{x})$,
we can substitute the bound for $|x\tilde{x}|$ from
the Surface Interpolation Lemma (Lemma~\ref{interplemma}),
yielding the bound
$\sqrt{r^2 +
 \left( \ebs(\tilde{x}) - \sqrt{\ebs(\tilde{x})^2 - r^2} \right)^2}$,
which is equal to~(\ref{proj2nearbound}).
\end{proof}

This brings us to the first main result of this section.

\begin{lemma}[Extended Triangle Normal Lemma]
\label{etnlsimple}
Let $\Sigma$ be a bounded $k$-manifold without boundary in $\R^d$ with
$k \geq 2$.
Let $\tau = \triangle vv'v''$ be a triangle whose vertices lie on $\Sigma$.
Let $R$ be $\tau$'s circumradius.
Suppose that $R \leq \kappa \, \lfs(v)$, $R \leq \kappa \, \lfs(v')$, and
$R \leq \kappa \, \lfs(v'')$ for some $\kappa \leq 1 / 2$.
Let $x$ be any point on $\tau$, and
let $\tilde{x}$ be the point nearest $x$ on $\Sigma$.
Then for any angle $\phi \in (0^\circ, 60^\circ]$,
\begin{equation}
\angle (N_\tau, N_{\tilde{x}}\Sigma) \leq
\max \left\{ \eta(\sqrt{2} \kappa) +
\arcsin \left( \kappa \cot \frac{\phi}{2} \right),
\eta(2 \kappa) +
\arcsin \left( \kappa \cot \left( 45^\circ - \frac{\phi}{4} \right) \right)
\right\},
\label{etnlsimpbound}
\end{equation}
where $\eta(\delta) = \eta_1(\delta)$ as defined in Lemma~\ref{lem:tnl}
if $d - k = 1$, or $\eta(\delta) = \eta_2(\delta)$ as defined in
Lemma~\ref{lem:tnlhigh} if $d - k \geq 2$.
\end{lemma}

Lemma~\ref{etnlsimple} is unusual because it has a parameter $\phi$;
the right-hand side of Inequality~(\ref{etnlsimpbound}) varies a bit
with~$\phi$.
The parameter $\phi$ is a threshold that determines
which vertex of $\tau$ is used as an anchor.
In codimension~$1$, a good choice of $\phi$ is $49^\circ$, because
it balances the two expressions in~(\ref{etnlsimpbound}) reasonably well
and delivers a bound below $90^\circ$ over the range $\kappa \in [0, 0.3734]$.
For a specific value of $\kappa$, one can tune $\phi$ to obtain
a slightly better bound, but the improvement is marginal.
In codimension $2$ or greater, the bound~(\ref{etnlsimpbound}) is weaker
because $\eta_2$ is weaker than $\eta_1$.
A good choice is $\phi = 48.5^\circ$, which 
delivers a bound below $90^\circ$ over the range $\kappa \in [0, 0.3527]$.
Figure~\ref{etnlplot} graphs the bound~(\ref{etnlsimpbound}) both
for codimension $1$ and for higher codimensions.

\begin{proof}
Suppose without loss of generality that
$v$ is the vertex of $\tau$ nearest $\tilde{x}$.
Let $w \in \{ v, v', v'' \}$ be the vertex at $\tau$'s largest plane angle.
Let $B_\tau$ be $\tau$'s smallest enclosing ball and
observe that its radius is $r \leq R \leq \lfs(v) / 2$.
By Lemma~\ref{lem:projinsphere}, $\tilde{x} \in B_\tau$, so
$|w\tilde{x}| \leq 2r \leq 2 \kappa \, \lfs(w)$.
By Lemma~\ref{proj2nearest},
$|v\tilde{x}| \leq \sqrt{2} r \leq \sqrt{2} \kappa \, \lfs(v)$.
By the Normal Variation Lemma,
$\angle (N_w\Sigma, N_{\tilde{x}}\Sigma) \leq \eta(2 \kappa)$ and
$\angle (N_v\Sigma, N_{\tilde{x}}\Sigma) \leq \eta(\sqrt{2} \kappa)$.

If $\tau$'s plane angle at the vertex $v$ is $\phi$ or greater, then
by the Triangle Normal Lemma (Lemma~\ref{lem:tnl} or~\ref{lem:tnlhigh}),
$\sin \angle (N_\tau, N_v\Sigma) \leq \frac{R}{\lfs(v)} \cot \frac{\phi}{2}
\leq \kappa \cot \frac{\phi}{2}$.
Then $\angle (N_\tau, N_{\tilde{x}}\Sigma) \leq
\angle (N_{\tilde{x}}\Sigma, N_v\Sigma) + \angle (N_\tau, N_v\Sigma) \leq
\eta(\sqrt{2} \kappa) + \arcsin (\kappa \cot \frac{\phi}{2})$.

Otherwise, $\tau$'s plane angle at $v$ is less than $\phi$, so
$\tau$'s plane angle at $w$ ($\tau$'s largest plane angle) is
greater than $(180^\circ - \phi) / 2$.
By the Triangle Normal Lemma, $\sin \angle (N_\tau, N_w\Sigma) \leq
\frac{R}{\lfs(w)} \cot (45^\circ - \phi / 4) \leq
\kappa \cot (45^\circ - \phi / 4)$.
Then $\angle (N_\tau, N_{\tilde{x}}\Sigma) \leq
\angle (N_{\tilde{x}}\Sigma, N_w\Sigma) + \angle (N_\tau, N_w\Sigma) \leq
\eta(2 \kappa) + \arcsin (\kappa \cot (45^\circ - \phi / 4))$.
\end{proof}

For our final act, we address the approximation accuracy of
restricted Delaunay triangulations of $\epsilon$-samples.
{\em Restricted Delaunay triangulations} (RDTs),
proposed by Edelsbrunner and Shah~\cite{edelsbrunner97}, have become
a well-established way of generating Delaunay-like triangulations
on curved surfaces~\cite{cheng12,dey07,edelsbrunner01}.
Given a $k$-manifold $\Sigma \subset \R^d$ and
a finite set of vertices $V \subset \Sigma$,
let $\del V$ be the ($d$-dimensional) Delaunay triangulation of $V$ and
let $\vor V$ be the Voronoi diagram of $V$.
Every $j$-simplex in $\del V$ is dual to some $(d - j)$-face of $\vor V$.
The restricted Delaunay triangulation $\resdel{V}{\Sigma}$ is
a subcomplex of $\del V$ consisting of the {\em restricted Delaunay simplices}:
the simplices whose Voronoi dual faces intersect $\Sigma$.

Here, we are specifically interested in
the restricted Delaunay triangles when $k \geq 2$.
Recall that for a triangle $\tau = \triangle vv'v''$,
$N_\tau$ is the set of all points in $\R^d$ that are equidistant from
$v$, $v'$, and $v''$, a flat of dimension $d - 2$ that
is orthogonal to $\tau$ and passes through $\tau$'s circumcenter.
Let $\tau^* \in \vor V$ denote
the Voronoi $(d - 2)$-face dual to some $\tau \in \del V$.
By definition, $\tau$ is a restricted Delaunay triangle if
there exists a point $u \in \tau^* \cap \Sigma$.
(There might be more than one such point.)
We call $u$ a {\em restricted Voronoi vertex} dual to~$\tau$.

A finite point set $V \subset \Sigma$ is called
an {\em $\epsilon$-sample of $\Sigma$} if
for every point $p \in \Sigma$, there is a vertex $w \in V$ such that
$|pw| \leq \epsilon \, \lfs(p)$.
That is, the ball centered at $p$ with radius $\epsilon \, \lfs(p)$ contains
at least one sample point.
One of the crowning results of provably good surface reconstruction is that
for a sufficiently small~$\epsilon$,
the restricted Delaunay triangulation $\resdel{V}{\Sigma}$ of
an $\epsilon$-sample $V$ of $\Sigma$
is homeomorphic to $\Sigma$~\cite{amenta99b,amenta02,dey07}.
(In a forthcoming sequel paper, we will use this paper's results and
other new ideas to improve the constant $\epsilon$ in that theorem.)
For small~$\epsilon$, $\resdel{V}{\Sigma}$ is also
a geometrically accurate approximation of $\Sigma$,
as we demonstrate below in Corollary~\ref{epsgood}.

Although one could apply Lemma~\ref{etnlsimple} to
restricted Delaunay triangles,
we will obtain a stronger (but less general) extended triangle normal lemma
by taking advantage of the fact that for each restricted Delaunay triangle,
a dual point $u$ lies on $\Sigma$.
Prior to that, we need a couple of short technical lemmas.

\begin{lemma}
\label{lem:vertexnearsurf}
Let $\Sigma \subset \R^d$ be a point set with a well-defined medial axis $M$.
Let $u \in \Sigma$ be a restricted Voronoi vertex and
let $\tau$ be its dual restricted Delaunay simplex.
Let $x \in \tau$ be a point that does not lie on $M$.
Let $\tilde{x}$ be the point on $\Sigma$ nearest $x$.
There is a vertex $v$ of $\tau$ such that $|v\tilde{x}| \leq |vu|$, and
such that $|v\tilde{x}| < |vu|$ if $\tilde{x} \neq u$.
\end{lemma}

\begin{proof}
If $\tilde{x} = u$ then the result follows immediately, so
assume that $\tilde{x} \neq u$.
As $x$ does not lie on the medial axis,
$\tilde{x}$ is the unique point on $\Sigma$ nearest $x$.
As $u$ also lies on $\Sigma$, $|x\tilde{x}| < |xu|$.
Let $\Pi$ be the hyperplane that bisects the line segment $\tilde{x}u$, and
observe that $x$ lies on the same side of $\Pi$ as $\tilde{x}$.
As $x \in \tau$ and $\tau$ is a simplex,
some vertex $v$ of $\tau$ lies on the same side of $\Pi$ as $\tilde{x}$, thus
$|v\tilde{x}| < |vu|$.
\end{proof}

The following simple lemma is implicit in Amenta and Bern~\cite{amenta99b} and
explicit in Amenta, Choi, Dey, and Leekha~\cite{amenta02}.

\begin{lemma}[Feature Translation Lemma]
\label{lem:ftl}
Let $\Sigma \subset \R^3$ be a smooth surface and
let $p, q \in \Sigma$ be points on $\Sigma$ such that
$|pq| \leq \epsilon \, \lfs(p)$ for some $\epsilon < 1$.
Then
\[
\lfs(p) \leq \frac{1}{1 - \epsilon} \lfs(q)
\hspace{.2in}  \mbox{and}  \hspace{.2in}
|pq| \leq \frac{\epsilon}{1 - \epsilon} \lfs(q).
\]
\end{lemma}

\begin{proof}
By the definition of the local feature size,
there is a medial axis point $m$ such that $|qm| = \lfs(q)$.
By the Triangle Inequality,
$\lfs(p) \leq |pm| \leq |pq| + |qm| \leq \epsilon \, \lfs(p) + \lfs(q)$.
Rearranging terms gives $\lfs(p) \leq \lfs(q) / (1 - \epsilon)$.
The second claim follows immediately.
\end{proof}

\begin{lemma}[Extended Triangle Normal Lemma for $\epsilon$-samples]
\label{etnleps}
Let $\Sigma$ be a bounded $k$-manifold without boundary in $\R^d$ with
$k \geq 2$.
Let $\tau = \triangle vv'v''$ be a triangle whose vertices lie on $\Sigma$.
Let $u$ be a point in $\Sigma \cap N_\tau$.
Let $s = |vu| = |v'u| = |v''u|$ and
suppose that $s \leq \epsilon \, \lfs(u)$ for some $\epsilon \leq 1 / 3$.
Let $x$ be any point on $\tau$, and
let $\tilde{x}$ be the point nearest $x$ on $\Sigma$.
Then for any angle $\phi \in (0^\circ, 60^\circ]$,
\begin{equation}
\angle (N_\tau, N_{\tilde{x}}\Sigma) \leq
\max \left\{
\eta \left( \frac{\epsilon}{1 - \epsilon} \right) +
\arcsin \left( \frac{\epsilon}{1 - \epsilon} \cot \frac{\phi}{2} \right),
2 \eta(\epsilon) + \arcsin \left(
\frac{\epsilon}{1 - \epsilon} \cot \left( 45^\circ - \frac{\phi}{4} \right)
\right)
\right\},
\label{etnlepsineq}
\end{equation}
where $\eta(\delta) = \eta_1(\delta)$ as defined in Lemma~\ref{lem:tnl}
if $d - k = 1$, or $\eta(\delta) = \eta_2(\delta)$ as defined in
Lemma~\ref{lem:tnlhigh} if $d - k \geq 2$.
\end{lemma}

A good choice of $\phi$ in codimension $1$ is $56.65^\circ$, which delivers
a bound below $90^\circ$ for all $\epsilon \leq 0.3202$.
A good choice of $\phi$ in higher codimensions is $56.75^\circ$, which delivers
a bound below $90^\circ$ for all $\epsilon \leq 0.3189$.
Figure~\ref{etnlepsplot} graphs the bound for both cases.

\begin{figure}
\centerline{
  \setlength{\unitlength}{4in}
  \begin{picture}(1,0.653)
  \put(0,0){\includegraphics[width=4in]{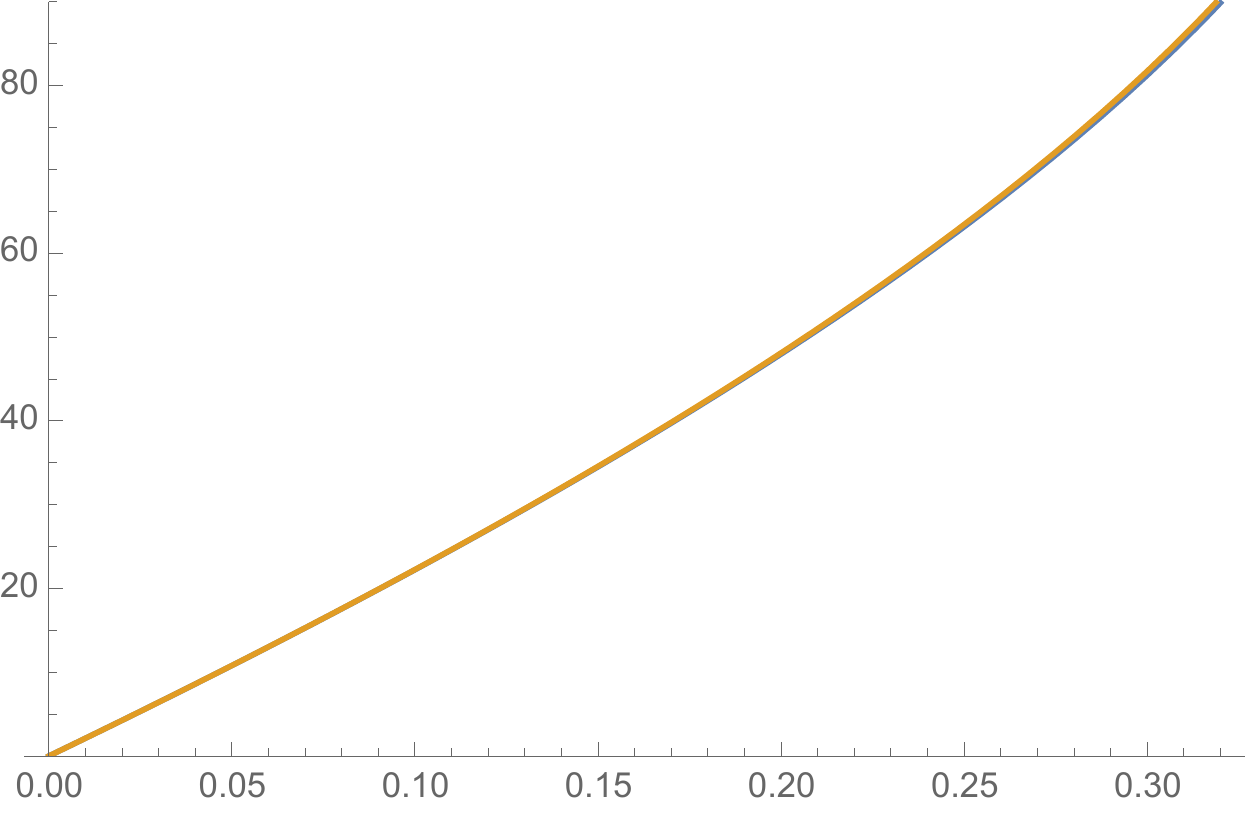}}
  \put(0.98,0.01){$\epsilon$}
  \put(0.07,0.62){bound on $\angle(N_\tau, N_{\tilde{x}}\Sigma)$}
  \put(0.56,0.3){codimension $1$}
  \put(0.24,0.35){higher codimensions}
  \end{picture}
}

%

\caption{\label{etnlepsplot}  \protect\small \sf
Upper bounds for $\angle(N_\tau, N_{\tilde{x}}\Sigma) =
\angle(\aff{\tau}, T_{\tilde{x}}\Sigma)$
as a function of the parameter $\epsilon$ of an $\epsilon$-sample of $\Sigma$,
where $\tau$ is a restricted Delaunay triangle whose vertices are
in the $\epsilon$-sample and $x$ is any point on $\tau$.
The blue curve is the upper bound in codimension~$1$
(with the choice $\phi = 56.65^\circ$) and
the brown curve is the upper bound in higher codimensions
(with the choice $\phi = 56.75^\circ$),
for which the Normal Variation Lemma is weaker.
The two curves are barely distinguishable because
$\eta_1(\epsilon)$ does not differ much from $\eta_2(\epsilon)$ for
$\epsilon < 0.32$.
}
\end{figure}
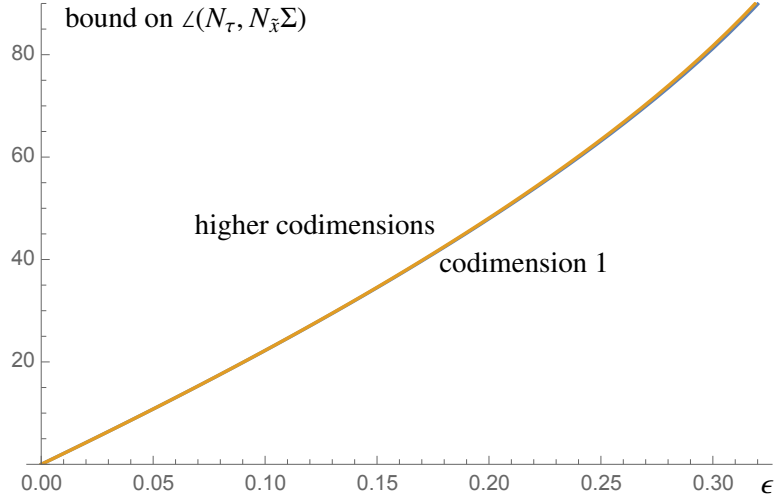

\begin{proof}
Let $R$ be $\tau$'s circumradius.
Let $B_\tau$ be the closed $d$-ball with center $u$ and radius $s$,
whose boundary passes through all three vertices of $\tau$.
As $\tau$'s circumcircle is a cross section of the boundary of $B_\tau$,
$R \leq s$.

Suppose without loss of generality that
$v$ is the vertex of $\tau$ nearest $\tilde{x}$.
Let $w \in \{ v, v', v'' \}$ be the vertex at $\tau$'s largest plane angle.
As $|vu| = |wu| = s \leq \epsilon \, \lfs(u)$,
by the Feature Translation Lemma (Lemma~\ref{lem:ftl}),
$\lfs(u) \leq \lfs(v) / (1 - \epsilon)$ and likewise
$\lfs(u) \leq \lfs(w) / (1 - \epsilon)$, so
$R \leq s \leq \frac{\epsilon}{1 - \epsilon} \, \lfs(v) \leq \lfs(v) / 2$
and likewise $R \leq \frac{\epsilon}{1 - \epsilon} \, \lfs(w)$.
By Lemma~\ref{lem:projinsphere} (with $B_\tau$ as defined above),
$|\tilde{x}u| \leq s$; hence $|\tilde{x}u| \leq \epsilon \, \lfs(u)$.
By Lemma~\ref{lem:vertexnearsurf}, $|\tilde{x}v| \leq s$;
hence $|\tilde{x}v| \leq \frac{\epsilon}{1 - \epsilon} \, \lfs(v)$.
By the Normal Variation Lemma,
$\angle (N_v\Sigma, N_u\Sigma) \leq \eta(\epsilon)$,
$\angle (N_w\Sigma, N_u\Sigma) \leq \eta(\epsilon)$,
$\angle (N_{\tilde{x}}\Sigma, N_u\Sigma) \leq \eta(\epsilon)$, and
$\angle (N_{\tilde{x}}\Sigma, N_v\Sigma) \leq
\eta \left( \frac{\epsilon}{1 - \epsilon} \right)$.

If $\tau$'s plane angle at the vertex $v$ is $\phi$ or greater, then
by the Triangle Normal Lemma (Lemma~\ref{lem:tnl} or~\ref{lem:tnlhigh}),
$\sin \angle (N_\tau, N_v\Sigma) \leq \frac{R}{\lfs(v)} \cot \frac{\phi}{2}
\leq \frac{\epsilon}{1 - \epsilon} \cot \frac{\phi}{2}$.
Then $\angle (N_\tau, N_{\tilde{x}}\Sigma) \leq
\angle (N_{\tilde{x}}\Sigma, N_v\Sigma) + \angle (N_\tau, N_v\Sigma) \leq
\eta \left( \frac{\epsilon}{1 - \epsilon} \right) +
\arcsin \left( \frac{\epsilon}{1 - \epsilon} \cot \frac{\phi}{2} \right)$.

Otherwise, $\tau$'s plane angle at $v$ is less than $\phi$, so
$\tau$'s plane angle at $w$ ($\tau$'s largest plane angle) is
greater than $(180^\circ - \phi) / 2$.
By the Triangle Normal Lemma, $\sin \angle (N_\tau, N_w\Sigma) \leq
\frac{R}{\lfs(w)} \cot (45^\circ - \phi / 4) \leq
\frac{\epsilon}{1 - \epsilon} \cot (45^\circ - \phi / 4)$.
Then $\angle (N_\tau, N_{\tilde{x}}\Sigma) \leq
\angle (N_{\tilde{x}}\Sigma, N_u\Sigma) + \angle (N_w\Sigma, N_u\Sigma) +
\angle (N_\tau, N_w\Sigma) \leq 2 \eta(\epsilon) + \arcsin
\left( \frac{\epsilon}{1 - \epsilon} \cot (45^\circ - \phi / 4) \right)$.
\end{proof}

Our final corollary summarizes the interpolation and normal errors for
restricted Delaunay triangulations of $\epsilon$-samples of manifolds.

\begin{corollary}
\label{epsgood}
Let $\Sigma$ be a bounded $k$-manifold without boundary in $\R^d$
with $k \geq 2$.
Let $V$ be an $\epsilon$-sample of $\Sigma$ for some $\epsilon < 1 / 2$.
Then for every restricted Delaunay triangle $\tau \in \resdel{V}{\Sigma}$ and
every point $x \in \tau$,
\begin{eqnarray*}
|x\tilde{x}| & \leq & \left( 1 -
             \sqrt{1 - \left( \frac{\epsilon}{1 - \epsilon} \right)^2} \right)
             \, \lfs(\tilde{x})
\hspace{.2in}  \mbox{and}  \hspace{.2in}  \\
|x\tilde{x}| & \leq & (1 - \epsilon - \sqrt{1 - 2 \epsilon}) \, \lfs(u)
\end{eqnarray*}
where $u$ is any restricted Voronoi vertex dual to $\tau$.
Moreover, if $\epsilon \leq 1 / 3$, then
$\angle (N_\tau, N_{\tilde{x}}\Sigma)$ satisfies~(\ref{etnlepsineq})
for any $\phi \in (0^\circ, 60^\circ]$.
\end{corollary}

\begin{proof}
Consider some $\tau = \triangle vv'v'' \in \resdel{V}{\Sigma}$.
By the definition of ``restricted Delaunay triangle,''
there is a point $u \in \Sigma \cap N_\tau$ that lies on the boundaries of
the Voronoi cells of all three vertices $v$, $v'$, and $v''$;
$u$ is (by definition) a restricted Voronoi vertex dual to $\tau$.
Let $s = |vu| = |v'u| = |v''u|$; the Voronoi property implies that
there is no vertex $w \in V$ such that $|wu| < s$.
As $V$ is an $\epsilon$-sample, $s \leq \epsilon \, \lfs(u)$.
The bound~(\ref{etnlepsineq}) follows by Lemma~\ref{etnleps}.

The ball $B_\tau$ with center $u$ and radius $s$ encloses $\tau$, so
the minimum enclosing ball of $\tau$ has radius
$r \leq s \leq \epsilon \, \lfs(u)$.
By Lemma~\ref{lem:projinsphere}, $\tilde{x} \in B_\tau$, thus
$|\tilde{x}u| \leq s \leq \epsilon \, \lfs(u)$.
By the Feature Translation Lemma (Lemma~\ref{lem:ftl}),
$\lfs(u) \leq \lfs(\tilde{x}) / (1 - \epsilon)$, and hence
$r \leq \epsilon \, \lfs(u) \leq \epsilon \, \lfs(\tilde{x}) / (1 - \epsilon)
< \lfs(\tilde{x})$ so we can apply Lemma~\ref{interplemma}, giving
\begin{eqnarray*}
|x\tilde{x}| & \leq & \lfs(\tilde{x}) - \sqrt{\lfs(\tilde{x})^2 - r^2}  \\
& \leq & \lfs(\tilde{x}) - \sqrt{\lfs(\tilde{x})^2 -
         \epsilon^2 \, \lfs(\tilde{x})^2 / (1 - \epsilon)^2}  \\
&    = & \left( 1 -
         \sqrt{1 - \left( \frac{\epsilon}{1 - \epsilon} \right)^2} \right)
         \, \lfs(\tilde{x})
\end{eqnarray*}
and as $L - \sqrt{L^2 - r^2}$ increases as $L$ decreases,
\begin{eqnarray*}
|x\tilde{x}|
& \leq & \lfs(\tilde{x}) - \sqrt{\lfs(\tilde{x})^2 - r^2}  \\
& \leq & \left( (1 - \epsilon) \, \lfs(u) -
         \sqrt{(1 - \epsilon)^2 \, \lfs(u)^2 - \epsilon^2 \, \lfs(u)^2} \right)
         \\
&    = & (1 - \epsilon - \sqrt{1 - 2 \epsilon}) \, \lfs(u)
\end{eqnarray*}
as claimed.
\end{proof}

For example, in a $0.2$-sample, we have
$|x\tilde{x}| \leq 0.0318 \, \lfs(\tilde{x})$,
$|x\tilde{x}| \leq 0.0255 \, \lfs(u)$, and
$\angle (N_\tau, N_{\tilde{x}}\Sigma) < 47.95^\circ$
for every point $x$ on every restricted Delaunay triangle.
Note that the normal errors can still be rather large when
the interpolation errors are reasonably small.

\newpage

\bibliographystyle{jrs}
\bibliography{jrs}

\end{document}